\pdfoutput=1

\newif\ifjonathan\jonathanfalse
\newif\ifanon\anonfalse
\newif\ifdraft\draftfalse
\newif\iflong\longfalse
\newif\ifpagelimits\pagelimitsfalse
\newif\ifcamera\cameratrue

\IfFileExists{local.tex}
  {\draftfalse\anonfalse\longtrue\camerafalse
}
  {}

\ifanon
\documentclass[acmsmall,review,anonymous]{acmart}
\else
\documentclass[format=acmsmall, review=false, screen=true]{acmart}
\fi

\setcopyright{rightsretained}
\acmJournal{PACMPL}
\acmYear{2017}
\acmVolume{1}
\acmNumber{ICFP}
\acmArticle{17}
\acmMonth{9}
\acmDOI{10.1145/3110261}
\acmPrice{}

\ifcamera\else
\makeatletter
\@ACM@printacmreffalse
\@ACM@printccsfalse
\def\@copyrightpermission{This work is licensed under a \href{https://creativecommons.org/licenses/by/4.0/}{Creative Commons Attribution 4.0 International License}}
\fancypagestyle{firstpagestyle}{%
  \fancyhf{}%
  \renewcommand{\headrulewidth}{\z@}%
  \renewcommand{\footrulewidth}{\z@}%
    \fancyhead[LE]{\ACM@linecountL\@headfootfont\thepage}%
    \fancyhead[RO]{\@headfootfont\thepage}%
    \fancyhead[RE]{\@headfootfont\@shortauthors}%
    \fancyhead[LO]{}%
    \fancyfoot[RO,LE]{}%
}
\fancypagestyle{standardpagestyle}{%
  \fancyhf{}%
  \renewcommand{\headrulewidth}{\z@}%
  \renewcommand{\footrulewidth}{\z@}%
    \fancyhead[LE]{\ACM@linecountL\@headfootfont\thepage}%
    \fancyhead[RO]{\@headfootfont\thepage}%
    \fancyhead[RE]{\@headfootfont\@shortauthors}%
    \fancyhead[LO]{\ACM@linecountL\@headfootfont\shorttitle}%
    \fancyfoot[RO,LE]{}%
}
\def\@mkbibcitation{}
\@ACM@manuscripttrue 

\def\@titlefont{\sffamily\LARGE\bfseries}

\makeatother
\fi

\usepackage[utf8]{inputenc}
\usepackage[T1]{fontenc}
\usepackage{amsmath}
\usepackage{amssymb}
\usepackage{stmaryrd}
\usepackage{amstext}
\usepackage{MnSymbol}
\usepackage{amsthm}
\usepackage{hyperref}
\usepackage{balance}
\usepackage{upquote} 
\usepackage{verbatim}

\usepackage{mathpartir}
\usepackage{tabularx}
  \newcolumntype{R}{>{\raggedleft\arraybackslash}X}%
  \newcolumntype{Y}{>{\centering\arraybackslash}X}%

\usepackage{xparse}
\usepackage{enumerate}
\usepackage{microtype}

\usepackage{tikz}
  \usetikzlibrary{calc}
  \usetikzlibrary{positioning}
  \usetikzlibrary{tikzmark}
  \usetikzlibrary{shapes.geometric}
  \usetikzlibrary{decorations.text}
  \usetikzlibrary{backgrounds}  

\usepackage{listings}
  
\lstdefinelanguage{fstar}{%
  morekeywords=[1]{type,and,val,fun,let,in,ref,try,if,then,else,match,with,open,as,module,rec,end,assume,private,when,forall,ghost,assert,function,logic,array,pattern,effect,requires,ensures,decreases,modifies, abstract},
  morekeywords=[2]{WP, Post, Pre, PURE, DIV, STATE, EXN, ALL, TotST, M, GHOST, BAD, GTot, Lemma},
  morekeywords=[3]{},
  morekeywords=[4]{},
  morestring=[b]",
  sensitive=true,%
  numbersep=4pt,
  columns=[l]fullflexible,
  texcl=true,
  mathescape=true,
  identifierstyle={\sffamily},
  keywordstyle=[1]{\sffamily\maybecolor{dkblue}},
  keywordstyle=[2]{\sffamily\maybecolor{dkblue}},
  keywordstyle=[3]{\maybecolor{dkred}},
  keywordstyle=[4]{\rmfamily\itshape},
  rangeprefix=(*---\ ,
  includerangemarker=false,
  stringstyle=\ttfamily,
  showspaces=false,
  morecomment=[n]{(*}{*)},
  commentstyle={\itshape\maybecolor{dkred}},
  literate={->}{$\rightarrow\,$}{1}
           {<-}{$\leftarrow\,$}{1}
           {<>}{$\neq\,$}{1}
	   {delta}{$\delta$}{1}
	   {exists}{$\exists$}{1}
	   {forall}{$\forall$}{1}
      	   {True}{$\top$}{1}
      	   {False}{$\perp$}{1}           
           {fun}{$\lambda$}{1}
           {\\in}{$\in\,$}{1}
           {~}{$\neg$}{1}           
           {'a}{$\alpha$}{1}
           {'b}{$\beta$}{1}
           {'c}{$\gamma$}{1}
           {tau}{$\tau\,$}{1}
           {STATE0}{$\mathsf{ST}^\prime\,$}{1}
	   {/\\}{$\wedge\;$}{1}
	   {\\/}{$\vee\;$}{1}
           {>=}{$\geq\ $}{2}
           {<=}{$\leq\ $}{2}
	   {<==>}{$\Longleftrightarrow \ $}{3}
	   {==>}{$\Longrightarrow \ $}{3}
           {_}{\textunderscore}{1}
           {--}{-}{1}   
           ,
  breaklines=false}

  \let\lst\lstinline
  \let\li\lstinline

\newcommand\maybecolor[1]{\color{#1}}
\definecolor{dkblue}{rgb}{0,0.1,0.5}
\definecolor{dkgreen}{rgb}{0,0.4,0}
\definecolor{dkred}{rgb}{0.6,0,0}
\definecolor{dkpurple}{rgb}{0.7,0,1.0}
\definecolor{olive}{rgb}{0.4, 0.4, 0.0}
\definecolor{teal}{rgb}{0.0,0.4,0.4}
\definecolor{azure}{rgb}{0.0, 0.5, 1.0}
\definecolor{butter3}{HTML}{C4A000}
\definecolor{lightblue}{rgb}{0.2,0.2,1.0}
\definecolor{lightgrey}{rgb}{0.8,0.8,0.8}
\definecolor{linkColor}{rgb}{0,0,0.5}
\definecolor{lightgray}{rgb}{.9,.9,.9}
\definecolor{darkgray}{rgb}{.4,.4,.4}
\definecolor{purple}{rgb}{0.65, 0.12, 0.82}

\renewcommand{\comment}[1]{{\ifdraft\color{red}#1\fi}}
\newcommand{\comm}[3]{\ifdraft{\color{#1}[#2: #3]}\fi}
\newcommand{\cf}[1]{\comm{olive}{Cédric}{#1}}
\newcommand{\nik}[1]{\comm{dkpurple}{Nik}{#1}}
\newcommand{\ch}[1]{\comm{teal}{CH}{#1}}

\newcommand{\jp}[1]{\comm{butter3}{JP}{#1}}

\newcommand{\ignore}[1]{}

\newcommand{\fstar}{\relax\ifmmode{F^\ast}\else{F$^\ast$}\fi\xspace}
\newcommand{\lowstar}{\relax\ifmmode{Low^\ast}\else{Low$^\ast\!$}\fi\xspace}
\newcommand{\lamstar}{\texorpdfstring{\ensuremath{\lambda}ow\ensuremath{^\ast}\xspace}{low*}}
\newcommand{\cstar}{\relax\ifmmode{C^\ast}\else{C$^\ast$}\fi\xspace}
\newcommand{\haclstar}{HACL$^\ast$\xspace}
\newcommand{\kremlin}{KreMLin\xspace}
\newcommand{\emfhst}{EMF$^\ast_{\text{HST}}$\xspace}
\newcommand{\emfst}{EMF$^\ast_{\text{ST}}$\xspace}

\newcommand{\fref}[1]{Figure~\ref{fig:#1}}
\newcommand{\sref}[1]{§\ref{sec:#1}}
\newcommand{\lref}[1]{line~\ref{line:#1}}

\newcommand{\kw}[1]               {\text{\textsf{#1}}}
\newcommand{\kword}{\kw{word}}

\newcommand{\kint}{\kw{int}}

\newcommand{\etlet}[2]            {\ensuremath{\kw{let}\ #1=#2}}
\newcommand{\elet}[3]             {\ensuremath{\kw{let}\ #1=#2\ \kw{in}\ #3}}

\newcommand{\eifthenelse}[3]      {\ensuremath{\kw{if}\ #1\ \kw{then}\ #2\ \kw{else}\ #3}}
\newcommand{\eif}[3] {\kw{if}\;#1\;\kw{then}\;#2\;\kw{else}\;#3}

\newcommand{\eapply}[2]           {\ensuremath{#1(#2)}}

\newcommand{\enewbuf}[2]          {\ensuremath{\kw{newbuf}\ #1\ #2}}
\newcommand{\ereadbuf}[2]         {\ensuremath{\kw{readbuf}\ #1\ #2}}
\newcommand{\ewritebuf}[3]        {\ensuremath{\kw{writebuf}\ #1\ #2\ #3}}
\newcommand{\esubbuf}[3]          {\ensuremath{\kw{subbuf}\ #1\ #2\ #3}}
\newcommand{\enewstruct}[1]       {\ensuremath{\kw{newstruct}\ #1}}
\newcommand{\ereadstruct}[1]         {\ensuremath{\kw{readstruct}\ #1}}
\newcommand{\ewritestruct}[2]        {\ensuremath{\kw{writestruct}\ #1\ #2}}
\newcommand{\structfield}[0]          {\triangleright}
\newcommand{\estructfield}[2]          {\ensuremath{{#1} \structfield {#2}}}

\newcommand{\ecfun}[5]                {%
  \ensuremath{\kw{fun}\;#1\,(#2:#3):#4\,\{\;#5\;\}}%
}
\newcommand{\ecfuntwo}[4]                {%
  \ensuremath{\kw{fun}\;(#1:#2):#3\,\{\;#4\;\}}%
}

\newcommand{\evardecl}[3]            {\ensuremath{#1\ #2 = #3}}

\newcommand{\earray}[3]           {\ensuremath{#1\ #2[#3]}}
\newcommand{\memset}[3]           {\kw{memset}\;#1\;#2\;#3}

\newcommand{\eread}[1]         {\ensuremath{*[#1]}}
\newcommand{\ewrite}[2]         {\ensuremath{*[#1] = #2}}
\newcommand{\ereturn}[1]         {\ensuremath{\kw{return}\ #1}}
\newcommand{\symread}{\kw{read}}
\newcommand{\symwrite}{\kw{write}}

\newcommand{\stmts}{ss}
\newcommand{\ls}[1]{\overrightharpoon{#1}}
\newcommand{\eptrfd}[2]{\&#1 \rightarrow #2}
\newcommand{\symhole}{\Box}
\newcommand{\None}{\bot}
\newcommand{\Some}[1]{\lfloor #1 \rfloor}
\newcommand{\sympartial}{\rightharpoonup}
\newcommand{\eval}[2]{{\left \lsem #1 \right \rsem}_{#2}}
\newcommand{\symget}{\kw{Get}}
\newcommand{\symset}{\kw{Set}}
\newcommand{\step}{\leadsto}
\newcommand{\astep}{\rightarrowtriangle}
\newcommand{\fplug}[2]{#1\;[#2]}
\newcommand{\option}[1]{\widetilde{#1}}
\newcommand{\withframe}{\kw{withframe}}

\newcommand{\epop}{\kw{pop}}
\newcommand{\subst}[3]{ [#2/ #1] #3 }
\newcommand{\lowtoc}[1]{\downarrow #1}
\newcommand{\lowtoce}[1]{\downlsquigarrow #1}
\newcommand{\lowtocd}[1]{\downdownarrows #1}
\newcommand{\ctolow}[1]{\uparrow #1}
\newcommand{\ctolowe}[1]{\uprsquigarrow #1}
\newcommand{\ctolowd}[1]{\upuparrows #1}
\newcommand{\ctolowc}[1]{\rhookuparrow #1}
\newcommand{\ctolowE}[1]{\twoheaduparrow #1}
\newcommand\defeq{\mathrel{\overset{\makebox[0pt]{\mbox{\normalfont\tiny\sffamily def}}}{=}}}
\newcommand{\sys}{\kw{sys}}
\newcommand{\unstuck}{\kw{unstuck}}
\newcommand{\normal}[2]{\Downarrow_{#2} #1}
\newcommand{\unravel}{\kw{unravel}}
\newcommand{\foldl}{\kw{foldl}}
\newcommand{\unravelframe}{\kw{unravel\_frame}}
\newcommand{\mem}{\kw{mem}}

\newcommand{\eand}{\wedge}
\newcommand{\eor}{\vee}

\newcommand{\olabel}{o}
\newcommand{\brt}{\kw{brT}}
\newcommand{\brf}{\kw{brF}}

\newcommand \lp{P}
\newcommand \lexp{e}
\newcommand \fd{f}
\newcommand \lv{v}
\newcommand \trace{\ell}

\newcommand \cp{\hat{P}}
\newcommand \cexp{\hat{e}}
\newcommand \cstmt{s}

\usepackage{xspace}

\newcommand*{\ii}[1]{\ensuremath{\mathit#1}}

\usepackage{float}
  \floatstyle{plaintop}
  \restylefloat{table}

\begin{document}


\setlength{\pdfpageheight}{\paperheight}
\setlength{\pdfpagewidth}{\paperwidth}




\title{Verified Low-Level Programming Embedded in \fstar
}

\ifanon
\author{}
\else
\author{Jonathan~Protzenko}
\affiliation{\institution{Microsoft Research}\ifcamera, USA\fi}
\author{Jean-Karim~Zinzindohoué}
\affiliation{\institution{INRIA Paris}\ifcamera, France\fi}
\author{Aseem~Rastogi}
\affiliation{\institution{Microsoft Research}\ifcamera, USA\fi}
\author{Tahina~Ramananandro}
\affiliation{\institution{Microsoft Research}\ifcamera, USA\fi}
\author{Peng~Wang}
\affiliation{\institution{MIT CSAIL}\ifcamera, USA\fi}
\author{Santiago~Zanella-Béguelin}
\affiliation{\institution{Microsoft Research}\ifcamera, USA\fi}
\author{Antoine~Delignat-Lavaud}
\affiliation{\institution{Microsoft Research}\ifcamera, USA\fi}
\author{C\u{a}t\u{a}lin~Hri\c{t}cu}
\affiliation{\institution{INRIA Paris}\ifcamera, France\fi}
\author{Karthikeyan Bhargavan}
\affiliation{\institution{INRIA Paris}\ifcamera, France\fi}
\author{Cédric~Fournet}
\affiliation{\institution{Microsoft Research}\ifcamera, USA\fi}
\author{Nikhil~Swamy}
\affiliation{\institution{Microsoft Research}\ifcamera, USA\fi}
\makeatletter
\renewcommand{\@shortauthors}{Protzenko \emph{et.al.}}
\makeatother
\fi








\begin{abstract}
We present \lowstar, a language for low-level programming and
verification, and its application to high-assurance optimized cryptographic libraries.
\lowstar is a shallow embedding of a small, sequential, well-behaved
subset of C in \fstar, a dependently-typed variant of ML aimed at
program verification.
Departing from ML, \lowstar does not involve any garbage collection or
implicit heap allocation; instead, it has a structured memory model
\`a la CompCert, and it provides the control required for writing
efficient low-level security-critical code.

By virtue of typing, any \lowstar program is memory safe.
In addition, the programmer can make full use of the verification
power of \fstar to write high-level specifications and verify the
functional correctness of \lowstar code using a combination of SMT
automation and sophisticated manual proofs.
At extraction time, specifications and proofs are erased, and the
remaining code enjoys a predictable translation to C.
We prove that this translation preserves semantics and
side-channel resistance.

We provide a new compiler back-end from \lowstar to C and, 
to evaluate our approach, we implement and verify various cryptographic algorithms,
constructions, and tools for a total of about 28,000
lines of code, specification and proof.
We show that our \lowstar code delivers performance competitive with
existing (unverified) C cryptographic libraries, suggesting our
approach may be applicable to larger-scale low-level
software.
\cf{This suggestion is lukewarm; could also talk about usability and
  the risks of crypto libraries. Noop.}
\end{abstract}

\maketitle

\section{Introduction}

In the pursuit of high performance, 
cryptographic software widely deployed throughout the internet 
is still often subject to dangerous 
attacks~\cite{ws96,
  freestart, Heartbleed, tlstrunc, BhargavanDFPS14, ap13, POODLE,
  tlsattacker, BEAST, cryptoeprint:2016:475, sweet32, CRIME, polybug2,
  chachaoverflow,
  PincusB04, Szekeres2013,
  AfekS07,UseAfterFreeCWE,
  Dobrovitski03,DoubleFreeCWE,
  IntegerOverflowCWE}. 
Recognizing a clear need, 
the programming language,
verification, and 
applied cryptography communities are 
devoting
significant efforts to develop implementations proven
secure by construction against 
broad classes of these attacks.

Focusing on low-level attacks caused by violations of
memory safety, several researchers have used high-level, type-safe
programming languages to implement standard protocols such as Transport Layer
Security (TLS).
For example, \citet{kaloper2015not} provide nqsbTLS, an implementation of TLS in
OCaml, which by virtue of its type and memory safety is impervious to attacks
(like \citet{Heartbleed}) that exploit buffer overflows.
\citet{bhargavan2014proving} program miTLS
in F\#, also enjoying type and memory safety, but go further using a refinement
type system to prove various higher-level
security properties of the protocol.
While this approach is attractive for its simplicity, to get
acceptable performance, both nqsbTLS and miTLS link with fast, unsafe
implementations of complex cryptographic algorithms, such as those provided by
\citet{nocrypto}, an implementation that mixes C and OCaml, and
libcrypto, a component of the widely used \citet{openssl}.
In the worst case, linking with vulnerable C code 
can void all the 
guarantees of the high-level code.

In this paper, we aim to bridge the gap between high-level,
safe-by-construction code, optimized for clarity and ease of
verification, and low-level code exerting fine control
over data representations and memory layout in order to achieve
better performance.
To this end, 
we introduce \lowstar, a domain-specific language
for verified, efficient, low-level programming embedded within 
\fstar~\cite{mumon}, an ML-like language with dependent
types designed for program verification.
We
use \fstar
to prove functional correctness
and security properties of high-level code. 
Where efficiency is paramount, we drop into its C-like \lowstar subset
while still relying on
the verification capabilities of~\fstar
to prove our code is memory safe, functionally correct, and secure.

We have applied \lowstar to program and verify a range of sequential low-level
programs, including libraries for multi-precision arithmetic and
buffers, and various cryptographic algorithms, constructions, and
protocols built on top of them.
Our experiments indicate that compiled \lowstar code yields performance
on par with existing C code. This code can be used on its own,
or used within existing software through the C ABI. In particular, our C code may
be linked to \fstar programs compiled to
OCaml, providing large speed-ups via its foreign-function interface (FFI)
without compromising safety or security.%




\subsection*{An Embedded DSL, Compiled Natively}

\lowstar programs are a subset of \fstar programs:
the programmer
writes \lowstar code using regular \fstar syntax, against a library we
provide that models a lower-level view of memory, akin to the
structured memory layout of a well-defined C program
(this is similar to the structured memory model
of CompCert~\cite{Leroy-Compcert-CACM, 2012-Leroy-Appel-Blazy-Stewart},
rather than the ``big array of bytes'' model systems programmers sometimes use).
\lowstar programs interoperate naturally with
other \fstar programs, and precise specifications of \lowstar and
\fstar code are intermingled when proving properties of their
combination. As usual in \fstar, programs are type-checked and compiled
to OCaml for execution, after erasing their computationally irrelevant parts,
such as proofs and specifications, using a process similar
to Coq's extraction mechanism~\cite{letouzey2002new}. 
In particular, our memory-model library compiles to a simple, heap-based OCaml implementation.

Importantly, \lowstar programs have a second, equivalent but more efficient
semantics via compilation to~C, with the predictable
performance that comes with a native implementation of their
lower-level memory model. 
This compilation is implemented by KreMLin, a new
compiler from the \lowstar subset of \fstar to~C. 
Figure~\ref{fig:bigpicture} illustrates the high-level design of
\lowstar and its compilation to native code.
Our main contributions are as follows:

{\hypersetup{hidelinks=true}
\begin{figure}
  \centering

  \definecolor{butter1}{HTML}{FCE94F}
  \definecolor{butter2}{HTML}{EDD400}
  \definecolor{orange1}{HTML}{FCAF3E}
  \definecolor{orange2}{HTML}{F57900}
  \definecolor{orange3}{HTML}{CE5C00}
  \definecolor{chocolate1}{HTML}{E9B96E}
  \definecolor{chocolate2}{HTML}{C17D11}
  \definecolor{chocolate3}{HTML}{8F5902}
  \definecolor{chameleon1}{HTML}{8AE234}
  \definecolor{chameleon2}{HTML}{73D216}
  \definecolor{chameleon3}{HTML}{4E9A06}
  \definecolor{skyblue1}{HTML}{729FCF}
  \definecolor{skyblue2}{HTML}{3465A4}
  \definecolor{skyblue3}{HTML}{204A87}
  \definecolor{plum1}{HTML}{AD7FA8}
  \definecolor{plum2}{HTML}{75507B}
  \definecolor{plum3}{HTML}{5C3566}
  \definecolor{red1}{HTML}{EF2929}
  \definecolor{red2}{HTML}{CC0000}
  \definecolor{red3}{HTML}{A40000}
  \definecolor{aluminium1}{HTML}{EEEEEC}
  \definecolor{aluminium2}{HTML}{D3D7CF}
  \definecolor{aluminium3}{HTML}{BABDB6}
  \definecolor{aluminium4}{HTML}{888A85}
  \definecolor{aluminium5}{HTML}{555753}
  \definecolor{aluminium6}{HTML}{2E3436}

  \definecolor{skyblue4}{HTML}{9CBBDD}
  \definecolor{plum4}{HTML}{C5A5C1}
  \definecolor{chameleon4}{HTML}{ADEA70}

  \begin{tikzpicture}[lang/.style={
    draw, thick, text width=4em, node distance=3cm, align=center,
    text height=.8em, text depth=.5ex
  },label/.style={
    text=\footnotesize
  },every node/.style={
    font=\sffamily\footnotesize
  }]

    \node [thick, fill=skyblue1, draw=skyblue3, minimum width=16em, minimum height=4em] (emf) {};
    \node [anchor=north west, draw=none] at ($(emf.north west)+(.2em,-.2em)$)
      { EM\fstar };
    \node [anchor=south east, thick, fill=skyblue2, draw=skyblue3, minimum width=12em, minimum height=3em]
      at ($(emf.south east)+(-.4em,.4em)$)
      (lowstar) {};
    \node [anchor=north west, draw=none, text=aluminium1] at (lowstar.north west)
      (lowlabel)
      { \lowstar };
    \node (foemf)
      [lang,fill=skyblue3,text=aluminium1,anchor=south east,text width=5.8em]
      at ($(lowstar.south east)+(-.4em,.4em)$)
      { 1\textsuperscript{st}-order EM\fstar };

    \node (lamstar) [lang, fill=plum2,text=aluminium1, anchor=east, node distance=4em]
      at ($(lowstar.east)+(0,-4.5em)$)
      { \lamstar };
    \node (cstar)   [lang, fill=plum2,text=aluminium1, left of=lamstar] { \cstar };
    \node (c)       [lang, fill=plum2,text=aluminium1, left of=cstar] { Clight };
    \node (cfile)   [lang, fill=chameleon3,text=aluminium1, anchor=west, node distance=4em]
      at ($(c.west)+(0,-4.5em)$)
      { \ttfamily{}.c };
    \node (exe)     [lang, fill=chameleon3,text=aluminium1, right of=cfile] { Exe };

    \draw [->, ultra thick]
      [postaction={
          decorate,
          decoration={
            raise=.4ex,
            text along path,
            text align=center,
            text={|\sffamily\scriptsize\color{aluminium1}|{$\approx$} erase} }
            }]
      [postaction={
          decorate,
          decoration={
            raise=-1.55ex,
            text along path,
            text align=center,
            text={|\sffamily\scriptsize\color{aluminium1}|{\S}3.0}
          }
      }]
    (lowlabel.south) to [out=-90,in=180] (foemf.west);
    \draw [ultra thick,->]
      (foemf.south)
      to [out=-90,in=90]
      node[left,yshift=-.5ex] { \scriptsize partial $\approx$ }
      node[right=.5ex,yshift=-.5ex] { \sref{lamstar} }
      (lamstar);
    \path [ultra thick,->] (cstar)
      edge [loop below]
      node [text=plum3,yshift=.3ex] (hoist) { \scriptsize hoist $\approx$ }
      (cstar);
    \draw [ultra thick,->] (lamstar) --
      node[above] { \normalsize$\approx$ }
      node[below] { \sref{lamstar-to-cstar} }
      ++ (cstar);
    \draw [ultra thick,->] (cstar) --
      node[above] { \normalsize$\approx$ }
      node[below] { \sref{to-clight} }
      ++ (c);
    \draw [ultra thick,->] (c) --
      node[midway,text=plum3] (print) { \hspace{2.5em}\scriptsize{}print }
      ++ (cfile);
    \draw [->, ultra thick] (cfile) --
      node[above] { \scriptsize{}compile }
      ++ (exe);

  \begin{scope}[on background layer]
    \draw [thick, fill=chameleon1, draw=chameleon3]
      ($(cfile.north west)+(-.4em,.4em)$) rectangle ($(exe.south east)+(.4em,-.4em)$);
    \draw [thick, draw=plum3,fill=plum1]
      let \p1 = (c.west) in
      let \p2 = (lamstar.north) in
      ($(\x1, \y2)+(-.4em,.4em)$) rectangle ($(lamstar.south east)+(.4em,-.4em)$);

    \node [draw=none,anchor=north east] at ($(emf.north west)+(0,.2em)$)
      {\footnotesize\color{skyblue3}\fstar};
    \path
      let \p1 = (c.west) in
      let \p2 = (lamstar.north) in
      node [draw=none,anchor=south west]
      at ($(\x1, \y2)+(-.4em,.45em)$)
      {\footnotesize\color{plum3}Kremlin };
    \node [draw=none,anchor=south west] at ($(exe.south east)+(.4em,-.8em)$)
      {\footnotesize\color{chameleon3}GCC/Clang/CompCert };
  \end{scope}

  \end{tikzpicture}

  \caption{\lowstar embedded in \fstar, compiled to C, with soundness and
  security guarantees (details in \sref{formal})}
  \label{fig:bigpicture}
  \vspace{-2ex}
\end{figure}
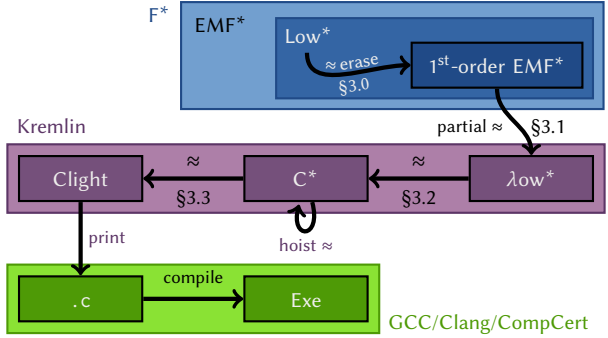
}


\paragraph*{Libraries for low-level programming within \fstar~(\S\ref{sec:examples})}
At its core, \fstar is a purely functional language
to which effects like state are added programmatically using monads.
In this work, we instantiate the state monad of \fstar to use a CompCert-like
structured memory model that separates the stack and the heap,
supporting bulk 
allocation and deallocation on the stack,
and allocating and freeing individual blocks 
on the heap.
Both the heap and the stack are further divided into disjoint logical
regions, which enables us to manage the separation properties
necessary for modular, stateful verification.
On top of this, we program a library of C-style arrays and structs passed by
reference, with support for pointer arithmetic and pointers to the interior of
an array or a struct.  
By virtue of \fstar typing, our libraries and all their well-typed
clients are guaranteed to be memory safe, e.g., they never access
out-of-bounds or deallocated memory.

\paragraph*{Designing \lowstar, a subset of \fstar easily compiled to C}
We intend to give \lowstar programmers precise control over the
performance profile of the generated C code. As much as possible, 
we aim for the programmer to control even the
syntactic structure of the C code, to
facilitate its review by security experts unfamiliar with \fstar. 
As such, to
a first approximation, \lowstar programs are \fstar programs
well-typed in the state monad described above, which, after all their
computationally irrelevant (ghost) parts have been erased, must
meet several restrictions, as follows: the code 
(1) must be first order, to prevent the need to allocate closures in C; 
(2) must make any heap allocation explicit; 
(3) must not use any recursive datatype, since these would have to be compiled using additional
indirections to C structs; and 
(4) must be monomorphic, since C does not support polymorphism directly. 
%
Importantly, \lowstar heavily leverages \fstar's capabilities for
partial evaluation, hence allowing the programmer to write high-level, reusable
code that is normalized via meta-programming into the \lowstar subset before
the restrictions are enforced.
We emphasize that these restrictions apply only to computationally
relevant code---proofs and specifications are free to use arbitrary
higher-order, dependently typed \fstar, and very often they do.

\paragraph*{A formal translation from \lowstar to Comp\-Cert Clight (\sref{formal})}
Justifying its dual interpretation as a subset of \fstar and a subset of C,
we provide a formal model of \lowstar, called \lamstar,
give a translation from \lamstar to 
Clight~\cite{Blazy-Leroy-Clight-09} and show that it
preserves 
trace equivalence with respect to the original
\fstar semantics. In addition to ensuring that the
functional behavior of a program is preserved, our trace equivalence also
guarantees that the compiler does not accidentally introduce 
side-channels due to 
memory access patterns (as would be the case without the restrictions above)
at least until it reaches Clight, a useful sanity check for cryptographic code.
%
Our theorems cover the translation of standalone \lamstar
programs to~C, proving that execution in C preserves the original
\fstar semantics of the \lamstar program. 

\paragraph*{
KreMLin, a compiler from \lowstar to C (\S\ref{sec:impl})}
Our formal development guides the implementation of KreMLin, a new
tool that emits C code from \lowstar. KreMLin is designed to emit
well-formatted, idiomatic C code suitable for manual review.
The resulting C programs can be
compiled with Comp\-Cert for greatest assurance, and with mainstream C
compilers, including GCC and Clang, for greatest performance.  
We have used KreMLin to extract to C the 20,000+ lines of 
\lowstar code we have written so far.
After compilation, our verified standalone C libraries can be
integrated within larger programs using standard means.
\cf{LATER: \nik{moved it here from the previous para}.Does KreMLin
  support the generation of FFI? Do we really claim that ocaml-interop
  is the only part of Kremlin not covered by our formal results? This
  is our usual ambiguity between \lowstar as our formal source
  language vs as the subset of \fstar currently supported by KreMLin.}

\paragraph*{An empirical evaluation (\S\ref{sec:moreexamples})}
We present a few developments of efficient, verified, interoperable
cryptographic libraries programmed in \lowstar.

(1) We provide \haclstar, a ``high-assurance crypto library''
implementing and proving (in $\sim$6,000 lines of \lowstar) several
cryptographic algorithms, including the Poly1305
MAC~\cite{bernstein2005poly1305}, the ChaCha20 cipher~\cite{chacha20},
and multiplication on the Curve25519 elliptic curve
~\cite{curve25519}.
We package these algorithms to provide the popular NaCl
API~\cite{bernstein2012security}, yielding the first performant
implementation of NaCl verified for memory safety and side-channel
resistance, along with functional correctness proofs for its core
components, including a verified bignum library customized for safe,
fast cryptographic use (\S\ref{sec:haclstar}).  
Using this API, we build new standalone applications such as
 \emph{PneuTube}, a new secure, asynchronous, file transfer 
application whose performance compares favorably with widely 
used, existing utilities like scp.
\cf{Should we cite the CSF paper? Not all of it is new}

(2) Emphasizing the applicability of \lowstar for high-level,
cryptographic security proofs on low-level code, we briefly describe
its use in programming and proving (in $\sim$14,000 lines of \lowstar)
the Authenticated Encryption with Associated Data (AEAD) construction
at the heart of the record layer of the new TLS 1.3 Internet Standard.
We prove memory safety, functional correctness, and cryptographic
security for its main ciphersuites, relying, where available, on
verified implementations of these ciphersuites provided by \haclstar.
The C code extracted from our verified implementation is easily
integrated within other applications, including, for example, an
implementation in \fstar of TLS separately verified and compiled to
OCaml (through OCaml's FFI).


%
\cf{LATER: cite the oakland paper.}

\paragraph*{Trusted computing base} To date,
we have focused on designing and evaluating our methodology of
programming and verifying low-level code shallowly embedded within a
high-level programming language and proof assistant.
We have yet to invest effort in minimizing the trusted computing base
of our work, an effort we plan to expend now that we have evidence
that our methodology is worthwhile.
Currently, the trusted computing base of our verified libraries
includes the implementation of the \fstar typechecker and the Z3 SMT
solver~\citep{MouraB08}.
Additionally, we trust the manual proofs of the metatheory relating
the semantics of \lamstar to CompCert Clight. The KreMLin tool is
informed by this metatheory, but is currently implemented in
unverified OCaml, and is also trusted.
Finally, we inherit the trust assumptions of the C compiler used to
compile the code extracted from \lowstar.

\paragraph*{Supplementary materials}
First, we provide, in the appendix, the hand proofs
of the theorems described in \sref{formal}.
The present paper is focused on the metatheory and tools; we also authored a
companion paper~\cite{record} that describes the cryptographic model we
used for the record layer of TLS 1.3.
Finally, we have an ongoing submission of a paper focused on our
\haclstar library~\cite{haclstar}, where we
describe in greater detail our
proof techniques for reusing the bignum formalization across different
algorithms and implementations, and provide a substantial performance
evaluation.

We also offer numerous software artifacts. Our tool \kremlin~\cite{kremlin} is
actively developed on GitHub, and so is \haclstar~\cite{haclstar-gh}. Most of
the \lowstar libraries live in the \fstar repository, also on GitHub. The
integration of \haclstar within miTLS is also available on GitHub. For
convenience, we offer a regularly-updated Docker image of Project
Everest~\cite{everest-website}, which bundles together
\fstar, miTLS, \haclstar, \kremlin. One may
fetch it via \texttt{docker pull projecteverest/everest}.  The Docker image
contains a \texttt{README.md} with an overview of the proofs and the code.

\ifpagelimits
\newpage
\fi
\begin{figure*}[t!]\small
\begin{tabular}{c|c}
\begin{lstlisting}[basicstyle=\footnotesize,language=fstar]
let chacha20 
  (len: uint32{len <= blocklen}) !\label{line:c20:args1}!
  (output: bytes{len = output.length}) !\label{line:c20:args2}!
  (key: keyBytes)
  (nonce: nonceBytes{disjoint [output; key; nonce]}) 
  (counter: uint32) : Stack unit!\label{line:c20:args3}!
  (requires (fun m0 ->  output $\in$ m0 /\ key $\in$ m0 /\ nonce $\in$ m0))!\label{line:c20:spec1}!
  (ensures (fun m0 _ m1 -> modifies$_1$ output m0 m1 /\ !\label{line:c20:spec2}!
      m1.[output] == 
      Seq.prefix len (Spec.chacha20 m0.[key]  m0.[nonce]) counter))) =!\label{line:c20:spec3}!
    push_frame ();! \label{line:c20:push}!
    let state = Buffer.create 0ul 32ul in! \label{line:c20:create}!
    let block = Buffer.sub state 16ul 16ul in! \label{line:c20:sub}!
    chacha20_init block key nonce counter;
    chacha20_update output state len;
    pop_frame ()! \label{line:c20:pop}!
\end{lstlisting}
&\quad
\;\;
\begin{lstlisting}[basicstyle=\footnotesize,language=C]
void chacha20 ( 
  uint32_t len,
  uint8_t *output, 
  uint8_t *key,
  uint8_t *nonce,
  uint32_t counter)




{
  uint32_t state[32] = { 0 };
  uint32_t *block = state + 16;
  chacha20_init(block, key, nonce, counter);
  chacha20_update(output, state, len);
}
\end{lstlisting}
\end{tabular}
\caption{A snippet from ChaCha20 in \lowstar (left) and its C compilation
(right)\jp{
  Add some disjoint clauses.
}\cf{This code is further away from hacl-star.}}
\label{fig:chacha20-both}
\end{figure*}
\section{A \lowstar Tutorial}
\label{sec:examples}

At the core of \lowstar is a library for programming with
structures and arrays manually allocated on the stack or the heap (\sref{dsl}).
\cf{A bit restrictive? FStar.HyperStack is important too!}
Memory safety demands reasoning about the extents and
liveness of these objects, while functional correctness and security
may require reasoning about their contents. Our library
provides 
specifications to allow client code to be
proven safe, correct and secure,
while KreMLin compiles such verified client code to C.

We illustrate the design of \lowstar using several examples from our
codebase. We show the ChaCha20 stream cipher~\cite{chacha20}, focusing on
memory safety (\sref{chacha20}), and the Poly1305
MAC~\cite{bernstein2005poly1305}, focusing on functional correctness.
(\sref{poly1305}).
Going beyond functional correctness, we explain how we prove a
combination of ChaCha20 and Poly1305 cryptographically secure
(\sref{crypto}).
Throughout, we point out key benefits of our approach, notably our use
of dependently typed metaprogramming to work at a relatively high-level of
abstraction at little performance cost. \nik{Can the stackinline
  stuff be folded into an extended version of \sref{crypto}?}

\subsection{A First Example: the ChaCha20 Stream Cipher}
\label{sec:chacha20}



Figure~\ref{fig:chacha20-both} shows code snippets for the core
function of ChaCha20~\cite{chacha20}, a modern stream cipher widely
used for fast symmetric encryption.
This function computes a block of pseudo-random bytes, usable for
encryption, for example, by XORing them with a plaintext message.
On the left is our \lowstar code; on the right its compilation to
C. 
The function takes as arguments an output length and buffer, and some
input key, nonce, and counter. 
It allocates 32 words of auxiliary, contiguous state on the stack; then it calls a
function to initialize the cipher block from the inputs (passing an
interior pointer to the state); and finally it calls another function that
computes a cipher block and copies its first \lst$len$ bytes to the
output buffer.

Aside from the erased specifications at
lines~\ref{line:c20:spec1}--\ref{line:c20:spec3}, the C code is in
one-to-one correspondence with its \lowstar counterpart.
These specifications capture the safe memory usage of \lst$chacha20$.
(Their syntax is explained next, in~\S\ref{sec:dsl}.)
For each argument passed by reference, and for the auxiliary state,
they keep track of their liveness and size. They also capture its
correctness, by describing the final state of the \lst$output$ buffer
using a pure function.

Lines~\ref{line:c20:args1}--\ref{line:c20:args2} use type refinements to require that the
\lst$len$ argument equals the length of the \lst$output$ buffer and it does not
exceed the block size. (Violation of these conditions would
lead to a buffer overrun in the call to \lst$chacha20_update$.)
Similarly, types \lst$keyBytes$ and \lst$nonceBytes$ specify
pointers to fixed-sized buffers of bytes.
The return type \lst$Stack unit$ on \lref{c20:args3} 
says that \lst$chacha20$ returns
nothing and may allocate on the stack, but not on the heap (in
particular, it has no memory leak).
On the next line, the pre-condition \li+requires+ that all arguments
passed by reference
be live.
On lines~\ref{line:c20:spec2}--\ref{line:c20:spec3}, the post-condition 
first \li+ensures+ that the function modifies at most the contents of \li+output+ (and, 
implicitly, that all buffers remain live).
We further explain this specification in the next subsection.
The rest of the post-condition specifies functional correctness:
the \lst$output$ buffer must contain a sequence of bytes
equal to the first \lst$len$ bytes of the cipher specified by 
function \lst$Spec.chacha20$ for the input values of \lst$key$, \lst$nonce$, and \lst$counter$.

As usual for symmetric ciphers, RFC 7539 specifies \lst$chacha20$ as
imperative pseudocode, and does not further discuss its mathematical
properties.
We implement this pseudocode as a series of pure functions in 
\fstar, which can be extracted to OCaml and tested for conformance 
with the RFC test vectors.
Functions such as \lst$Spec.chacha20$ then serve as logical
specifications for verifying our stateful implementation. 
In particular, the last postcondition of \lst$chacha20$ ensures that
its result is determined by its inputs.
We describe more sophisticated functional correctness proofs 
for Poly1305 in \S\ref{sec:poly1305}.

\subsection{\lowstar: An Embedded DSL for Low-Level Code}
\label{sec:dsl}

As in ML, by default \fstar does not provide an
explicit means to reclaim memory or to allocate memory on the stack, nor does
it provide support for pointing to the interior of arrays. Next, 
we sketch the design of a new \fstar library that provides a
structured memory model suitable for program verification, while
supporting low-level features like explicit freeing, stack
allocation, and interior pointers.
In subsequent sections, we describe how programs
type-checked against this library can be compiled safely to C.
First, however, we begin with some background on \fstar.

\paragraph*{Background:} \fstar is a dependently
typed language with support for user-defined monadic effects.
Its types separate computations from values, giving the former
\emph{computation types} of the form \lst!M t$_1$ $\ldots$ t$_n$! where
\lst!M! is an effect label and \lst!t$_1$ $\ldots$ t$_n$! are
\emph{value types}. For example, \lst!Stack unit (...) (...)! on lines
\ref{line:c20:spec1}--\ref{line:c20:spec2}
of Figure~\ref{fig:chacha20-both} is an instance of a
computation type, while types like \lst!unit! are value types.
There are two distinguished computation types: \lst!Tot t! is the type
of a total computation returning a \lst!t!-typed value;
\lst!Ghost t!, a computationally irrelevant computation returning a
\lst!t!-typed value.
Ghost computations are useful for specifications and proofs but
are erased when extracting to OCaml or C.
%
%

To add state to \fstar, one defines a state monad represented (as
usual) as a function from some initial memory \lst!m$_0$:s! to a
pair of a result \lst!r:a! and a final memory \lst!m$_1$:s!, for
some type of memory \lst!s!. Stateful computations are specified
using the computation type:
\begin{lstlisting}[numbers=none]
  ST (a:Type) (pre: s -> Type) (post: s -> a -> s -> Type)
\end{lstlisting}
Here, \lst$ST$ is a computation type constructor applied to three
arguments: a result type \lst!a!; a pre-condition predicate on the
initial memory, \lst$pre$; and a post-condition predicate relating the
initial memory, result and final memory. We generally annotate the
pre-condition with the keyword \lst$requires$ and the post-condition
with \lst$ensures$ for better readability.  A computation \lst$e$ of
type
\lst$ST a (requires pre) (ensures post)$, when run in an initial
memory \lst!m$_0$:s! satisfying \lst!pre m!, produces a result
\lst!r:a! and final memory \lst!m$_1$:s! satisfying
\lst!post m$_0$ r m$_1$!, unless it diverges.\footnote{\fstar recently
  gained support for proving stateful computations terminating. We
  have begun making use of this feature to prove our code terminating,
  wherever appropriate, but make no further mention of this.} \fstar
uses an SMT solver to discharge the verification conditions it
computes when type-checking a program.

\paragraph{Hyper-stacks: A region-based memory model for \lowstar}
For \lowstar, we instantiate the type \lst$s$ in the state monad to
\lst$HyperStack.mem$ (which we refer to as just
``hyper-stack''), a new region-based memory model~\citep{tt97regions}
covering both stack and heap. Hyper-stacks are a
generalization of hyper-heaps, a memory model proposed previously for
\fstar~\cite{mumon}, designed to provide lightweight support for
separation and framing for stateful verification. 
Hyper-stacks augment hyper-heaps with a shape invariant to indicate
that the lifetime of a certain set of regions follow a specific
stack-like discipline. We sketch the \fstar signature of hyper-stacks
next.

\paragraph*{A logical specification of memory} Hyper-stacks partition
memory into a set of regions.
%
%
Each region is identified by an
\lst$rid$ and regions are classified as either stack or heap regions,
according to the predicate \lst$is_stack_region$---we use the type
abbreviation \lst$sid$ for stack regions and \lst$hid$ for heap
regions. A distinguished stack region, \lst$root$, outlives all other stack regions. The
snippet below is the corresponding \fstar code.

\begin{lstlisting}[numbers=none]
type rid
val is_stack_region: rid -> Tot bool
type sid = r:rid{is_stack_region r}
type hid = r:rid{$\lnot$ (is_stack_region r)}
val root: sid
\end{lstlisting}

Next, we show the (partial) signature of \lst$mem$, our model of the
entire memory, which is equipped with a select/update
theory~\cite{mccarthy62} for typed references \lst$ref a$.
Additionally, we have a function to refer to the
\lst$region_of$ a reference, and a relation \lst$r \in m$ to indicate
that a reference is live in a given memory.

\begin{lstlisting}[numbers=none]
type mem
type ref : Type -> Type
val region_of: ref a -> Ghost rid
val `_ \in _`  : ref a -> mem -> Tot Type  (* a ref is contained in a mem *)
val `_ [_] `     : mem -> ref a -> Ghost a   (* selecting a ref *)
val `_ [_] <- _` : mem -> ref a -> a -> Ghost mem (* updating a ref *)
val rref r a = x:ref a {region_of x = r} (* abbrev. for a ref in region r *)
\end{lstlisting}

\paragraph*{Heap regions} By defining the \lst$ST$ monad
over the \lst$mem$ type, we can program stateful primitives for creating
new heap regions, and allocating, reading,
writing and freeing references in those regions---we show some of
their signatures below. Assuming an infinite amount of memory,
\lst$alloc$'s pre-condition is trivial while its post-condition
indicates that it returns a fresh reference in region \lst$r$
initialized appropriately.  Freeing and dereferencing ($!$) require their
argument to be present in the current memory, eliminating double-free
and use-after-free bugs.

\ch{Could save quite a bit of space by bringing the requires on the
  same line as ST. It looks more crammed, but might have to make
  these kind of compromises for space.}

\begin{lstlisting}[numbers=none]
val alloc: r:hid -> init:a -> ST (rref r a) (ensures  (fun m0 x m1 -> x $\not\in$ m0 /\ x \in m1 /\ m1 = (m0[x]<-init)))
val free: r:hid -> x:rref r a -> ST unit (requires (fun m -> x \in m)) (ensures  (fun m$_0$ $\_$ m$_1$ -> x $\not\in$ m$_1$ /\ forall y<>x. m$_0$[y] = m$_1$[y]))
val ($!$): x:ref a -> ST a (requires (fun m -> x \in m)) (ensures  (fun m0 y m1 -> m0 = m1 /\ y = m1[x]))
\end{lstlisting}

Since we support freeing individual references within a region, our
model of regions could seem similar to \citet{berger02reaps}'s
\emph{reaps}. However, at present, we do not support
freeing heap objects \emph{en masse} by deleting heap regions;
indeed, this would require using a special memory allocator.
Instead, for us heap regions serve only to {\em logically} partition
the heap in support of separation and modular
verification, as is already the case for hyper-heaps~\cite{mumon},
and heap region creation is currently compiled to a no-op by KreMLin.

\paragraph*{Stack regions,} which we will henceforth call {\em stack
  frames}, serve not just as a reasoning device, but provide the
efficient C stack-based memory management mechanism. KreMLin maps
stack frame creation and destruction directly to
the C calling convention and lexical scope.
To model this, we extend the signature of \lst$mem$ to include a
\lst$tip$ region representing the currently active stack frame, ghost
operations to \lst$push$ and \lst$pop$ frames on the stack
of an explicitly threaded memory,
and their effectful analogs, \lst$push_frame$ and \lst$pop_frame$ that modify
the current memory.
In \lst$chacha20$ in Fig.~\ref{fig:chacha20-both}, the
\lst$push_frame$ and \lst$pop_frame$ correspond precisely to the
braces in the C program that enclose a function body's scope.
We also provide a derived combinator, \lst[language={}]{with_frame},
which combines \lst$push_frame$ and \lst$pop_frame$ into a single,
well-scoped operation. Programmers are encouraged to use the
\lst[language={}]{with_frame} combinator, but, when more convenient
for verification, may also use \lst$push_frame$ and \lst$pop_frame$
directly. KreMLin ensures that all uses of \lst$push_frame$ and
\lst$pop_frame$ are well-scoped.
Finally, we show the signature of \lst$salloc$ which allocates a
reference in the current \lst$tip$ stack frame.

\begin{lstlisting}[numbers=none]
val tip: mem -> Ghost sid
val push: mem -> Ghost mem
val pop:  m:mem{tip m <> root} -> Ghost mem
val push_frame: unit -> ST unit (ensures  (fun m0 () m1 -> m1 = push m0))
val pop_frame: unit -> ST unit (requires (fun m -> tip m <> root)) (ensures  (fun m0 () m1 -> m1 = pop m0))
val salloc: init:a -> ST (ref a) (ensures (fun m0 x m1 -> x $\not\in$ m0 /\ x \in m1 /\ region_of x = tip m1 /\
                                               tip m0 = tip m1 /\ m1 = (m0[x] <- init)))
\end{lstlisting}

\paragraph*{The \lst$Stack$ effect} The specification of \lst$chacha20$
claims that it uses only stack allocation and has no memory leaks,
using the \lst$Stack$ computation type. This is straightforward to
define in terms of \lst$ST$, as shown below.

\begin{lstlisting}[numbers=none]
effect Stack a pre post = ST a (requires pre)
                           (ensures (fun m0 x m1 -> post m0 x m1 /\  tip m0 = tip m1 /\ (forall r. r \in m1<==>r \in m0)))
\end{lstlisting}

\noindent \lst$Stack$ computations are \lst$ST$ computations that
leave the stack tip unchanged (i.e., they pop all frames they may
have pushed) and yield a final memory with the same domain as the initial memory.
This ensures that \lowstar code with \lst+Stack+ effect has explicitly
deallocated all heap allocated references before returning,
ruling out memory leaks.
As such, we expect all externally callable \lowstar functions to have
\lst$Stack$ effect.
Other code can safely pass pointers to objects allocated in their
heaps into \lowstar functions with \li+Stack+ effect since the definition of
\li+Stack+ forbids the \lowstar code from freeing these references.

\cf{Comment that \lst$Stack$ may in principle return 'dangling
  pointers'?  type safety ensures such pointers cannot be dereferenced.}

\paragraph{Modeling arrays} Hyper-stacks separate heap and
stack memory, but each region of memory still only supports abstract,
ML-style references. A crucial element of low-level programming is
control over the specific layout of objects, especially for arrays and structs.
We describe first our modeling of arrays by implementing an abstract
type for buffers in \lowstar, using just the references provided by
hyper-stacks. Relying on its abstraction, KreMLin compiles our buffers to
native C arrays.

The type `\lst$buffer a$' below is a single-constructor inductive type
with 4 arguments.
Its main \lst$content$ argument
holds a reference to a \lst$seq a$, a purely functional sequence of
\lst$a$'s whose length is determined by the first argument
\lst$max_length$. The refinement type \lst$b:buffer uint32{length b = n}$ is
translated to a C declaration \li+uint32_t b[n]+ by KreMLin and, relying on C
pointer decay,\ch{ha?} further referred to via \li+uint32_t *+.

\begin{lstlisting}[numbers=none]
abstract type buffer a =
  | MkBuffer: max_length:uint32
    -> content:ref (s:seq a{Seq.length s = max_length})
    -> idx:uint32
    -> length:uint32 {idx + length <= max_length} -> buffer a
\end{lstlisting}


\noindent 
The last two arguments of a buffer are there to support creating smaller
sub-buffers from a larger buffer, via the \li+Buffer.sub+ operation
below. A call to `\lst$Buffer.sub b i l$' returning \lst!b$'$! is compiled
to C pointer arithmetic \li!b + i! (as seen in
Figure~\ref{fig:chacha20-both} line~\ref{line:c20:sub} in
\lst$chacha20$). To accurately model this, the \li+content+ field is
shared between \li+b+ and \li+b$'$+, but \li+idx+ and \li+length+
differ, to indicate that the sub-buffer \li+b$'$+ covers only a
sub-range of the original buffer \li+b+. 
The \lst+sub+ operation has computation type \li+Tot+, meaning that it
does not read or modify the state. The refinement on the result
\lst+b$'$+ indicates its length and, using the \lst$includes$
relation, records that \lst$b$ and \lst+b$'$+ are aliased.

\cf{pls check; is there a better place to explain machine integers?
  Also, shouldn't we account for the size of \lst$a$ to prevent
  overflows? Or use 64-bit indexes when compiling to a 64-bit
  architecture? We may also explain that (presumably)
  \lst$Buffer.create$ fails to allocate oversized buffers---I guess C
  programmers explicitly check for allocation errors. }

\cf{We had long discussions about the purity of \lst$sub$; noop for now?}

\begin{lstlisting}[numbers=none]
val sub: b:buffer a -> i:uint32 -> len:uint32{i + len <= b.length} -> Tot (b':buffer a{b'.length = len /\ b `includes` b'})
\end{lstlisting}

\cf{The first refinement above is at odd with my understanding of
  integer overflows. Have we decided to omit \lst$UInt32.v$ for
  simplicity?}

We also provide statically bounds-checked operations for indexing and
updating buffers. The signature of the \lst$index$ function below requires the
buffer to be live and the index location to be within bounds. Its
postcondition ensures that the memory is unchanged and describes what
is returned in terms of the logical model of a buffer as a sequence.

\begin{lstlisting}[numbers=none]
let get (m:mem) (b:buffer a) (i:uint32{i < b.length}) : Ghost a = Seq.index (m[b.content]) (b.idx + i)
val index: b:buffer a -> i:uint32{i < b.length} -> Stack a
  (requires (fun m -> b.content \in m))
  (ensures (fun m0 z m1 -> m1 = m0 /\ z = get m1 b i))
\end{lstlisting}

All lengths and indices are 32-bit machine integers, and refer to the number of
elements in the buffer, not the number of bytes the buffer occupies. 
This currently prevents addressing very large buffers on 64-bit platforms.
(To this end, we may parameterize our development over a C data model,
wherein indices for buffers would reflect the underlying (proper)
\li+ptrdiff_t+ type.)

Similarly, memory allocation remains platform-specific. 
It may cause a (fatal) error as it runs out of memory.
More technically, the type of \li+create+ may not suffice to prevent
pointer-arithmetic overflow; if the element size is greater than a
byte, and if the number of elements is $2^{32}$, then the argument
passed to \li+malloc+ will overflow on a platform where %
\li+sizeof size_t == 4+.
To prevent such cases, KreMLin inserts defensive dynamic checks
(which typically end up eliminated by the C compiler since our
stack-allocated buffer lengths are compile-time constants).
In the future, we may statically prevent it by mirroring the C \li+sizeof+
operator at the \fstar level, and requiring that for each
\li+Buffer.create+ operation, the resulting allocation size, in bytes,
is no greater than what \li+size_t+ can hold.

\paragraph{Modeling structs}
\label{sec:structs}

Generalizing `\lst$buffer t$' (abstractly, a reference to a finite map
from natural numbers \cf{unsigned machine integers?}%
to \lst$t$), we model C-style structs as an
abstract reference to a `\lst$struct key value$', that is, a map from keys
\lst$k:key$ to values whose type `\lst$value k$' depends on the
key. 
For example, we represent the type of a colored point as follows,
using a struct with two fields \lst$X$ and \lst$Y$ for coordinates and one 
field \lst$Color$, itself a nested struct of RGB
values.

\begin{lstlisting}[numbers=none]
type color_fields = R | G | B
type color = struct color_fields (fun R | G | B -> uint32)
type colored_point_fields = X | Y | Color
type colored_point = struct colored_point_fields (fun X | Y -> int32 | Color -> color)
\end{lstlisting}

\cf{Wondering how to better balance the buffer vs struct presentation,
  e.g. we don't have local examples for buffers and we don't discuss
  their framing}

C structs are flatly allocated; the declaration above models a contiguous memory
block that holds 20 bytes or more, depending on alignment constraints. As such,
we cannot directly perform pointer arithmetic within that block; rather, we 
navigate it by referring to fields. 
To this end, our library of structs provides an interface to
manipulate pointers to these C-like structs, including pointers that
follow a path of fields throughout nested structs. The main type
provided by our library is the indexed type \lst$ptr$ shown below,
encapsulating a base reference \lst$content: ref from$ and a path
\lst$p$ of fields leading to a value of type \lst$to$.

\begin{lstlisting}[numbers=none]
abstract type ptr: Type -> Type = Ptr: #from:Type -> content: ref from -> #to: Type -> p: path from to -> ptr to
\end{lstlisting}

When allocating a struct on the stack, the caller provides a
`\lst$struct k v$' literal and obtains a
`\lst$ptr (struct k v)$', a pointer to a struct literal in the current
stack frame (a \lst$Ptr$ with an empty path).

The \lst$extend$ operator below supports extending the access path
associated with a `\lst$ptr (struct k v)$' to obtain a pointer to one of
its fields.

\begin{lstlisting}[numbers=none]
val extend: #key: eqtype -> #value: (key -> Tot Type) -> p: ptr (struct key value) -> fd: key -> ST (ptr (value fd))
  (requires (fun h -> live h p))
  (ensures (fun h0 p$'$ h1 -> h0 == h1 /\ p$'$ == field p fd))
\end{lstlisting}

Finally, the \lst$read$ and \lst$write$ operations allows accessing and
mutating the field referred to by a \lst$ptr$.

\begin{lstlisting}[numbers=none]
val read: #a:Type -> p: ptr a -> ST value
  (requires (fun h -> live h p))
  (ensures (fun h0 v h1 -> live h0 p /\ h0 == h1 /\ v == as_value h0 p))

val write: #a:Type -> b:ptr a -> z:a -> ST unit
  (requires (fun h -> live h b))
  (ensures (fun h0 _ h1 -> live h0 b /\ live h1 b /\ modifies_1 b h0 h1 /\ as_value h1 b == z))
\end{lstlisting}

\subsection{Using \lowstar for Proofs of Functional Correctness and Side-Channel Resistance}
\label{sec:poly1305}

This section and the next illustrate our ``high-level
verification for low-level code'' methodology.
Although programming at a low-level, we rely on features like type
abstraction and dependently typed meta-pro\-gramming, to prove our
code functionally correct, cryptographically secure, and free of a
class of side-channels.

We start with Poly1305~\cite{bernstein2005poly1305}, a Message
Authentication Code (MAC) algorithm.\footnote{Implementation bugs in
  Poly1305 are still a practical concern: in 2016 alone, the Poly1305 OpenSSL
  implementation experienced two correctness bugs~\cite{polybug,polybug2}
  and a buffer overflow~\cite{CVE7054}.}
%
%
Unlike \lst$chacha20$, for which
the main property of interest is implementation safety, Poly1305 has a
mathematical definition in terms of a polynomial in the prime field
$\ii{GF}(2^{130}-5)$, against which we prove our code functionally correct.
Relying on correctness, we then prove injectivity lemmas on encodings
of messages into field polynomials, and we finally prove
cryptographic security of a one-time MAC construction for Poly1305,
specifically showing unforgeability against chosen message attacks
(UF1CMA).
This game-based proof involves an idealization step, justified by a
probabilistic proof on paper, following the methodology we explain in
\S\ref{sec:crypto}.

For side-channel resistance, we use type abstraction to ensure that our
code's branching and memory access patterns are secret independent.
This style of \fstar{} coding is advocated by Zinzindohou{\'e} et
al.~\cite{ZBB16}; we place it on formal ground by showing that
it is a sound way of enforcing secret independence at the source level
(\sref{lamstar}) and that our compilation carries such properties to
the level of Clight (\sref{to-clight}).
To carry our results further down, one may validate the output of the C
compiler by relying on recent tools proving side-channel resistance at
the assembly level~\cite{almeida-usenix2016, almeida16fse}.
We sketch our methodology on a small snippet from our
specialized arithmetic (bigint) library upon which we built Poly1305.

\paragraph*{Representing field elements using bigints}
We represent elements of the field underlying Poly1305 as $130$-bit
integers stored in \lowstar buffers of machine integers called
\emph{limbs}.
Spreading out bits evenly across $32$-bit words yields five limbs
$\ell_i$, each holding $26$ bits of significant data.
A ghost function $\kw{eval} = \sum_{i=0}^4 \ell_i\times2^{26\times i}$ maps
each buffer to the mathematical integer it represents.
Efficient bigint arithmetic departs significantly from elementary
school algorithms. Additions, for instance, can be made more efficient
by leveraging the extra $6$ bits of data in each limb to delay carry
propagation.
For Poly1305, a bigint \lst$b$ is in compact form in state \lst$m$
(i.e., \lst$compact m b$) when all its limbs fit in $26$ bits.
Compactness does not guarantee uniqueness of representation as
$2^{130}-5$ and $0$ are the same number in the field but they have two
different compact representations that both fit in $130$ bits---this
is true for similar reasons for the range $[0,5)$.

\paragraph*{Abstracting integers as a side-channel mitigation}
Modern cryptographic implementations are expected to be protected
against side-channel 
attacks~\cite{Kocher1996}.
As such, we aim to show that the branching behavior and memory
accesses of our crypto code are independent of secrets. To
enforce this, we use an abstract type \lst$limb$ to represent limbs,
all of whose operations reveal no information about the contents of
the \lst$limb$, either through its result or through its branching behavior and
memory accesses. For example, rather than providing a
comparison operator, \lst$eq_leak: limb -> limb -> Tot bool$, whose
boolean result reveals information about the input limbs, we use a
masking operation (\lst$eq_mask$) to compute equality securely. Unlike OCaml, \fstar's
equality is not fully polymorphic, being restricted to only those types
that support decidable equality, \lst$limb$ not being among them. 

\begin{lstlisting}[numbers=none]
val v : limb -> Ghost nat  (* limbs only ghostly revealed as numbers *)
val eq_mask: x:limb -> y:limb -> Tot (z:limb{if v x <> v y then v z = 0 else v z = pow2 26 - 1})
\end{lstlisting}

\noindent In the signature above, \lst$v$ is a function that reveals
an abstract \lst$limb$ as a natural number, but only in ghost code---a
style referred to as translucent abstraction~\cite{mumon}. The
signature of \lst$eq_mask$ claims that it returns a zero limb if the
two arguments differ, although computationally relevant code cannot
observe this fact. \ch{bad transition, what's the precise connection}Note,
the number of limbs in a Poly1305 bigint is a
public constant, i.e., \lst$bigint = b:(buffer limb){b.length = 5}$.

\begin{figure*}[t!]\small
\begin{tabular}{c|c}
\begin{lstlisting}[language=fstar]
let normalize (b:bigint) : Stack unit
  (requires (fun m$_0$ -> compact m$_0$ b))
  (ensures (fun m$_0$ () m$_1$ -> compact m$_1$ b /\ modifies$_1$ b m$_0$ m$_1$ /\
             eval m$_1$ b = eval m$_0$ b % (pow2 130 - 5)))
= let h0 = ST.get() in (* a logical snapshot of the initial state *)
  let ones = 67108863ul in (* 2^26 - 1 *)
  let m5   = 67108859ul in (* 2^26 - 5 *)
  let m    = (eq_mask b.(4ul) ones) & (eq_mask b.(3ul) ones)
        & (eq_mask b.(2ul) ones) & (eq_mask b.(1ul) ones) 
        & (gte_mask b.(0ul) m5) in
  b.(0ul) <- b.(0ul) - m5 & m;
  b.(1ul) <- b.(1ul) - b.(1ul) & m; b.(2ul) <- b.(2ul) - b.(2ul) & m;
  b.(3ul) <- b.(3ul) - b.(3ul) & m; b.(4ul) <- b.(4ul) - b.(4ul) & m;
  lemma_norm h0 (ST.get()) b m (* relates mask to eval modulo *) !\label{line:finalize:lemma}!
\end{lstlisting}

&\quad

\begin{lstlisting}
val poly1305_mac:
  tag:nbytes 16ul ->
  len:u32 ->
  msg:nbytes len{disjoint tag msg} ->
  key:nbytes 32ul{disjoint msg key /\
                  disjoint tag key} ->
  Stack unit
(requires (fun m -> msg $\in$ m /\ key $\in$ m /\ tag $\in$ m))
(ensures (fun m$_0$ _ m$_1$ ->
 let r = Spec.clamp m$_0$[sub key 0ul 16ul] in
 let s = m$_0$[sub key 16ul 16ul] in
 modifies {tag} m$_0$ m$_1$ /\
 m$_1$[tag] ==
 Spec.mac_1305 (encode_bytes m$_0$[msg]) r s))
\end{lstlisting}
\end{tabular}
\caption{Unique representation of a Poly1305 bigint (left) and the top-level spec of Poly1305 (right)}
\label{fig:finalize}
\end{figure*}

\paragraph*{Proving \lst$normalize$ correct and side-channel resistant}
The \lst$normalize$ function of \fref{finalize} modifies a compact
bigint in-place to reduce it to its canonical representation. The code
is rather opaque, since it operates by strategically masking each limb
in a secret independent manner. However, its specification clearly
shows its intent: the new contents of the input bigint is the same as
the original one, \emph{modulo $2^{130}-5$}. At
line~\autoref{line:finalize:lemma}, we see a call to a \fstar lemma,
which relates the masking operations to the modular arithmetic in the
specification---the lemma is erased during extraction.

\paragraph*{A top-level functional correctness spec} Using our
bigint library, we implement \lst$poly1305_mac$ and prove it
functionally correct. Its specification (\fref{finalize}, right)
states that the final value of the 16 byte tag (\lst@m$_1$[tag]@) is the
value of \lst$Spec.mac_1305$, a polynomial of the message and the key
encoded as field elements. We use this mathematical specification as a
basis for the game-based cryptographic proofs of constructions built
on top of Poly1305, such as the AEAD construction, described next.



\subsection{Cryptographic Provable-Security Example: AEAD}
\label{sec:crypto}


Going beyond functional correctness, we sketch how we use \lowstar
to do security proofs in the standard model of cryptography, using
``authenticated encryption with associated data'' (AEAD) as a
sample construction. 
AEAD is the main protection mechanism for the TLS record layer; it
secures most Internet traffic.

AEAD has a generic security proof by reduction to two core
functionalities: a stream cipher (such as ChaCha20) and a one-time-MAC
(such as Poly1305).
The cryptographic, game-based argument supposes that these two
algorithms meet their intended \emph{ideal functionalities}, e.g.,
that the cipher is indistinguishable from a random function.
Idealization is not perfect, but is supposed to hold against
computationally limited adversaries, except with small probabilities,
say, $\varepsilon_\mathrm{ChaCha20}$ and $\varepsilon_\mathrm{Poly1305}$.
The argument then shows that the AEAD construction also meets its own
ideal functionality, except with probability, say, $\varepsilon_\mathrm{Chacha20} +
\varepsilon_\mathrm{Poly1305}$.

To apply this security argument to our implementation of AEAD, we need
to encode such assumptions.  To this end, we supplement our real
\lowstar code with ideal \fstar code.
%
%
For example, ideal AEAD is programmed as follows:
\begin{itemize}
\item encryption generates a fresh random ciphertext, and it records
  it together with the encryption arguments in a log.

\item decryption simply looks up an entry in the log that matches its
  arguments and returns the corresponding plaintext, or reports an
  error.
\end{itemize}
These functions capture both confidentiality (ciphertexts do not
depend on plaintexts) and integrity (decryption only succeeds on
ciphertexts output by encryption).
Their behaviors are precisely captured by typing, using pre- and
post-conditions about the ghost log shared between them, and
abstract types to protect plaintexts and keys.
%
%
%
%
We show below the abstract type of keys and the encryption function for
idealizing AEAD.
\begin{lstlisting}[numbers=none]
type entry = cipher * data * nonce * plain
abstract type key = { key: keyBytes; log: if Flag.aead then ref (seq entry) else unit }
let encrypt (k:key) (n:nonce) (p:plain) (a:data)  =
  if Flag.aead then let c = random_bytes !$\ell_c$! in k.log := (c, a, n, p) :: k.log; c
  else encrypt k.key n p a
\end{lstlisting}

A module \li+Flag+ declares a set of abstract booleans (\emph{idealization
flags}) that precisely capture each cryptographic assumption.
For every real functionality that we wish to idealize, we branch on
the corresponding flag. In the code above, for instance we idealize
encryption when \li+Flag.prf+ is set.
%
%
%

This style of programming heavily relies on the normalization capabilities of
\fstar. 
At verification time, flags are kept abstract, so that we verify both the real and ideal versions of the code.
%
At extraction time, we reveal these booleans to be \li+false+, allowing the
\fstar normalizer to drop the \li+then+ branch, and replace the
\li+log+ field with \li+unit+, meaning that both the high-level,
list-manipulating code and corresponding type definitions are erased, leaving
only low-level code from the \li+else+ branch to be extracted.

Using this technique, we verify by typing that our AEAD code,
when using \emph{any} ideal cipher and one-time MAC, perfectly
implements ideal AEAD.
We also rely on typing to verify that our code complies with the
pre-conditions of the intermediate proof steps. Finally, we also prove
that our code does not reuse nonces, a common cryptographic pitfall.

\paragraph{Inlining and Type Abstraction} 

In cryptographic constructions, we often rely on type abstraction to
protect private state that depends on keys and other secrets.

A typical C application, such as OpenSSL, achieves limited type abstraction
as follows. The library exposes a public C header file for its clients,
relying on \li+void *+ and opaque heap allocation functions for type abstraction.
\begin{lstlisting}[numbers=none,language=C]
typedef void *KEY_HANDLE;
void KEY_init(KEY_HANDLE *key);
void KEY_release(KEY_HANDLE key);
\end{lstlisting}
Opportunities for mistakes abound, since the \li+void *+ casts are unchecked.
Furthermore, abstraction only occurs at the public header boundary, not between
internal translation units. Finally, this pattern does not allow the caller to
efficiently allocate the actual key on the stack.

The \lowstar discipline allows the programmer to achieve type abstraction and
modularity, while still supporting efficient stack allocation.
As an example, for computing one-time MACs incrementally, we use an accumulator
that holds the current value of a polynomial computation, which depends on a
secret key.  For cryptographic soundness, we must ensure that no information
about such intermediate values leak to the rest of the code.

To this end, all operations on accumulators are defined in a single
module of the form below---our code is similar but more complex, as it
supports MAC algorithms with different field representations and key
formats, and also keeps track of the functional correctness of the
polynomial computation.

\begin{lstlisting}[numbers=none]
module OneTimeMAC
type elem = lbytes (v accLen) (* intermediate value (representing a a field element) *) 
abstract type key (i:macID) = elem
abstract type accum (i:macID) = elem
(* newAcc allocates on the caller's frame *)
let newAcc (i:macID) : StackInline (accum i) (...) = Buffer.create 0ul accLen 
let extend (i:macID) (key: macKey i) (acc:accum i) (word:elem) : Stack unit (...) = add acc word; mul acc key
\end{lstlisting}

The index \lst$i$ is used to separate multiple instances of MACs; for
instance, it prevents calls to extend an accumulator with the wrong
key. Our type-based separation between different kinds of elements is
purely static. At runtime, the accumulator, and probably the key, are
just bytes on the stack (or in registers), whereas the calls to
\lst$add$ and \lst$mul$ are also likely to be have been inlined in the
code that uses MACs.

The \lst$newAcc$ function creates a new buffer for a given index, initialized to
0. The function returns a pointer to the buffer it allocates. The
\li+StackInline+ effect indicates that the function does need to push a frame
before allocation, but instead allocates in its caller's stack frame.
\kremlin textually inlines the function in its caller's body at every
call-site, ensuring that the allocation performed by \li+newAcc+
indeed happens in the caller's stack frame.
From the perspective of \lowstar, \li+newAcc+ is a function in a
separate module, and type abstraction is preserved.

\ifpagelimits
\newpage
\fi
\section{A formal translation from \lowstar to Clight}
\label{sec:formal}

\newcommand\emf{{\sc emf}$^\star$\xspace}
\newcommand\emfST{{\sc emf}$^\star_{\text {\sc st}}$\xspace} 



Figure~\ref{fig:bigpicture} on page~\pageref{fig:bigpicture} provides
an overview of our translation from \lowstar to CompCert Clight,
starting with \emf, a recently proposed model
of \fstar~\cite{dm4free}; then \lamstar, a formal core of \lowstar
after all erasure of ghost code and specifications; then \cstar, an
intermediate language that switches the calling convention closer to
C; and finally to Clight.
In the end, our theorems establish that: (a) the safety and functional
correctness properties verified at the \fstar level carry on to the
generated Clight code (via semantics preservation), and (b) \lowstar
programs that use the secrets parametrically enjoy the trace
equivalence property, at least until the Clight level, thereby
providing protection against side-channels.



\paragraph*{Prelude: Internal transformations in \emf}
We begin by briefly describing a few internal transformations on \emf,
focusing in the rest of this section on the pipeline from \lamstar to
Clight---the formal details are in the appendix. To express
computational irrelevance, we extend \emf with a primitive \kw{Ghost}
effect. An erasure transformation removes ghost subterms, and we prove
that this pass preserves semantics, via a logical relations
argument. Next, we rely on a prior result~\cite{dm4free} showing
that \emf programs in the \lst$ST$ monad can be safely reinterpreted
in \emfST, a calculus with primitive state. We obtain an instance
of \emfST suitable for \lowstar by instantiating its state type
with \lst$HyperStack.mem$. To facilitate the remainder of the
development, we transcribe \emfST to \lamstar, which is a restriction
of \emfST to first-order terms that only use stack memory, leaving the
heap out of \lamstar, since it is not a particularly interesting
aspect of the proof. This transcription step is essentially
straightforward, but is not backed by a specific proof. We plan to
fill this gap as we aim to mechanize our entire proof in the future.


\nik{Mention non-termination here?}

\subsection{\lamstar: A Formal Core of \lowstar Post-Erasure}
\label{sec:lamstar}

\begin{figure}[t]
\vspace{-1em}
\begin{small}
\[\!\!\!\!\!
  \begin{array}{r@{~}c@{~}l@{}}
    \tau  & ::= & \kw{int} \mid \kw{unit} \mid \{\ls{\fd=\tau}\} \mid \kw{buf}~\tau \mid \alpha\\
    \lv   & ::= & x \mid n \mid () \mid \{\ls{\fd=\lv}\} \mid (b, n, \ls{\fd}) \\
    \lexp & ::= & \elet{x:\tau}{\ereadbuf {\lexp_1}{\lexp_2}}{\lexp} \mid \elet{\_}{\ewritebuf{\lexp_1}{\lexp_2}{\lexp_3}}{\lexp} \\
          & \mid &\elet{x}{\enewbuf{n}{(\lexp_1:\tau)}}{\lexp_2} \mid \esubbuf{\lexp_1}{\lexp_2}  \\
          & \mid & \elet{x:\tau}{\ereadstruct {\lexp_1}}{\lexp} \mid \elet{\_}{\ewritestruct{\lexp_1}{\lexp_2}}{\lexp} \\
          & \mid &\elet{x}{\enewstruct{(\lexp_1:\tau)}}{\lexp_2} \mid \estructfield{\lexp_1}{\fd}  \\
          & \mid &\withframe\;\lexp \mid \epop\;\lexp \mid \eif {\lexp_1}{\lexp_2}{\lexp_3} \\
          & \mid &\elet{x:\tau}{d\;\lexp_1}{\lexp_2} \mid \elet{x:\tau}{\lexp_1}{\lexp_2} \mid \{\ls{\fd=\lexp}\} \mid \lexp.\fd  \mid \lv \\
    \lp   & ::= &\cdot \mid \etlet{d}{\lambda y:\tau_1. \; \lexp : \tau_2}, \lp \\
\end{array} 
\]
\end{small}
\caption{\lamstar syntax}
\label{fig:lamstar-syntax}
\end{figure}

The meat of our formalization of \lowstar begins with \lamstar, a
first-order, stateful language, whose state is structured as a stack
of memory regions. It has a simple calling convention using a
traditional, substitutive $\beta$-reduction rule. Its small-step
operational semantics is instrumented to produce traces that record
branching and the accessed memory addresses. As such, our traces
account for side-channel vulnerabilities in programs based on the
program counter model~\cite{molnar05pcmodel} augmented to track
potential leaks through cache behavior~\cite{barthe-ccs2014}. We
define a simple type system for \lamstar and prove that programs
well-typed with respect to some values at an abstract type produce
traces independent of those values, e.g., our bigint library when
translated to \lamstar is well-typed with respect to an abstract type
of \lst$limb$s and leaks no information about them via their traces.

\paragraph{Syntax} Figure~\ref{fig:lamstar-syntax} shows the syntax
of \lamstar. A program $P$ is a sequence of top-level function
definitions, $d$. We omit loops but allow recursive function definitions.
Values $v$ include constants, immutable records, and buffers $(b, n, [])$ and mutable structures $(b, n, \ls{\fd})$
passed by reference, where $b$ is the address of the buffer or structure, $n$ is the
offset in the buffer, and $\ls{\fd}$ designates the path to the structure field to take a reference of (this path, as a list, can be longer than 1 in the case of nested mutable structures.)
Stack allocated buffers (\kw{readbuf}, \kw{writebuf}, \kw{newbuf}, and
\kw{subbuf}), and their mutable structure counterparts (\kw{readstruct}, \kw{writestruct}, \kw{newstruct}, $\structfield$), are the main feature of the expression language, along
with $\withframe\;\lexp$, which pushes a new frame on the stack for
the evaluation of $\lexp$, after which it is popped (using
$\epop\;\lexp$, an administrative form internal to the calculus).
Once a frame is popped, all its local buffers and mutable structures become inaccessible.

Mutable structures can be nested, and stored into buffers, in both cases without extra indirection. However, the converse is not true, as \lamstar currently does not allow arbitrary nesting of arrays within mutable structures without explicit indirection via separately allocated buffers. We leave such generalization as future work.






\paragraph{Type system} \lamstar types include the base
types \kw{int} and \kw{unit}, record types $\{\ls{\fd=\tau}\}$, buffer
types $\kw{buf}\;\tau$, mutable structure types $\kw{struct}\;\tau$, and abstract types $\alpha$.  The typing
judgment has the form, $\Gamma_P; \Sigma; \Gamma \vdash e : \tau$,
where $\Gamma_P$ includes the function signatures; $\Sigma$ is the
store typing; and $\Gamma$ is the usual context of variables. We elide
the rules, as it is a standard, simply-typed type system. The type
system guarantees preservation, but not progress, since it does not
attempt to account for bounds checks or buffer/mutable structure lifetime. However,
memory safety (and progress) is a consequence of \lowstar typing
and its semantics-preserving erasure to \lamstar.

\paragraph{Semantics} We define evaluation contexts $E$
for standard call-by-value, left-to-right evaluation order. The memory
$H$ is a stack of frames, where each frame maps addresses $b$ to a
sequence of values $\ls{v}$. The \lamstar small-step semantics
judgment has the form $P \vdash (H, \lexp) \rightarrow_{\trace}
(H', \lexp')$, meaning that under the program $P$, configuration $(H,
e)$ steps to $(H', e')$ emitting a trace $\trace$, including
reads and writes to buffer references or mutable structure references, and branching behavior, as
shown below.

\vspace{-0.5cm}
\[
\small
\begin{array}{rl}
\trace ::= \cdot \mid \kw{read}(b, n, \ls{\fd}) \mid \kw{write}(b,
n, \ls{\fd}) \mid \kw{brT} \mid \kw{brF} \mid \trace_1, \trace_2
\end{array}
\]

\begin{figure*}
  \small
  \begin{mathpar}
  \inferrule* [Right=WF]
  {
  }
  {
    \lp \vdash (H, \withframe\;\lexp) \rightarrow_{\cdot} (H;\{\}, \epop\;\lexp)
  }

\inferrule* [Right=Pop]
{
}
{
  \lp \vdash (H;\_, \epop\;\lv) \rightarrow_{\cdot} (H, \lv)
}

\inferrule* [Right=LIfF]
{
\;
}
{
  \lp \vdash (H, \eif{0}{\lexp_1}{\lexp_2}) \rightarrow_\brf (H, \lexp_2)
}

\inferrule* [Right=App]
{
  \lp(f)=\lambda y:\tau_1.\;\lexp_1:\tau_2
}
{
  \lp \vdash (H, \elet{x:\tau}{f\;v}{\lexp}) \rightarrow (H, \elet{x:\tau}{\lexp_1[v/y]}{e})
}

\inferrule* [Right=LRd]
{
  H(b, n+n_1, []) = \lv \\
  \trace = \kw{read}(b, n + n_1, [])
}
{
  \lp \vdash (H, \elet{x}{\ereadbuf{(b,n,[])}{n_1}}{\lexp}) \rightarrow_{\trace} (H, \lexp[\lv/x])
}

\inferrule* [Right=New]
{
  b \notin \kw{dom}(H;h) \quad   h_1 = h[b\mapsto \lv^n] \quad \lexp_1 = \lexp[(b, 0)/x] \\
 \trace = \kw{write}(b, 0), \dots, \kw{write}(b, n - 1)
}
{
  \lp \vdash (H;h, \elet{x}{\enewbuf{n}{(\lv:\tau)}}{\lexp}) \rightarrow_{\trace} (H;h_1, \lexp_1)
}

\end{mathpar}

\caption{Selected semantic rules from \lamstar}
\label{fig:lamstar-sem}
\end{figure*}

Figure~\ref{fig:lamstar-sem} shows selected reduction rules from
\lamstar.
Rule {\sc{WF}} pushes an empty frame on the stack, and rule {\sc{Pop}}
pops the topmost frame once the expression has been evaluated.
Rule {\sc{LIfF}} is standard, except for the trace $\kw{brF}$ recorded
on the transition.
Rule {\sc{App}} is a standard, substitutive $\beta$-reduction.
Rule {\sc{LRd}} returns the value at the $(n + n_1)$ offset in the
buffer at address $b$, and emits a $\kw{read}(b, n + n_1, [])$ event.
Rule {\sc{New}} initializes the new buffer, and emits write events
corresponding to each offset in the buffer.

\paragraph{Secret independence}
A \lamstar program can be written against an interface providing
secrets at an abstract type.
For example, for the abstract type \lst$limb$, one might augment the
function signatures $\Gamma_P$ of a program with an interface for the
abstract type $\Gamma_{\kw{limb}} =$
\lst@eq_mask : limb$^2$ -> limb@, and typecheck a source program
with free \lst$limb$ variables ($\Gamma =$ \lst$secret:limb$),
and empty store typing, using the judgment
$\Gamma_{\kw{limb}}, \Gamma_p; \cdot; \Gamma \vdash \lexp : \tau$.
Given any representation $\tau$ for \lst$limb$, an implementation for
\lst$eq_mask$ whose trace is input independent, and any pair of values
$v_0:\tau, v_1:\tau$, we prove that running $e[v_0/\text{\lst{secret}}]$ and $e[v_1/\text{\lst{secret}}]$
produces identical traces, i.e., the traces reveal no information
about the secret $v_i$. We sketch the formal development next, leaving
details to the appendix.

Given a derivation
$\Gamma_s, \Gamma_P; \Sigma; \Gamma \vdash e : \tau$, let $\Delta$ map
type variables in the interface $\Gamma_s$ to concrete types and let $P_s$ contain
the implementations of the functions (from $\Gamma_s$) that operate on
secrets.
To capture the secret independence of $P_s$, we define a notion of an
\emph{equivalence modulo secrets}, a type-indexed relation for values
($v_1 \equiv_{\tau} v_2$) and memories ($\Sigma \vdash H_1 \equiv
H_2$). Intuitively two values (resp. memories) are equivalent modulo
secrets if they only differ in subterms that have abstract types in
the domain of the $\Delta$ map---we abbreviate ``equivalent modulo
secrets'' as ``related'' below.
We then require that each function $f \in P_s$, when applied in
related stores to related values, always returns related results, while
producing \emph{identical} traces.
Practically, $P_s$ is a (small) library written carefully to ensure
secret independence.

Our secret-independence theorem is then as follows:

\begin{theorem}[Secret independence]
  Given
  \begin{enumerate}
  \item a program well-typed against a secret interface, $\Gamma_s$,
    i.e, $\Gamma_s, \Gamma_P; \Sigma; \Gamma \vdash (H, e) : \tau$,
    
  \item a well-typed implementation of the $\Gamma_s$ interface,
    $\Gamma_s; \Sigma; \cdot \vdash_{\Delta} P_s$, such that $P_s$ is
    equivalent modulo secrets,

  \item a pair $(\rho_1, \rho_2)$ of well-typed substitutions for $\Gamma$,
  \end{enumerate}
  
  then either:
  \begin{enumerate}
    \item 
      both programs cannot reduce further, i.e.
      $P_s, P \vdash (H, e)[\rho_1] \nrightarrow$
      and
      $P_s, P \vdash (H, e)[\rho_2] \nrightarrow$, or
    \item 
      both programs make progress with the same trace, i.e.
      there exists $\Sigma' \supseteq \Sigma, \Gamma' \supseteq \Gamma,
      H', e'$, a pair $(\rho_1', \rho_2')$ of well-typed substitutions
      for $\Gamma'$, and a trace $\trace$ such that

  \begin{enumerate}[i)]
  \item
         $P_s, P \vdash (H, e)[\rho_1] \rightarrow^{+}_{\trace} (H', e')[\rho'_1]$
         and
         $P_s, P \vdash (H, e)[\rho_2] \rightarrow^{+}_{\trace} (H', e')[\rho'_2]$, and
  \item  $\Gamma_s, \Gamma_P; \Sigma'; \Gamma' \vdash (H', e') : \tau$
  \end{enumerate}
  \end{enumerate}
\end{theorem}

\subsection{\cstar: An Intermediate Language}
\label{sec:lamstar-to-cstar}

\label{sec:low-to-c}
We move from \lamstar to Clight in two steps. The \cstar intermediate
language retains \lamstar's explicit scoping structure, but switches
the calling convention to maintain an explicit call-stack of
continuations (separate from the stack memory regions). \cstar also
switches to a more C-like syntax, separates side effect-free
expressions from effectful statements.
    
\vspace{-.7em}
\[
  \begin{array}{rl}
    \cp & ::= \ls{\ecfun fx{\tau}{\tau}{\ls{\cstmt}}} \\
    \cexp & ::= n \mid () \mid x \mid \cexp+\cexp \mid \{\ls{\fd=\cexp}\} \mid \cexp.\fd \mid \eptrfd{\cexp}{\fd} \\
    \cstmt & ::= \evardecl {\tau}x{\cexp} \mid \evardecl{\tau}{x}{\eapply f{\cexp}} \mid \eif{\cexp}{\ls{\cstmt}}{\ls{\cstmt}} \mid \ereturn \cexp \\
    & \mid \{\ls{\cstmt}\} \mid \earray {\tau}xn \mid \evardecl{\tau}{x}{\eread {\cexp}} \mid \ewrite {\cexp}{\cexp} \mid \memset{\cexp}{n}{\cexp} \\
  \end{array}
\]

\noindent
The syntax is unsurprising, with two notable exceptions.
First, despite the closeness to C syntax, contrary to C and similarly
to \lamstar, block scopes are not required for branches of a
conditional statement, so that any local variable or local array
declared in a conditional branch, if not enclosed by a further block,
is still live after the conditional statement.
Second, non-array local variables are immutable after
initialization.


\begin{figure*}
  \small
  \begin{mathpar}
  \inferrule* [Right=Block]
  {
  }
  {
    \cp \vdash (S, V, \{\ls{\cstmt_1}\};\ls{\cstmt_2}) \step (S;(\{\},V,\symhole;\ls{\cstmt_2}), V, \ls{\cstmt_1})
  }

  \inferrule* [Right=Empty]
{
}
{
  \cp \vdash (S; (M, V', E), V, []) \step (S, V', \fplug{E}{()})
}

\inferrule* [Right=CIfF]
{
  \eval{\cexp}{(V)} = 0
}
{
  \cp \vdash (S, V, \eif{\cexp}{\ls{\cstmt_1}}{\ls{\cstmt_2}};\ls{\cstmt}) \step_\brf (S, V, \ls{\cstmt_2};\ls{\cstmt})
}

\inferrule* [Right=Call]
{
  \cp(f)=\ecfuntwo{y}{\tau_1}{\tau_2}{\ls{\cstmt_1}} \\
  \eval{\cexp}{(V)}=v
}
{
  \cp \vdash (S, V, \tau\;x=f\;\cexp; \ls{\cstmt}) \step (S;(\None, V, \tau\;x=\symhole;\ls{\cstmt}), \{\}[y\mapsto v], \ls{\cstmt_1})
}

\inferrule* [Right=CRead]
{
  \eval{\cexp}{(V)} = (b, n, \ls{\fd}) \quad
  \symget(S, (b, n, \ls{\fd})) = v \quad
  \trace = \symread\;(b,n,\ls{\fd})
}
{
  \cp \vdash (S, V, \evardecl{\tau}{x}{\eread \cexp}; \ls{\cstmt}) \step_{\trace} (S, V[x \mapsto v], \ls{\cstmt})
}

\inferrule* [Right=ArrDecl]
{
  \quad \\\\
  S = S'; (M, V, E) \\
  b\not\in S \\
  V' = V[x\mapsto (b, 0, [])]
}
{
  \cp \vdash (S, V, \tau\;x[n]; \ls{\cstmt}) \step (S';(M[b\mapsto \None^n], V, E), V', \ls{\cstmt})
}

\end{mathpar}
\caption{Selected semantic rules from \cstar}
\label{fig:cstar-sem}
\end{figure*}

\paragraph{Operational semantics, in contrast to \lamstar}
A \cstar evaluation configuration $C$ consists of a stack $S$, a
variable assignment $V$ and a statement list $\ls{\cstmt}$ to be
reduced. A stack is a list of frames. A frame $F$ includes frame
memory $M$, local variable assignment $V$ to be restored upon function
exit, and continuation $E$ to be restored upon function exit.
Frame memory $M$ is optional: when it is $\bot$, the frame is called a
``call frame''; otherwise a ``block frame'', allocated whenever
entering a statement block and deallocated upon exiting such block. A
frame memory is just a partial map from block identifiers to value
lists. Each \cstar statement performs at most one function call, or
otherwise, at most one side effect. Thus, \cstar is deterministic.

The semantics of \cstar is shown to the right in
Figure~\ref{fig:cstar-sem}, also illustrating the translation
from \lamstar to \cstar. There are three main differences.
First, \cstar's calling convention (rule {\sc
Call}) shows an explicit call frame being pushed on the stack, unlike
\lamstar's $\beta$ reduction.
Additionally, \cstar expressions do not have side effects and do not
access memory; thus, their evaluation order does not matter and their
evaluation can be formalized as a big-step semantics; by themselves,
expressions do not produce events. This is apparent in rules like {\sc
CIfF} and {\sc CRead}, where the expressions are evaluated atomically
in the premises.
Finally, \kw{newbuf} in \lamstar is translated to an array declaration
followed by a separate initialization. In \cstar, declaring an array
allocates a fresh memory block in the current memory frame, and makes
its memory locations available but uninitialized.  Memory write
(resp. read) produces a $\symwrite$ (resp. $\symread$)
event. $\memset{\cexp_1}{m}{\cexp_2}$ produces $m$ $\symwrite$ events,
and can be used only for arrays.





\paragraph{Correctness of the \lamstar-to-\cstar transformation}

We proved that execution traces are exactly preserved from
\lamstar to \cstar:

\begin{lemma}[\lamstar to \cstar] \label{lem:lamstar-to-cstar}
 Let $\lp$ be a \lamstar program and $\lexp$ be a \lamstar entry point
 expression, and assume that they compile: $\lowtocd(\lp) = \cp$ for
 some \cstar program $\cp$ and $\lowtoc(\lexp) = \ls{\cstmt}; \cexp$
 for some \cstar list of statements $\ls{\cstmt}$ and expression
 $\cexp$.
 
 Let $V$ be a mapping of local variables containing the initial values
 of secrets. Then, the \cstar program $\cp$ terminates with trace
 $\trace$ and return value $\lv$, i.e., $\cp \vdash ([], V,
 \ls{\cstmt}; \kw{return} ~ \cexp) \stackrel{\trace,\ast}{\rightarrow}
 ([], V', \kw{return} ~ \lv)$ if, and only if, so does the \lamstar
 program: $\lp \vdash (\{\}, \lexp[V])
 \stackrel{\trace,\ast}{\rightarrow} (H', \lv)$; and similarly for
 divergence.
\end{lemma}

In particular, if the source \lamstar program is safe, then so is the
target \cstar program. It also follows that the trace equality
security property is preserved from \lamstar to \cstar.
We prove this theorem by bisimulation. In fact, it is enough to prove
that any \lamstar behavior is a \cstar behavior, and flip the diagram
since \cstar is deterministic. That \cstar semantics use
big-step semantics for \cstar expressions complicates the bisimulation
proof a bit because \lamstar and \cstar steps may go out-of-sync at
times. Within the proof we used a relaxed notion of simulation
(``quasi-refinement'') that allows this temporary discrepancy by some
stuttering, but still implies bisimulation.

\subsection{From \cstar to CompCert Clight and Beyond}
\label{sec:to-clight}

CompCert Clight is a deterministic (up to system I/O) subset of C with
no side effects in expressions, and actual byte-level representation
of values. Clight has a realistic formal
semantics \cite{Blazy-Leroy-Clight-09,compcert-url} and tractable
enough to carry out the correctness proofs of our transformations
from \lamstar to C.
More importantly, Clight is the source language of the CompCert compiler
backend, which we can
thus leverage to preserve at least safety and functional correctness properties
of \lowstar programs down to assembly.\footnote{As a subset of C,
Clight can be compiled by any C compiler, but only CompCert provides
formal guarantees.}

Recall that we need to produce an event in the trace whenever a memory
location is read or written, and whenever a conditional branch is
taken, to account for memory accesses and statements in the semantics
of the generated Clight code for the purpose of our noninterference
security guarantees. However, the semantics of CompCert
Clight \emph{per se} produces no events on memory accesses; instead,
CompCert provides a syntactic program annotation mechanism using
no-op \emph{built-in calls}, whose only purpose is to add extra events
in the trace. Thus, we leverage this mechanism by prepending each
memory access and conditional statement in the Clight generated code
with one such built-in call producing the corresponding events.

The main two
differences between \cstar and Clight, which our translation deals
with as described below, are immutable local structures, and scope management
for local variables.

\paragraph{Immutable local structures}
\cstar handles immutable local structures as
first-class values, whereas Clight only supports non-compound
data (integers, floating-points or pointers) as values.

If we naively translate immutable local \cstar structures to C structures in
Clight, then CompCert will allocate them in memory.
This increases the number of memory accesses, which not only
introduces discrepancies in the security preservation proof from \cstar
to Clight, but also introduces significant performance
overhead compared to GCC, which optimizes away structures whose
addresses are never taken.

Instead, we split an immutable structure into the sequence of all its
non-compound fields, each of which is to be taken as a potentially
non-stack-allocated local variable,\footnote{Our benchmark without
this structure erasure runs 20\% slower than with structure erasure,
both with CompCert 2.7.
Without structure erasure, code
generated with CompCert is 60\% slower than with
{\tt gcc -O1}. CompCert-generated code without structure erasure may even
segfault, due to stack overflow, which structure erasure successfully
overcomes.} except for functions that return structures, where, as
usual, we add, as an extra argument to the callee, a pointer to the
memory location written to by the callee and read by the caller.

\paragraph{Local variable hoisting}
Scoping rules for \cstar local arrays are not exactly the same as in
C, in particular for branches of conditional statements. So, it is
necessary to hoist all local variables to function-scope.  CompCert
2.7.1 does support such hoisting but as an unproven elaboration
step. While existing formal proofs (e.g., Dockins'
\cite[\S 9.3]{dockins-phd}) only prove functional
correctness, we also prove preservation of security guarantees, as
shown below.

\paragraph{Proof techniques}
Contrary to the \lamstar-to-\cstar transformation, our subsequent
passes modify the memory layout leading to differences in traces
between \cstar to Clight, due to pointer values. Thus, we need to
address security preservation separately from functional correctness.

   For each pass changing the memory layout, we split it into three
    passes. First, we \emph{reinterpret} the program by replacing each
    pointer value in event traces with the function name and recursion
    depth of its function call, the name of the corresponding local
    variable, and the array index and structure field name within this
    local variable. Then, we perform the actual transformation and
    prove that it exactly preserves traces in this new ``abstract''
    trace model. Finally, we reinterpret the generated code back to
    concrete pointer values.  We successfully used this technique to
    prove functional correctness and security preservation for
    variable hoisting.
    
  For each pass that adds new memory accesses, we split it into
    two passes. First, a reinterpretation pass produces new events
    corresponding to the provisional memory accesses (without actually
    performing those memory accesses). Then, this pass is followed by
    the actual trace-preserving transformation that goes back to the
    non-reinterpreted language but adds the actual memory accesses
    into the program.
    We successfully used this technique to prove functional
    correctness and security preservation for structure return, where
    we add new events and memory accesses whenever a \cstar function
    returns a structure value.

In both cases, we mean \emph{reinterpretation} as defining a new
language with the same syntax and small-step semantic rules except
that the produced traces are different, and relating executions of
the same program in the two languages. There, it is easy to prove
functional correctness, but for security preservation, we need to
prove an invariant on two small-step executions of the same program
with different secrets, to show that two equal pointer values in event
traces coming from those two different executions will actually turn
into two equal abstract pointer values in the reinterpreted language.

Our detailed functional correctness and security preservation proofs
from \lamstar to Clight can be found in the appendix.

\paragraph{Towards verified assembly code}
We conjecture that our reinterpretation techniques can be generalized to most passes
of CompCert down to assembly. 
While we leave such generalization as future work, some guarantees
from C to assembly can be derived by instrumenting CompCert \cite{barthe-ccs2014}
and LLVM \cite{DBLP:conf/popl/ZhaoNMZ12,DBLP:conf/pldi/ZhaoNMZ13,almeida-usenix2016} 
and turning them into \emph{certifying} (rather than certified) compilers where
security guarantees are statically rechecked on the compiled code
through translation validation, thus re-establishing them
independently of source-level security proofs. In this case, rather
than being fully preserved down to the compiled code,
\lowstar-level proofs are still useful to \emph{practically} reduce the
risk of failures in translation validation.
\ch{Can we make this much shorter? it talks about applying
  something we didn't show at all to other layers.}

\ifpagelimits
\newpage
\fi

\section{KreMLin: a Compiler from \lowstar to C}
\label{sec:impl}

\subsection{From \lowstar to Efficient, Elegant C}

As explained previously, \lamstar is the core of \lowstar, post
erasure. Transforming \lowstar into \lamstar proceeds in several
stages. First, we rely on \fstar's existing normalizer and erasure and
extraction facility (similar to features in Coq~\citep{Letouzey08}),
to obtain an ML-like AST for \lowstar terms. Then, we use our new tool
KreMLin that transforms this AST further until it falls within the
\lamstar subset formalized above. KreMLin then performs the \lamstar
to \cstar transformation, followed by the \cstar to C transformation
and pretty-printing to a set of C files. KreMLin generates C11 code
that may be compiled by GCC; Clang; Microsoft's C compiler or
CompCert. We describe the main transformations performed by KreMLin,
beyond those formalized in \sref{formal}, next.

\paragraph{Structures by value}
We described earlier (\sref{structs}) our \lowstar struct library that grants the
programmer fine-grained control over the memory layout, as well as
mutability of interior fields. As an alternative, \kremlin supports immutable,
by-value structs. Such structures, being pure, come with no liveness proof obligations.
The performance of the generated C code, however, is less
predictable: in many cases, the C compiler will optimize and pass such structs
by reference, but on some ABIs (x86), the worst-case scenario may be costly.

Concretely, the \fstar programmer uses tuples and inductive
types. Tuples are monomorphized into specialized inductive
types. Then, inductive types are translated into idiomatic C code:
single-branch inductive types (e.g., records) become actual C structs,
inductives with only constant constructors become C enums, and other
inductives becomes C tagged unions, leveraging C11 anonymous unions
for syntactic elegance. Pattern matches become, respectively,
switches, let-bindings, or a series of cascading if-then-elses.

\paragraph{Whole-program transformations}
\kremlin perform a series of whole-program transformations. First, the
programmer is free to use parameterized type abbreviations. \kremlin substitutes
an application of a type abbreviation with its definition, since C's \li+typedef+ does
not support parameters. (C++11 alias templates would support this use-case.)
Second, \kremlin recursively inlines all \li+StackInline+ functions, as required for
soundness (cf. \sref{crypto}).
Third, \kremlin performs a reachability analysis. Any function that is not
reachable from the \li+main+ function or, in the case of a library, from a
distinguished API module, is dropped. This is essential for generating
palatable C code that does not contain unused helper functions used only for
verification.
Fourth, \kremlin supports a concept of ``bundle'', meaning that several \fstar
modules may be grouped together into a single C translation unit, marking all of
the functions as \li+static+, except for those reachable via the distinguished API
module. This not only makes the code much more idiomatic, but also triggers a
cascade of optimizations that the C compiler is unable to perform across
translation units.

\paragraph{Going to an expression language}
\fstar is, just like ML, an expression language. Two transformations are
required to go to a statement language: \emph{stratification} and
\emph{hoisting}. Stratification places buffer
allocations, assignments and conditionals in statement
position before going to \cstar. Hoisting, as discussed in
\sref{to-clight}, deals with the discrepancy between C99 block scope
and \lowstar \li[language={}]{with_frame}; a buffer allocated under a \li+then+
branch must be hoisted to the nearest enclosing \li+push_frame+,
otherwise its lifetime would be shortened by the resulting C99
block after translation.

\paragraph{Readability}
KreMLin puts a strong emphasis on generating readable C, in the hope that
security experts not familiar with \fstar can review the generated C code.
Names are preserved; we use \li+enum+ and \li+switch+ whenever possible;
functions that take \li+unit+ are compiled into functions with no parameters;
functions that return \li+unit+ are compiled into \li+void+-returning functions.
The internal architecture relies on an abstract C AST and what we believe is a
correct C pretty-printer.

\paragraph{Implementation}
KreMLin represents about 10,000 lines of OCaml, along with a minimal set of
primitives implemented in a few hundred lines of C. After \fstar has extracted and
erased the AEAD development, KreMLin takes less than a second to generate the
entire set of C files. The implementation of KreMLin is optimized for
readability and modularity; there was no specific performance concern in this
first prototype version. KreMLin was designed to support multiple backends; we
are currently implementing a WebAssembly backend to provide verified, efficient
cryptographic libraries for the web.

\subsection{Integrating \kremlin's Output}

\kremlin generates a set of C files that have
no dependencies, beyond a single \li+.h+ file and C11 standard headers, meaning
\kremlin's output can be readily integrated into an existing source tree.

To allow code sharing and re-use, programmers may want to generate a shared
library, that is, a \li+.dll+ or \li+.so+ file that can be distributed along
with a public header (\li+.h+) file. The programmer can achieve this by writing
a distinguished API module in \fstar, exposing only carefully-crafted function
signatures. As exemplified earlier (\fref{chacha20-both}), the translation is
predictable, meaning the programmer can precisely control, in \fstar, what
becomes, in C, the library's public header. The
bundle feature of \kremlin then generates a single C file for the library;
upon compiling it into a shared object, the only visible symbols are those
exposed by the programmer in the header file.

We used this approach for our \haclstar library. Our public header file exposes
functions that have the exact same signature as their counterpart in the NaCL
library. If an existing binary was compiled against NaCL's public header file,
then one can configure the dynamic linker to use our \haclstar library instead,
without recompiling the original program (using the infamous ``LD preload trick'').

The functions exposed by the library comply with the C ABI for the chosen toolchain.
This means that one may use the library from a variety of programming languages,
relying on foreign-function interfaces to interoperate. One popular approach is
to generate bindings for the C library \emph{at run-time} using the ctypes
and the \li+libffi+~\cite{libffi} libraries. This is an approach leveraged by
languages such as JavaScript, Python or OCaml, and requires no recompilation.

An alternative is to write bindings by hand, which allows for better performance
and control over how data is transformed at the boundary, but requires writing
and recompiling potentially error-prone C code. This is the historical way of
writing bindings for many languages, including OCaml. We plan to have \kremlin
generate these bindings automatically. We used this approach in miTLS,
effectively making it a mixed C/OCaml project. We intend to eventually lower all
of miTLS into \lowstar.

\section{Building Verified \lowstar Libraries and Applications}
\label{sec:moreexamples}
\newcommand\fixme[1]{{\color{red}{#1}}}
\begin{table}[ht] \centering
\begin{tabular}{|l|r|r|r|r|}
  \hline
  \multicolumn{1}{|c|}{Codebase} & \multicolumn{1}{c|}{LoC} & \multicolumn{1}{c|}{C LoC} & \multicolumn{1}{c|}{\%annot} & \multicolumn{1}{c|}{Verif. time} \\
  \hline
  \lowstar standard library& 8,936  &  N/A   &  N/A   & 8m \\\hline
  \haclstar                & 6,050  & 11,220 &  28\%  & 12h \\
  miTLS AEAD               & 13,743 & 14,292 & 56.5\% & 1h 10m \\
  \hline
\end{tabular}
\caption{Evaluation of verified \lowstar libraries and applications (time reported on an Intel Core E5 1620v3 CPU)}
\label{tab:applications}
\end{table}

In this section, we describe two examples (summarized in
Table~\ref{tab:applications}) that show how \lowstar can be used to
build applications that balance complex verification goals with high
performance.
First, we describe \haclstar{}, an efficient library of cryptographic
primitives that are verified to be memory safe, side-channel
resistant, and, where there exists a simple mathematical
specification, functionally correct.
%
%
Then, we show how to use \lowstar for type-based cryptographic
security verification by implementing and verifying the AEAD
construction in the Transport Layer Security (TLS) protocol.
We show how this \lowstar library can be integrated within miTLS, 
an \fstar implementation of TLS that is compiled to OCaml.

\subsection{\haclstar: A Fast and Safe Cryptographic Library}
\label{sec:haclstar}


In the wake of numerous security vulnerabilities,
\citet{bernstein2012security} argue that libraries like OpenSSL are
inherently vulnerable to attacks because they are too large, offer too many obsolete options, and
expose a complex API that programmers find hard to use securely.
Instead they propose a new cryptographic API called NaCl that uses a
small set of modern cryptographic primitives, such as
Curve25519~\cite{curve25519} for key exchange, the Salsa family of
symmetric encryption algorithms~\cite{bernstein2008salsa20}, which
includes Salsa20 and ChaCha20, and Poly1305 for message
authentication~\cite{bernstein2005poly1305}.
These primitives were all designed to be fast and easy to implement
in a side-channel resistant coding style.
Furthermore, the NaCl API does not directly expose these low-level
primitives to the programmer. Instead it recommends the use of simple
composite functions for symmetric key authenticated encryption
(\texttt{secretbox}/\texttt{secretbox\_open}) 
and for public key authenticated encryption (\texttt{box}/\texttt{box\_open}).

The simplicity, speed, and robustness of the NaCl API has proved popular among 
developers. Its most popular implementation is Sodium~\cite{libsodium}, which has bindings
for dozens of programming languages and is written mostly in C, with a few components in assembly.
An alternative implementation called TweetNaCl~\cite{bernstein2014tweetnacl} seeks
to provide a concise implementation that is both readable and \emph{auditable} for 
memory safety bugs, a useful point of comparison for our work.
With \lowstar, we show how we can take this approach even further by
placing it on formal, machine-checked ground, without compromising
performance.

\paragraph*{A Verified NaCl Library}
We implement the NaCl API, including all its component algorithms, 
in a \lowstar library called \haclstar, mechanically verifying that 
all our code is memory safe, functionally correct,  and side-channel resistant.
The C code generated from \haclstar is ABI-compatible and can be used as a drop-in replacement for
Sodium or TweetNaCl in any application, in C or any other language, that relies on these libraries.
Our code is written and optimized for 64-bit platforms; on 32-bit
machines, we rely on a stub library for performing 64x64-bit 
multiplications and other 128-bit operations.

We implement and verify four cryptographic primitives:
ChaCha20, Salsa20, Poly1305, and Curve25519, and then use them to
build three cryptographic constructions: AEAD, \texttt{secretbox} 
and \texttt{box}. 
For all our primitives, we prove that our stateful optimized code
matches a high-level functional specification written in \fstar.
These are new verified implementations. Previously, \citet{saw-cryptol} used SAW
and Cryptol to verify C and Java implementations of Chacha20,
Salsa20, Poly1305, AES, and ECDSA. Using a different methodology, \citet{vale}
verifies an assembly version of Poly1305.
Curve25519 has been verified before:
\citet{chen2014verifying} verified an optimized 
low-level assembly implementation using an SMT solver;
\citet{ZBB16} wrote and verified a high-level
library of three curves, including Curve25519, in F* and generated an
OCaml implementation from it.
Our verified Curve25519 code explores a third direction by targeting 
reference C code that is both fast and readable.

A companion paper currently under review~\cite{haclstar} is entirely devoted to the \haclstar
library, and contains an in-depth evaluation of the proof methodology, several
new algorithms that were verified since the present paper was written, along
with a more comprehensive performance analysis.

\begin{table}[h]
  \footnotesize
\begin{tabular}{|l|r|r|r|r||r|}
  \hline
  \multicolumn{1}{|c|}{Algorithm} & \multicolumn{1}{c|}{\haclstar} & \multicolumn{1}{c|}{Sodium} & \multicolumn{1}{c|}{TweetNaCl} &\multicolumn{1}{c||}{OpenSSL} & \multicolumn{1}{c|}{eBACS Fastest}   \\
  \hline
  ChaCha20      & 6.17 cy/B  & 6.97 cy/B  & - & 8.04 cy/B & 1.23 cy/B\\
  Salsa20       & 6.34 cy/B  & 8.44 cy/B  & 15.14 cy/B & - & 1.39 cy/B\\
  Poly1305      & 2.07 cy/B  & 2.48 cy/B  & 32.32 cy/B & 2.16 cy/B & 0.68 cy/B \\
  Curve25519    & 157k cy/mul & 162k cy/mul  & 1663k cy/mul & 359k cy/mul & 145k cy/mul\\
  \hline
  AEAD-ChaCha20-Poly1305          & 8.37 cy/B & 9.60 cy/B  & - & 8.53 cy/B  & \\
  SecretBox     & 8.43 cy/B & 11.03 cy/B  & 50.56 cy/B & -  &  \\
  Box           & 18.10 cy/B & 20.97 cy/B  & 149.22 cy/B & -  &  \\
  \hline
\end{tabular}
\caption{Performance in CPU cycles: 64-bit \haclstar,  64-bit Sodium (pure C, no assembly), 
32-bit TweetNaCl, 64-bit OpenSSL (pure C, no assembly), and the fastest assembly implementation
included in eBACS SUPERCOP. All code was compiled using \texttt{gcc -O3} optimized and
run on a 64-bit Intel Xeon CPU E5-1630. Results are averaged over 1000 measurements, each
processing a random block of $2^{14}$ bytes; Curve25519 was averaged over 1000 random key-pairs.}
\label{tab:haclperf}
\end{table}

\begin{table}[h]
  \footnotesize
  \begin{tabular}{|l|r|r|r|}
    \hline
    \multicolumn{1}{|c|}{Algorithm} & \multicolumn{1}{c|}{\haclstar} & \multicolumn{1}{c|}{OpenSSL} & \multicolumn{1}{c|}{CNG} \\
    \hline
    Curve25519 \hspace{2em} & 17700 mul/s ($\sigma=246$) & 8033 mul/s ($\sigma=120$) & 7490 mul/s ($\sigma=114$) \\
    \hline
  \end{tabular}
  \caption{Performance in operations per second: 64-bit \haclstar, 64-bit
  OpenSSL (assembly disabled) and Microsoft's ``Crypto New Generation'' (CNG)
  library on a 64-bit Windows 10 machine. These results were obtained by writing
  an OpenSSL engine that calls back to either \haclstar, CNG, or OpenSSL
  itself (so as to include the overhead of going through a pluggable engine).
  The \li+speed ecdhx25519+ command runs multiplications for 10s, then counts
  the number of multiplications performed. We show the average over 10 runs of
  this command. The machine is a desktop machine with a 64-bit Intel Xeon CPU
  E5-1620 v2 nominally clocked at 3.70Ghz.}
\end{table}

\paragraph*{Performance}
Table~\ref{tab:haclperf} compares the performance of \haclstar to Sodium, TweetNaCl,
and OpenSSL by running each primitive on a 16KB input; we chose this size since it
corresponds to the maximum record size in TLS and represents a good balance
between small network messages and large files. 
We report averages over 1000 iterations expressed in cycles/byte.
For Curve25519, we measure the time taken for one call to scalar multiplication.
For comparison with state-of-the-art assembly implementations, for each primitive,
we also include the best performance for any implementation (assembly or C)
included in the eBACS SUPERCOP benchmarking framework.\footnote{\url{https://bench.cr.yp.to/supercop.html}} 
These fastest implementation are typically in architecture-specific assembly.

We performed these tests on a variety of 64-bit Intel CPUs (the most
popular desktop configuration) and these performance numbers were
similar across machines. To confirm these measurements, we also ran
the full eBACS SUPERCOP benchmarks on our code, as well as the OpenSSL
\texttt{speed} benchmarks, and the results closely mirrored
Table~\ref{tab:haclperf}. However, we warn the performance numbers 
could be quite different on (say) 32-bit ARM platforms.

We observe that for ChaCha20, Salsa20, and Poly1305, \haclstar achieves comparable performance 
to the optimized C code in OpenSSL and Sodium and significantly better performance than 
TweetNaCl's concise C implementation.
Assembly implementations of these primitives are about 3-4 times faster; they typically
rely on CPU-specific vector instructions and careful hand-optimizations.

Our Curve25519 implementation is about the same speed as Sodium's
64-bit C implementation (\texttt{donna\_c64}) and an order of magnitude
faster than TweetNaCl's 32-bit code.  It is also significantly faster
than OpenSSL because even 64-bit OpenSSL uses a Curve25519 implementation 
that was optimized for 32-bit integers, whereas the implementations in 
Sodium and \haclstar take advantage of the 64x64-bit multiplier available 
on Intel's 64-bit platforms.
The previous \fstar implementation of Curve25519~\cite{ZBB16} 
running in OCaml was not optimized for performance; it
is more than 100x slower than \haclstar.
The fastest assembly code for Curve25519 on eBACS is the one verified
in ~\cite{chen2014verifying}. This implementation is only 1.08x faster
than our C code, at least on the platform on which we tested, which
supported vector instructions up to 256 bits.  We anticipate that the
assembly code may be significantly faster on platforms that support
larger 512-bit vector instructions.

AEAD and \texttt{secretbox} essentially amount to a ChaCha20/Salsa20 cipher sequentially
followed by Poly1305, and their performance reflects the sum of the two primitives.
Box uses Curve25519 to compute a symmetric key, which it then uses to encrypt 
a 16KB input. Here, the cost of symmetric encryption dominates over Curve25519.

In summary, our measurements show that \haclstar is as fast as (or faster than)
state-of-the-art C crypto libraries and within a small factor of hand-optimized assembly code.
This finding is not entirely unexpected, since we wrote our \lowstar code 
by effectively porting the fastest C implementations to \fstar, and any
algorithmic optimization that is implemented in C can, in principle,
be written (and verified) in \lowstar.
What is perhaps surprising is that we get good performance even though
our \lowstar code, and consequently the generated C, heavily relies on 
functional programming patterns such as tail-recursion, and even though
we try to write generic compact code wherever possible, rather than trying 
to mimic the verbose inlined style of assembly code.
We find that modern compilers like GCC and CLANG are able to optimize
our code quite well, and we are able to benefit from their advancements,
without having to change our coding style.
Where needed, KreMLin helps the C compiler by inserting attributes like
\texttt{const}, \texttt{static} and \texttt{inline} that act as optimization hints.

\paragraph*{Balancing Trust and Performance}
All the above performance numbers were obtained with GCC-6 with most
architecture-specific optimizations turned on
(\texttt{-march=native}).  Consequently, any bug in GCC or its plugins
could break the correctness and security guarantees we proved in
\fstar for our source code.  For example, GCC has an auto-vectorizer
that significantly improves the performance of our ChaCha20 and
Salsa20 code in certain use cases, but does so by substantially
changing its structure to take advantage of the parallelism provided
by SIMD vector instructions. To avoid trusting this powerful 
but unverified mechanism, and for more consistent results across platforms,
we turned off auto-vectorization (\texttt{-fno-tree-vectorize}) 
for the numbers in Table~\ref{tab:haclperf}. For similar reasons, 
we turned off link-time optimization (\texttt{-fno-lto}) since 
it relies on an external linker plugin, and can change the 
semantics of our library every time it is linked with a new application.

Ideally, we would completely remove the burden of trust on the C
compiler by moving to CompCert, but at significant performance
cost. Our Salsa20 and ChaCha20 code incurs a relatively modest 3x
slowdown when compiled with CompCert 3.0 (with \texttt{-O3}). 
However, our Poly1305 and Curve25519 code incurs a 30-60x slowdown,
which makes the use of CompCert impractical for our library.
We anticipate that this penalty will reduce as CompCert
improves, and as we learn how to generate C code that would be
easier for CompCert to optimize. For now, we continue to use GCC and
CLANG and comprehensively test the generated code using third-party
tools. For example, we test our code against other implementations,
and run all the tests packaged with OpenSSL. We also 
test our compiled code for side-channel leaks using tools like 
DUDECT.\footnote{\url{https://github.com/oreparaz/dudect}}

\paragraph*{PneuTube: Fast encrypted file transfer}
Using \haclstar, we can build a variety of high-assurance security
applications directly in \lowstar.
PneuTube is a \lowstar program that securely transfers files from a host $A$ to a
host $B$ across an untrusted network.
Unlike classic secure channel protocols like TLS and SSH, PneuTube is
\emph{asynchronous}, meaning that if $B$ is offline, the file may be
cached at some untrusted cloud storage provider and retrieved later. 

PneuTube breaks the file into \emph{blocks} and encrypts each block
using the \texttt{box} API in \haclstar (with an optimization that
caches the result of Curve25519).
It also protects file metadata, including the file name and
modification time, and it hides the file size by padding the file
before encryption to a user-defined size.
We verify that our code is memory-safe, side-channel resistant, and that
it uses the I/O libraries correctly (e.g., it only reads or writes
a file or a socket between calling open and close).

PneuTube's performance is determined by a combination of the crypto
library, disk access (to read and write the file at each end) and  network I/O. 
Its aynchronous design is particularly rewarding on high-latency
network connections, but even when transferring a 1GB file from one
TCP port to another on the same machine, PneuTube takes just 6s.
In comparison, SCP (using SSH with ChaCha20-Poly1305) takes 8 seconds.

\subsection{Cryptographically Secure AEAD for miTLS}
\label{sec:aead}

We use our cryptographically secure AEAD library (\S\ref{sec:crypto})
within miTLS~\cite{mitls}, an existing implementation of TLS in \fstar.
In a previous verification effort, AEAD encryption was idealized as a
cryptographic assumption (concretely realized using bindings to
OpenSSL) to show that miTLS implements a secure authenticated channel.
However, given vulnerabilities such as CVE-2016-7054, this AEAD
idealization is a leap of faith that can undermine security when the
real implementation diverges from its ideal behavior.


We integrated our verified AEAD construction within miTLS at two
levels~\cite{record}.
First, we replace the previous AEAD idealization with a module that
implements a similar ideal interface but translates the state and
buffers to \lowstar representations.
This reduces the security of TLS to the PRF and MAC idealizations in AEAD.
We integrate AEAD at the C level by substituting
the OpenSSL bindings with bindings to the C-extracted version of AEAD.
This introduces a slight security gap, as a small adapter that
translates miTLS bytes to \lowstar buffers and calls into AEAD in C is
not verified.
We confirm that miTLS with our verified AEAD interoperates with
mainstream implementations of TLS 1.2 and TLS 1.3 on ChaCha20-Poly1305
ciphersuites.

\section{Related Work}
\label{sec:related}

\ch{Do we need to relate to Fiat Crypto.
  Independent parallel work.
  \url{https://github.com/mit-plv/fiat-crypto}
  \url{https://people.csail.mit.edu/jgross/personal-website/papers/2017-fiat-crypto-pldi-draft.pdf}
}

\ch{Idris also compiles to C. Want to relate?
  \url{http://docs.idris-lang.org/en/latest/reference/codegen.html}}

Many approaches have been proposed for verifying the functional
correctness and security of efficient low-level code.
A first approach is to build verification frameworks for C using
verification condition generators and SMT solvers~\cite{Kirchner:2015,
  cohen2009vcc, verifast}.
While this approach has the advantage of being able to verify existing
C code, this is very challenging: one needs to deal with the
complexity of C and with any possible optimization trick in the book.
Moreover, one needs an expressive specification language and escape
hatches for doing manual proofs in case SMT automation fails.
So others have deeply embedded C, or C-like languages, into proof
assistants such as Coq~\cite{beringer2015verified,
  Appel15, ChenWSLG16} and Isabelle~\cite{WinwoodKSACN09,
  Schirmer2006} and built program logics and verification
infrastructure starting from that.
This has the advantage of using the full expressive power of the proof
assistant for specifying and verifying properties of low-level programs.
This remains a very labor-intensive task though, because C programs
are very low-level and working with a deep embedding is often cumbersome.
Acknowledging that uninteresting low-level reasoning was a determining
factor in the size of the seL4 verification effort~\cite{Klein09sel4:formal},
\citet{GreenawayLAK14, GreenawayAK12} have recently
proposed sophisticated tools for automatically abstracting the
low-level C semantics into higher-level monadic specifications to ease reasoning.
We take a different approach: we give up on verifying existing C code
and embrace the idea of writing low-level code in a subset of C
shallowly embedded in \fstar{}.
This shallow embedding has significant advantages in terms of reducing
verification effort and thus scaling up verification to larger programs.
This also allows us to port to C only the parts of an \fstar program
that are a performance bottleneck, and still be able to verify the
complete program.

\ch{Can we use the \%annot column to quantitatively support the
  ``reducing verification effort'' claim?}



Verifying the correctness of low-level cryptographic code is
receiving increasing attention~\cite{Appel15, cryptol-s2n,
  beringer2015verified}.
The verified cryptographic applications we have written in \lowstar
and use for evaluation in this paper are an order of magnitude larger
than most previous work.
Moreover, for AEAD we target not only functional correctness, but also
cryptographic security.


In order to prevent the most devastating low-level attacks, several researchers
have advocated dialects of C equipped with type systems for memory safety
\cite{condit2007dependent, jim2002cyclone, Tarditi16}.
Others have designed new languages with type
systems aimed at low-level programming, including for instance linear
types as a way to deal with memory
management~\cite{amani2016cogent, matsakis2014rust}. One drawback is
the expressiveness limitations of such type systems:
once memory safety relies on more complex invariants than these type
systems can express, compromises need to be made, in terms of
verification or efficiency.
\lowstar can perform arbitrarily sophisticated reasoning to establish
memory safety, but does not enjoy the benefits of efficient decision
procedures~\cite{rust} and currently cannot deal with concurrency.

We are not the first to propose writing efficient and verified C code
in a high-level language. LMS-Verify~\cite{lms-verify} recently
extended the LMS meta-programming framework for Scala with
support for lightweight verification. Verification
happens at the generated C level, which has the advantage of taking
the code generation machinery out of the TCB, but has the disadvantage
of being far away from the original source code.

Bedrock~\cite{chlipala2013bedrock} is a generative meta-programming
tool for verified low-level programming in Coq.
The idea is to start from assembly and build up structured code
generators that are associated verification condition generators.
The main advantage of this ``macro assembly language'' view of
low-level verification is that no performance is sacrificed while
obtaining some amount of abstraction.
One disadvantage is that the verified code is not portable.

\ch{Could also compare to Ironclad / Ironfleet, Compiler from Dafny to
  x86 -- translation validation ... their high-level specs are at
  least checked before translating; they do some crypto too}

Our companion paper ``Implementing and Proving
the TLS 1.3 Record Layer''~\cite{record} is available online.
It describes a cryptographic model and
proof of security for AEAD using a combination of \fstar verification
and meta-level cryptographic idealization arguments. To make the point that
verified code need not be slow, the paper mentions that the
AEAD implementation can be ``extracted to C using an experimental backend
for \fstar'', but makes no further claims about this backend. The
current work introduces the design, formalization,
implementation, and experimental evaluation of this C backend for \fstar.
%


\comment{
old ideas of compiling ML to Rust, compare to optimizing compilers
(e.g. OCaml), etc. esp. how we compare to OCaml (extra bit, no pointer
arithmetic).
}\ch{Not sure what this is about}


\section{Conclusion}
This paper advocates a new methodology for carrying out high-level
proofs on low-level code. By embedding a low-level language and memory
model within \fstar, the programmer not only enjoys sophisticated
proofs but also gets to write their low-level code in a more modular
style, using features functional programmers take for granted,
including recursion and type abstraction. Our toolchain, relying on
partial evaluation and the latest advances in C compilers, shows that
we can write code in a style suitable for verification \emph{and}
enjoy the same performance as hand-written C code.

We are currently making progress in three different directions. First,
continuing our integration of AEAD within miTLS, we aim to port the
miTLS protocol layer to \lowstar, in order to get an entire verified,
TLS library in C. Second, parts of our toolchain are unverified. We
plan to formalize and verify using \fstar parts of the \kremlin tool,
notably the \lamstar to \cstar transformation. Third, we are working
on embedding assembly instructions within \lowstar, allowing us to
selectively optimize our code further towards closing the performance
gap that still remains relative to architecture-specific, hand-written
assembly routines.

\section*{Acknowledgments}
We thank the anonymous reviewers for their excellent reviews. We also thank
Abhishek Anand and Mike Hicks, for useful feedback and discussion which helped
shape the work presented here, as well as Armaël Guéneau, for his work on a
mechanized proof.
J-K. Zinzindohou\'e and K. Bhargavan received funding from the
European Research Council (ERC) under the European Union’s Horizon
2020 research and innovation programme (grant agreement no. 683032 - CIRCUS).
C. Hri\c{t}cu was in part supported by the European Research Council
under ERC Starting Grant SECOMP (715753).


\iflong
\newpage
\appendix

\section{Big-stepping a small-step semantics}

Actual observable behaviors will not account for the detailed
sequence of small-step transitions taken. Given an execution first
represented as the sequence of its transition steps from the initial
state, we follow CompCert to derive an observable behavior by only
characterizing termination or divergence and collecting the event
traces, thus erasing all remaining information about the execution
(number of transition steps, sequence of configurations, etc.)  by
\emph{big-stepping} the small-step semantics as shown in
Figure~\ref{fig:bigstep}.

\begin{figure}[!htbp]
  \begin{mathpar}
  \arraycolsep=1.4pt\def\arraystretch{3.2}

\inferrule*
{
  s_0 \stackrel{t_0}{\rightarrow} s_1 \stackrel{t_1}{\rightarrow} \dots \stackrel{t_{n-2}}{\rightarrow} s_{n-1}\stackrel{t_{n-1}}{\rightarrow} s_n \\
  s_0 ~ \text{initial} \\
  s_n ~ \text{final with return value} ~ r \\
  t = t_0 ; t_1 ; \dots ; t_{n-2} ; t_{n-1}
}
{
  \kw{Terminates}(t, r)
}

\\

\inferrule*
{
  s_0 \stackrel{t_0}{\rightarrow} s_1 \stackrel{t_1}{\rightarrow} \dots \\
  s_0 ~ \text{initial} \\
  T = t_0 ; t_1 ; \dots
}
{
  \kw{Diverges}(T)
}

\\

\inferrule*
{
  s_0 \stackrel{t_0}{\rightarrow} s_1 \stackrel{t_1}{\rightarrow} \dots \stackrel{t_{n-2}}{\rightarrow} s_{n-1}\stackrel{t_{n-1}}{\rightarrow} s_n \\
  s_0 ~ \text{initial} \\
  s_n ~ \text{not final} \\
  t = t_0 ; t_1 ; \dots ; t_{n-2} ; t_{n-1}
}
{
  \kw{GoesWrong}(t)
}
\end{mathpar}
\caption{Big-stepping a small-step semantics}
\label{fig:bigstep}
\end{figure}

\section{\cstar and \lamstar Definition}

Notations used in the document are summarized in Figure \ref{fig:notations}.
Function name $f$ and variable name $x$ are of different syntax classes. A term
is closed if it does not contain unbound variables (but can contain function
names). The grammar of \cstar and \lamstar are listed in Figure \ref{fig:cstar-syntax}
and \ref{fig:lowstar-syntax} respectively. \cstar syntax is defined in such a way
that \cstar expressions do not have side effects (but can fail to evaluate because
of e.g. referring to a nonexistent variable). Locations, which only appear during reduction, consist
of a block id, an offset and a list of field names (a ``field path''). The
``getting field address'' syntax $\eptrfd{e}{fd}$ is for constructing a pointer
to a field of a struct pointed to by pointer $e$.

In \lamstar syntax, buffer allocation, buffer write and function application (as well as mutable struct allocation, mutable struct write) are
distinctive syntax constructs (not special cases of let-binding). In this way we
force effectful operations to be in let-normal-form, to be aligned with \cstar (\cstar
does not allow effectfull expressions because of C's nondeterministic expression
evaluation order). Let-binding and anonymous let-binding are also distinctive
syntax constructs, because they need to be translated into different \cstar
constructs. Locations and $\epop\;le$ only appear during reduction.

The operational semantics of \cstar is listed in Figure \ref{fig:cstar-expr-eval}
and \ref{fig:cstar-stmts-reduction}. Because C expressions do not have a
deterministic evaluation order, in \cstar we use a mixed big-step/small-step
operational semantics, where \cstar expressions are evaluated with big-step
semantics defined by the evaluation function (interpreter) $\eval{e}{(p,V)}$,
while \cstar statements are evaluated with small-step semantics. Definitions used in
\cstar semantics are summarized in Figure \ref{fig:cstar-semantics-defs}. A \cstar
evaluation configuration $C$ consist of a stack $S$, a variable assignment $V$
and a statement list $ss$ to be reduced. A stack is a list of frames. A frame
$F$ includes frame memory $M$, variable assignment $V$ to be restored upon
function exit, and continuation $E$ to be restored upon function exit. Frame
memory $M$ is optional: when it is none, the frame is called a ``call frame'';
otherwise a ``block frame''. A frame memory is just a partial map from block ids
to value lists.

Both \cstar and \lamstar reductions generate traces that include memory read/write with
the address, and branching to true/false. Reduction steps that don't have these
effects are silent.

\begin{figure}[!htbp]
\begin{center}
  \begin{tabularx}{\columnwidth}{lRlR}
    $\ls{a}$ & list &
    $\option{a}$ & option $a$ \\
    $\None$ & None &
    $\Some{a}$ & Some a \\
    $n$ & integer &
    $x$ & variable name \\
    $f$ & function name &
    $fd$ & field name \\
    $a\sympartial b$ & partial map & $\{\}$ & empty map \\
    $\{x\mapsto a\}$ & singleton map & $m[x\mapsto a]$ & map update \\
    $[]$ & empty list & $a;b$ & list concat or cons \\
    $\subst{x}{a}{b}$ & substitute $a$ for $x$ in $b$ & & \\
  \end{tabularx}
\end{center}
\caption{Notations}
\label{fig:notations}
\end{figure}

\subsection{\cstar Definition}

The following are the definitions of the syntax and operational semantics of \cstar.

\begin{figure}[!htbp]
\begin{center}
  \begin{tabularx}{\columnwidth}{rlR}
    $p ::= $ & & program \\
      & $\ls d$                      & series of declarations \\[1.2mm]

    $d ::= $ & & declaration \\
      & $\ecfun fxtt{ss}$                & top-level function \\
    & $\evardecl txv$               & top-level value \\
    [1.2mm]

    $ss ::= $ & & statement lists \\
    & $\ls{s}$                  &  \\
    [1.2mm]

    $s ::= $ & & statements \\
    & $\evardecl txe$                  & immutable variable declaration \\
    & $\earray txn$                & array declaration \\
    & $\memset{e}{n}{e}$                & memory set \\
    & $\evardecl{t}{x}{\eapply fe}$    & application \\
    & $\evardecl{t}{x}{\eread e}$      & read \\
    & $\ewrite ee$               & write \\
    & $\eif{e}{ss}{ss}$ & conditional \\
    & \{\stmts\} &  block \\
    & e & expression \\
    & \ereturn e & return \\
    [1.2mm]

    $e ::= $ & & expressions \\
    & $n$    & integer constant \\
    & $()$ & unit value \\
      & $x$                           & variable \\
      & $e_1+e_2$                        & pointer add ($e_1$ is a pointer and $e_2$ is an $\kw{int}$) \\
    & $\{\ls{fd=e}\}$ & struct \\
    & $e.fd$ & struct field projection \\
    & $\eptrfd{e}{fd}$ & struct field address ($e$ is a pointer) \\
    & $loc$                        & location \\
    [1.2mm]

    $loc ::= $ & & locations \\
    & $(b, n, \ls{fd})$ &  \\

  \end{tabularx}
\end{center}
\caption{\cstar Syntax}
\label{fig:cstar-syntax}
\end{figure}

\begin{figure}[!htbp]
\begin{small}
\begin{center}
  \begin{tabularx}{\columnwidth}{rlR}

    $v ::= $ & & values \\
    & $n$    & constant \\
    & $()$ & unit value \\
    & $\{\ls{fd=v}\}$ & constant struct \\
    & $loc$                        & location \\
    [1.2mm]

    $E ::=$ & & evaluation ctx (plug expr to get stmts) \\
    & $\symhole; ss$ & discard returned value \\
    & $t\;x=\symhole; ss$ & receive returned value \\
    [1.2mm]

    $F ::=$ & & frames \\
    & $(\None, V, E)$ & call frame \\
    & $(\Some{M}, V, E)$ & block frame \\
    [1.2mm]

    $M ::=$ & & memory \\
    & $b\sympartial \ls{\option{v}}$ & map from block id to list of optional values \\
    [1.2mm]

    $V ::=$ & & variable assignments \\
    & $x\sympartial v$ & map from variable to value \\
    [1.2mm]

    $S ::=$ & & stack \\
    & $\ls{F}$ & list of frames \\
    [1.2mm]

    $C ::=$ & & configuration \\
    & $(S, V, ss)$ &  \\
    [1.2mm]

    $l ::=$ & & label \\
    & $\symread\;loc$ & read \\
    & $\symwrite\;loc$ & write \\
    & $\brt$ & branch true \\
    & $\brf$ & branch false \\

  \end{tabularx}
\end{center}
\end{small}
\caption{\cstar Semantics Definitions}
\label{fig:cstar-semantics-defs}
\end{figure}

\begin{figure*}[!htbp]
\begin{small}
\begin{flushleft}
\fbox{$\eval{e}{(p, V)}=v$}
\end{flushleft}
\begin{mathpar}

\inferrule* [Right=Var]
{
  V(x) = v
}
{
  \eval{x}{(p,V)}=v
}

\quad
\quad
\quad

\inferrule* [Right=PtrAdd]
{
  \eval{e_1}{(p,V)} = (b, n, []) \\
  \eval{e_2}{(p,V)} = n' \\
}
{
  \eval{e_1+e_2}{(p,V)}=(b, n+n', [])
}

\\

\inferrule* [Right=PtrFd]
{
  \eval{e}{(p,V)} = (b, n, \ls{fd}) \\
}
{
  \eval{\eptrfd{e}{fd}}{(p,V)}=(b, n, \ls{fd};fd)
}

\quad
\quad
\quad

\inferrule* [Right=GVar]
{
  x\not\in V \\
  p(x) = v
}
{
  \eval{x}{(p,V)}=v
}

\quad
\quad
\quad
\quad

\inferrule* [Right=NonMutStruct]
{
  \;
}
{
  \eval{\{\ls{fd=e}\}}{(p,V)}=\{\ls{fd=\eval{e}{(p,V)}}\}
}

\\

\inferrule* [Right=Proj]
{
  \eval{e}{(p,V)}=\{\ls{fd=v}\} \\
  \{\ls{fd=v}\}(fd')=v'
}
{
  \eval{e.fd'}{(p,V)}=v'
}

\quad
\quad
\quad
\quad

\inferrule*[Right=Val]
{
}{
  \eval{v}{(p,V)} = v
}
\end{mathpar}
\end{small}
\caption{\cstar Expression Evaluation}
\label{fig:cstar-expr-eval}
\end{figure*}

\begin{figure*}[!htbp]
\begin{small}
\begin{flushleft}
\fbox{$p \vdash C\step_l C'$}
\end{flushleft}
\begin{mathpar}
\inferrule* [Right=VarDecl]
{
  \eval{e}{(p,V)}=v
}
{
  p \vdash (S, V, t\;x=e; ss) \step (S, V[x\mapsto v], ss)
}

\quad
\quad
\quad
\quad
\quad
\quad

\inferrule* [Right=ArrDecl]
{
  S = S'; (M, V, E) \\
  b\not\in S
}
{
  p \vdash (S, V, t\;x[n]; ss) \step (S';(M[b\mapsto \None^n], V, E), V[x\mapsto (b, 0, [])], ss)
}

\\

\inferrule* [Right=Memset]
{
  \eval{e_1}{(p, V)} = (b, n, []) \\
  \eval{e_2}{(p, V)} = v \\
  \symset(S, (b, n, []), v^m) = S'
}
{
  p \vdash (S, V, \memset{e_1}{m}{e_2}; ss) \step_{\symwrite\;(b,n,[]), \dots, \symwrite\;(b,n+m-1,[])} (S', V, ss)
}

\\

\inferrule* [Right=Read]
{
  \eval{e}{(p, V)} = (b, n, \ls{fd}) \\
  \symget(S, (b, n, \ls{fd})) = v
}
{
  p \vdash (S, V, t\;x=*[e]; ss) \step_{\symread\;(b,n,\ls{fd})} (S, V[x\mapsto v], ss)
}

\quad
\quad
\quad
\quad

\\

\inferrule* [Right=Write]
{
  \eval{e_1}{(p, V)} = (b, n, \ls{fd}) \\
  \eval{e_2}{(p, V)} = v \\
  \symset(S, (b, n, \ls{fd}), v) = S'
}
{
  p \vdash (S, V, *e_1=e_2; ss) \step_{\symwrite\;(b,n,\ls{fd})} (S', V, ss)
}

\quad
\quad
\quad
\quad
\quad

\inferrule* [Right=Ret]
{
  \eval{e}{(p,V)}=v
}
{
  p \vdash (S;(\None, V',E), V, \ereturn\;e; ss) \step (S, V', \fplug{E}{v})
}

\\

\inferrule* [Right=Call]
{
  p(f)=\ecfuntwo{y}{t_1}{t_2}{ss_1} \\
  \eval{e}{(p,V)}=v
}
{
  p \vdash (S, V, t\;x=f\;e; ss) \step (S;(\None, V, t\;x=\symhole;ss), \{\}[y\mapsto v], ss_1)
}

\quad
\quad
\quad
\quad

\inferrule* [Right=RetBlk]
{
  \eval{e}{(p,V)}=v
}
{
  p \vdash (S;(M, V',E), V, \ereturn\;e; ss) \step (S, \{\}, \ereturn\;v)
}

\\

\inferrule* [Right=Expr]
{
  \eval{e}{(p,V)} = v
}
{
  p \vdash (S, V, e; ss) \step (S, V, ss)
}

\quad
\quad
\quad
\quad

\inferrule* [Right=Empty]
{
  \;
}
{
  p \vdash (S; (M, V', E), V, []) \step (S, V', \fplug{E}{()})
}

\\

\inferrule* [Right=Block]
{
  \;
}
{
  p \vdash (S, V, \{ss_1\};ss_2) \step (S;(\{\},V,\symhole;ss_2), V, ss_1)
}

\\

\inferrule* [Right=IfT]
{
  \eval{e}{(p,V)} = n \\
  n\not=0
}
{
  p \vdash (S, V, \eif{e}{ss_1}{ss_2};ss) \step_\brt (S, V, ss_1;ss)
}

\quad
\quad
\quad

\inferrule* [Right=IfF]
{
  \eval{e}{(p,V)} = n \\
  n=0
}
{
  lp \vdash (S, V, \eif{e}{ss_1}{ss_2};ss) \step_\brf (S, V, ss_2;ss)
}

\end{mathpar}
\end{small}
\caption{\cstar Configuration Reduction}
\label{fig:cstar-stmts-reduction}
\end{figure*}

\clearpage

\subsection{\lamstar Definition}

The following are the definitions of the syntax and operational semantics of \lamstar.

\begin{figure}[!htbp]
\begin{center}
  \begin{tabularx}{\columnwidth}{rlR}
    $lp ::= $ & & program \\
      & $\ls{ld}$                      & series of declarations \\[1.2mm]

    $ld ::= $ & & declaration \\
      & $\etlet{x}{\lambda y:t.\;le:t}$                & top-level function \\
    & $\etlet{x:t}{v}$               & top-level value \\
    [1.2mm]

    $le ::= $ & & expressions \\
    & $n$    & constant \\
    & $()$ & unit value \\
    & $x$                           & variable \\
    & $\{\ls{fd=le}\}$ & struct as value \\
    & $le.fd$ & immutable struct field projection \\
    & $\estructfield{le}{fd}$           & sub-structure: mutable structure field projection \\
    & $\esubbuf{le}{le}$                & sub-buffer \\
    & $\eif {le}{le}{le}$            & conditional \\
    & $\elet{x:t}{le}{le}$                   & let-binding \\
    & $\elet{\_}{le}{le}$                   & anonymous let-binding \\
    & $\elet{x:t}{f\;le}{le}$                   & application \\
    & $\elet{x}{\enewbuf{n}{(le:t)}}{le}$                 & new buffer \\
    & $\elet{x:t}{\ereadbuf {le}{le}}{le}$  & read from buffer \\
    & $\elet{\_}{\ewritebuf{le}{le}{le}}{le}$              & write to buffer \\
    & $\elet{x}{\enewstruct{(le:t)}}{le}$                 & new mutable structure \\
    & $\elet{x:t}{\ereadstruct{le}}{le}$                & read from mutable structure \\
    & $\elet{\_}{\ewritestruct{le}}{le}$               & write to mutable structure \\
    & $\withframe\;le$ & with-frame \\
    & $\epop\;le$ & pop frame \\
    & $loc$                        & location \\
    [1.2mm]

  \end{tabularx}
\end{center}
\caption{\lamstar Syntax}
\label{fig:lowstar-syntax}
\end{figure}

\clearpage

\begin{figure}[!htbp]
  $$\begin{array}{rrcl}

    \multicolumn{4}{l}{\textrm{values}} \\
    & lv & ::= & n \mid () \mid \{\ls{fd=lv}\} \mid loc \\
    \\
    \multicolumn{4}{l}{\textrm{evaluation contexts}} \\
    & LE & ::= & \symhole \mid LE.fd \mid \estructfield{LE}{fd} \mid \{\ls{fd=lv};fd=LE;\ls{fd=le}\} \\
    & & \mid & \esubbuf{LE}{le} \mid \esubbuf{lv}{LE} \\
    & & \mid & \eif {LE}{le}{le} \\
    & & \mid & \elet{x:t}{LE}{le} \mid \elet{\_}{LE}{le} \\
    & & \mid & \elet{x:t}{f\;LE}{le} \\
    & & \mid & \elet{x}{\enewbuf{n}{(LE:t)}}{le} \\
    & & \mid & \elet{x:t}{\ereadbuf{LE}{le}}{le} \\
    & & \mid & \elet{x:t}{\ereadbuf{lv}{LE}}{le} \\
    & & \mid & \elet{\_}{\ewritebuf{LE}{le}{le}}{le} \\
    & & \mid & \elet{\_}{\ewritebuf{lv}{LE}{le}}{le} \\
    & & \mid & \elet{\_}{\ewritebuf{lv}{lv}{LE}}{le} \\
    & & \mid & \elet{x}{\enewstruct{(LE:t)}}{le} \\
    & & \mid & \elet{x:t}{\ereadstruct{LE}}{le} \\
    & & \mid & \elet{\_}{\ewritestruct{LE}{le}}{le} \\
    & & \mid & \elet{\_}{\ewritestruct{lv}{LE}}{le} \\
    & & \mid & \epop\;LE \\
    \\
    \multicolumn{4}{l}{\textrm{stack}} \\
    & H & ::= & \ls{h} \\
    \\
    \multicolumn{4}{l}{\textrm{stack frame}} \\
    & h & ::= & b\sympartial \ls{v} \\

\end{array}$$
\caption{\lamstar Semantics Definitions}
\label{fig:lowstar-semantics-defs}
\end{figure}

\clearpage

\begin{figure*}[!htbp]
\begin{scriptsize}
\begin{flushleft}
\fbox{$lp \vdash (H, le) \astep_l (H', le')$}\quad\text{and}\quad
\fbox{$lp \vdash (H, le) \step_l (H', le')$}
\end{flushleft}
\begin{mathpar}

\inferrule* [Right=ReadBuf]
{
  H(b, n+n', []) = lv
}
{
  lp \vdash (H, \elet{x}{\ereadbuf{(b,n,[])}{n'}}{le}) \astep_{\symread\;(b,n+n',[])} (H, \subst{x}{lv}{le})
}

\\

\inferrule* [Right=ReadStruct]
{
  H(b, n, \ls{fd}) = lv
}
{
  lp \vdash (H, \elet{x}{\ereadstruct{(b,n,\ls{fd})}}{le}) \astep_{\symread\;(b,n,\ls{fd})} (H, \subst{x}{lv}{le})
}

\\

\inferrule* [Right=App]
{
  lp(f)=\lambda y:t_1.\;le_1:t_2
}
{
  lp \vdash (H, \elet{x:t}{f\;v}{le}) \astep (H, \elet{x:t}{\subst{y}{v}{le_1}}{le})
}

\\

\inferrule* [Right=WriteBuf]
{
  (b,n+n',[])\in H
}
{
  lp \vdash (H, \elet{\_}{\ewritebuf{(b,n,[])}{n'}{lv}}{le}) \astep_{\symwrite\;(b,n+n',[])} (H[(b, n+n',[])\mapsto lv], le)
}

\\

\inferrule* [Right=WriteStruct]
{
  (b,n,\ls{fd})\in H
}
{
  lp \vdash (H, \elet{\_}{\ewritestruct{(b,n,\ls{fd})}{lv}}{le}) \astep_{\symwrite\;(b,n,\ls{fd})} (H[(b, n,\ls{fd})\mapsto lv], le)
}

\\

\inferrule* [Right=Subbuf]
{
  \;
}
{
  lp \vdash (H, \esubbuf{(b,n,[])}{n'}) \astep (H, (b, n+n',[]))
}

\\

\inferrule* [Right=StructField]
{
  \;
}
{
  lp \vdash (H, \estructfield{(b,n,\ls{fd})}{fd'}) \astep (H, (b, n, (\ls{fd}; fd')))
}

\\

\inferrule* [Right=Let]
{
  \;
}
{
  lp \vdash (H, \elet{x:t}{v}{le}) \astep (H, \subst{x}{v}{le})
}

\quad
\quad
\quad

\inferrule* [Right=ALet]
{
  \;
}
{
  lp \vdash (H, \elet{\_}{v}{le}) \astep (H, le)
}

\quad
\quad
\quad
\quad

\inferrule* [Right=Proj]
{
  \{\ls{fd=lv}\}(fd') = lv'
}
{
  lp \vdash (H, \{\ls{fd=lv}\}.fd') \astep (H, lv')
}

\\

\inferrule* [Right=IfT]
{
  n\not=0
}
{
  lp \vdash (H, \eif{n}{le_1}{le_2}) \astep_\brt (H, le_1)
}

\quad
\quad
\quad

\inferrule* [Right=IfF]
{
  n=0
}
{
  lp \vdash (H, \eif{n}{le_1}{le_2}) \astep_\brf (H, le_2)
}

\\

\inferrule* [Right=NewBuf]
{
  b\not\in H
}
{
  lp \vdash (H, \elet{x}{\enewbuf{n}{(lv:t)}}{le}) \astep_{\symwrite\;(b,0,[]),\dots,\symwrite\;(b,n-1,[])} (H[b\mapsto lv^n], \subst{x}{(b, 0,[])}{le})
}

\\

\inferrule* [Right=NewStruct]
{
  b\not\in H
}
{
  lp \vdash (H, \elet{x}{\enewstruct{(lv:t)}}{le}) \astep_{\symwrite\;(b,0,[])} (H[b\mapsto lv], \subst{x}{(b, 0,[])}{le})
}

\\

\inferrule* [Right=WF]
{
  \;
}
{
  lp \vdash (H, \withframe\;le) \astep (H;\{\}, \epop\;le)
}

\quad
\quad
\quad

\inferrule* [Right=Pop]
{
  \;
}
{
  lp \vdash (H;h, \epop\;lv) \astep (H, lv)
}

\\

\inferrule* [Right=Step]
{
  lp \vdash (H, le) \astep_l (H', le')
}
{
  lp \vdash (H, \fplug{LE}{le}) \step_l (H', \fplug{LE}{le'})
}

\end{mathpar}
\end{scriptsize}
\caption{\lamstar Atomic Reduction and Reduction}
\label{fig:lowstar-reduction}
\end{figure*}

\newpage

\section{\lamstar to \cstar Compilation}

The compilation procedure is defined in Figure \ref{fig:lowtoc} as inference rules, which should be read as a function defined by pattern-matching, with earlier rules shadowing later rules. The compilation is a partial function, encoding syntactic constraints on \lamstar programs that can be compiled. For example, compilable \lamstar top-level functions must be wrapped in a $\withframe$ construct.

\begin{figure*}[!htbp]
\begin{scriptsize}
\begin{flushleft}
\fbox{$\lowtoce{le}=e$}\quad\text{and}\quad
\fbox{$\lowtoc{le}=ss$}\quad\text{and}\quad
\fbox{$\lowtocd{ld}=d$}
\end{flushleft}
\begin{mathpar}
\inferrule*
{
  \;
}
{
  \lowtoce{n} = n
}

\quad
\quad

\inferrule*
{
  \;
}
{
  \lowtoce{(b,n,[])} = (b,n,[])
}

\quad
\quad

\inferrule*
{
  \;
}
{
  \lowtoce{\{\ls{fd=le}\}} = \{\ls{fd=\lowtoce{le}}\}
}

\quad
\quad

\inferrule*
{
  \;
}
{
  \lowtoce{x} = x
}

\quad
\quad

\inferrule*
{
  \;
}
{
  \lowtoce{\esubbuf{le_1}{le_2}} = \lowtoce{le_1} + \lowtoce{le_2}
}

\\

\inferrule*
{
  \;
}
{
  \lowtoce{\estructfield{le}{fd}} = \eptrfd{\lowtoce{le}}{fd}
}

\\

\inferrule*
{
  \lowtoce{le}=e \\
  \lowtoc{le_1}=ss
}
{
  \lowtoc{(\elet{x:t}{f\;le}{le_1})} = (t\;x=f(e); ss)
}

\quad
\quad

\inferrule*
{
  \lowtoce{le_i}=e_i \;(i=1,2) \\
  \lowtoc{le}=ss
}
{
  \lowtoc{\elet{x:t}{\ereadbuf{le_1}{le_2}}{le}} = (t\;x=\eread{e_1+e_2};ss)
}

\\

\inferrule*
{
  \lowtoce{le_i}=e_i \;(i=1) \\
  \lowtoc{le}=ss
}
{
  \lowtoc{\elet{x:t}{\ereadstruct{le_1}}{le}} = (t\;x=\eread{e_1};ss)
}

\\

\inferrule*
{
  \lowtoce{le_i}=e_i \;(i=1,2,3) \\
  \lowtoc{le}=ss
}
{
  \lowtoc{(\elet{\_}{\ewritebuf{le_1}{le_2}{le_3}}{le})} = (\ewrite{e_1+e_2}{e_3}; ss)
}

\\

\inferrule*
{
  \lowtoce{le_i}=e_i \;(i=1,2) \\
  \lowtoc{le}=ss
}
{
  \lowtoc{(\elet{\_}{\ewritestruct{le_1}{le_2}}{le})} = (\ewrite{e_1}{e_2}; ss)
}

\\

\inferrule*
{
  \lowtoce{le}=e \\
  \lowtoc{le_1}=ss
}
{
  \lowtoc{(\elet{x}{\enewbuf{n}{(le:t)}}{le_1})} = (\earray txn; \memset{x}{n}{e}; ss)
}

\\

\inferrule*
{
  \lowtoce{le}=e \\
  \lowtoc{le_1}=ss
}
{
  \lowtoc{(\elet{x}{\enewstruct{(le:t)}}{le_1})} = (\earray tx1; \memset{x}{1}{e}; ss)
}

\\

\inferrule*
{
  \;
}
{
  \lowtoc{(\withframe\;le)} = \{\lowtoc{le}\}
}

\quad
\quad

\inferrule*
{
  \lowtoce{le}=e \\
  \lowtoc{le_i}=ss_i \;(i=1,2)
}
{
  \lowtoc{(\eif{le}{le_1}{le_2})} = (\eif{e}{ss_1}{ss_2})
}

\\

\inferrule*
{
  \lowtoce{le}=e \\
  \lowtoc{le_1}=ss
}
{
  \lowtoc{(\elet{x:t}{le}{le_1})} = (t\;x=e; ss)
}

\quad
\quad

\inferrule*
{
  \lowtoce{le}=e \\
  \lowtoc{le_1}=ss
}
{
  \lowtoc{(\elet{\_}{le}{le_1})} = (e; ss)
}

\quad
\quad

\inferrule*
{
  \;
}
{
  \lowtoc{le} = \lowtoce{le}
}

\\

\inferrule*
{
  \;
}
{
  \lowtocd{(\etlet{x:t}{lv})} = (t\;x=\lowtoce{lv})
}

\quad
\quad

\inferrule*
{
  \lowtoc{le}=ss;e
}
{
  \lowtocd{(\etlet{f}{\lambda x:t_1.\;\withframe\;le:t_2})} = \ecfun{f}{x}{t_1}{t_2}{ss; \ereturn{e}}
}

\end{mathpar}
\end{scriptsize}
\caption{\lamstar to \cstar compilation}
\label{fig:lowtoc}
\end{figure*}

\newpage

\section{Bisimulation Proof} \label{section-bisim}

The main results are Theorem \ref{thm-safety} and \ref{thm-bisim}, in terms of some notions defined before them in this section. The two theorems are proved by using the crucial Lemma \ref{lemma-reverse} to ``flip the diagram'', i.e., proving \cstar refines \lamstar by proving \lamstar refines \cstar. The flipping relies on the fact that \cstar is deterministic modulo renaming of block identifiers. An alternative way of determinization to renaming of block identifiers is to have the stream of random coins for choosing block identifiers as part of the configuration (state).

That \cstar semantics use
big-step semantics for \cstar expressions complicates the bisimulation
proof a bit because \lamstar and \cstar steps may go out-of-sync at
times. Within the proof we used a relaxed notion of simulation
(``quasi-refinement'') that allows this temporary discrepancy by some
stuttering, but still implies bisimulation.

The specific relation between a \lamstar configuration and its \cstar counterpart is defined in Definition \ref{def-R} as Relation $R$.  It is defined in terms of a \cstar-to-\lamstar back-translation, listed in Fig \ref{fig:ctolow}. We need a back-translation here instead of the \lamstar-to-\cstar forward-translation as defined before, because a \cstar configuration has a clear call-stack with each frame containing its own variable environment and continuation, while a \lamstar configuration contains one giant \lamstar expression which makes it impossible to recover the call-stack. Hence for relating two configurations in the middle of reduction, only the \cstar-to-\lamstar direction is possible.

\begin{definition}[(Labelled) Transition System] \label{def-trsys}
  A transition system is a 5-tuple $(\Sigma, L, \step, s_0, F)$, where $\Sigma$ is a set of states, L is a monoid which is the set of labels, $\step\;\subseteq\Sigma\times \kw{option}\;L\times \Sigma$ is the step relation, $s_0$ is the initial state, and $F$ is a set of designated final state.
\end{definition}

In the following text, we use $\epsilon$ to denote an empty label (the unit of the monoid $L$), $l$ to range over non-empty (non-$\epsilon$) labels and $\olabel$ to range over possibly empty labels. When the label is empty, we can omit it. The label of multiple steps is the combined label of each steps, using the addition operator of monoid $L$. We define $a\Downarrow_\olabel a'\defeq a\step^*_\olabel a'$ and $a'\in F$.

\begin{definition}[Safety] \label{def-safety}
  A transition system $A$ is safe iff for all $s$ so that $s_0\step^* s$, $s$ is unstuck, where $\unstuck(s)$ is defined as either $s\in F$ or there exists $s'$ so that $s\step s'$.
\end{definition}

\begin{definition}[Refinement] \label{def-refine}
  A transition system $A$ refines a transition system $B$ (with the same label set) by $R$ (or $R$ is a refinement for transition system $A$ of transition system $B$) iff
  \begin{enumerate}
  \item there exists a well-founded measure $|-|_b$ (indexed by a $B$-state $b$) defined on the set of $A$-states $\{a:\Sigma_A\;|\;a\;R\;b\}$;
  \item $a_0\;R\;b_0$, that is, the two initial states are in  relation $R$;
  \item for all $a:\Sigma_A$ and $b:\Sigma_B$ such that $a\;R\;b$,
    \begin{enumerate}
    \item if $a\step_\olabel a'$ for some $a':\Sigma_A$, then there exists $b':\Sigma_B$ and $n$ such that $b\step^n_\olabel b'$ and $a'\;R\;b'$ and that $n=0$ implies that $|a|_b>|a'|_b$;
    \item if $a\in F_A$, then there exists $b':\Sigma_B$ such that $b\Downarrow_\epsilon b'$ and $a\;R\;b'$.
    \end{enumerate}
  \end{enumerate}
  $A$ refines $B$ iff there exists $R$ so that $A$ refines $B$ by $R$.
\end{definition}

\begin{definition}[Bisimulation] \label{def-bisim}
  A transition system $A$ bisimulates a transition system $B$ iff $A$ refines $B$ and $B$ refines $A$.
\end{definition}

\begin{definition}[Transition System of \cstar and \lamstar] \label{def-trsys-cstar}
  $$\begin{array}{rcl}
    \sys_{\cstar}(p,V,ss)&\defeq&(C, \{\ls{l}\}, p\vdash\step, ([], V, ss), \{([], V', \ereturn{e})\}) \\
    \sys_{\lamstar}(lp,le)&\defeq&(\{(H,le)\}, \{\ls{l}\}, lp\vdash\step, ([], le), \{([], lv)\})
  \end{array}$$
\end{definition}

In the following text, we treat label $\symread/\symwrite\;(b,n)$ and $\symread/\symwrite\;(b,n,[])$ as equal, and if we have a \lamstar value or substitution, we freely use it as a \cstar one because coercion from \lamstar value to \cstar value is straight-forward.

\begin{theorem}[Safety] \label{thm-safety}
  For all \lamstar program $lp$, closed expression $le$ and closing substitution $V$, if $\lowtocd{lp}=p$, $\lowtoc{le}=ss$ and $\sys_{\lamstar}(lp,V(le))$ is safe, then $\sys_{\cstar}(p, V, ss)$ is safe.
\end{theorem}
\begin{proof}
Appeal to Lemma \ref{lemma-reverse}, Lemma \ref{lemma-cstar-deter} and Lemma \ref{lemma-back-refine}.
\end{proof}

\begin{theorem}[Bisimulation] \label{thm-bisim}
  For all \lamstar program $lp$, closed expression $le$ and closing substitution $V$, if $\lowtocd{lp}=p$ and $\lowtoc{le}=ss$, then $\sys_{\cstar}(p, V, ss)$ bisimulates $\sys_{\lamstar}(lp,V(le))$.
\end{theorem}
\begin{proof}
Appeal to Corollary \ref{coro-reverse}, Lemma \ref{lemma-cstar-deter} and Lemma \ref{lemma-back-refine}.
\end{proof}

\begin{definition}[Determinism] \label{def-deter}
  A transition system $A$ is deterministic iff for all $s$ so that $s_0\step^* s$, $s\in F$ implies that $s$ cannot take any step, and $s\step_{o_1} s_1$ and $s\step_{o_2} s_2$ implies that $o_1 = o_2$ and $s_1 = s_2$.
\end{definition}

\begin{definition}[Quasi-Refinement] \label{def-quasi-refine}
  A transition system $A$ quasi-refines a transition system $B$ (with the same label set) by $R$ (or $R$ is a quasi-refinement for transition system $A$ of transition system $B$) iff
  \begin{enumerate}
  \item there exists a well-founded measure $|-|_b$ (indexed by a $B$-state $b$) defined on the set of $A$-states $\{a:\Sigma_A\;|\;a\;R\;b\}$;
  \item $a_0\;R\;b_0$, that is, the two initial states are in  relation $R$;
  \item for all $a:\Sigma_A$ and $b:\Sigma_B$ such that $a\;R\;b$,
    \begin{enumerate}
    \item if $a\step_\olabel a'$ for some $a':\Sigma_A$, then there exists $a'':\Sigma_A$, $b':\Sigma_B$ and $n$ such that $a'\step^*_\epsilon a''$ and $b\step^n_\olabel b'$ and $a''\;R\;b'$ and that $n=0$ implies that $|a|_b>|a'|_b$;
    \item if $a\in F_A$, then there exists $b':\Sigma_B$ such that $b\Downarrow_\epsilon b'$ and $a\;R\;b'$.
    \end{enumerate}
  \end{enumerate}
  $A$ quasi-refines $B$ iff there exists $R$ so that $A$ quasi-refines $B$ by $R$.
\end{definition}

\begin{lemma}[Quasi-refine-Refine] \label{lemma-quasi}
  If transition system $A$ is deterministic, then $A$ quasi-refines transition system $B$ implies that $A$ refines $B$.
\end{lemma}
\begin{proof}
  Let $R$ be the quasi-refinement for $A$ of $B$. \\
  Define $R'$ to be: $a\;R'\;b$ iff $\exists n.\;(\exists a'.\;a\step^n_\epsilon a' \eand a'\;R\;b).$ \\
  We are to show that $A$ refines $B$ by $R'$. Unfold Definition \ref{def-refine}. \\
  For Condition 1, define $|a|_b$ to be the minimal of the number $n$ in the definition of $R'$, which uniquely exists. \\
  For Condition 2, we are to show $a_0\;R'\;b_0$. We know that $a_0\;R\;b_0$, so it's obviously true. \\
  For Condition 3(a), we have $a\;R'\;b$ and $a\step_\olabel a'$. \\
  We are to exhibit $b'$ and $n$ so that $b\step^n_\olabel b'$ and $a'\;R'\;b'$ and that $n=0$ implies $|a|_b>|a'|_b$. \\
  From $a\;R'\;b$, we have $a\step^m_\epsilon a''$ and $a''\;R\;b$. \\
  If $m=0$, we know $a\;R\;b$. Because $A$ quasi-refines $B$ by $R$, we have $a'\step^*_\epsilon a_2$ and $b\step^n_\olabel b'$ and $a_2\;R\;b'$ and that $n=0$ implies $|a|_b>|a'|_b$. \\
  Pick $b'$ to be $b'$ and $n$ to be $n$. It suffices to show $a'\;R'\;b'$, which is true because $a'\step^*_\epsilon a_2$ and $a_2\;R\;b'$. \\
  If $m>0$, pick $b'$ to be $b$ and $n$ to be 0. Because $A$ is deterministic, we know $a'\;R'\;b$ with $m-1$ and $|a|_b=m$ and $|a'|_b=m-1$. \\
  For Condition 3(b), we have $a\;R'\;b$ and $a\in F_A$. \\
  We are to exhibit $b'$ such that $b\Downarrow_\epsilon b'$ and $a\;R'\;b'$. \\
  Because $A$ is deterministic and $a\in F_A$, we have $m=0$ and $a\;R\;b$. \\
  Because $A$ quasi-refines $B$ with $R$, we have $b\Downarrow_\epsilon b'$ and $a\;R\;b'$. \\
  Pick $b'$ to be $b'$. It suffices to show $a\;R'\;b'$, which is trivially true. \\
\end{proof}

\begin{lemma}[Refine-Safety] \label{lemma-refine-safety}
  If transition system $A$ refines transition system $B$ by $R$ and for any $a$ and $b$ we have $a\;R\;b$ implies $\unstuck(a)$, then $A$ is safe.
\end{lemma}
\begin{proof}
  From Definition \ref{def-refine} we know $(\exists b.\;a\;R\;b)$ is an invariant of $A$. Hence $\unstuck(a)$ is also an invariant of $A$.
\end{proof}

\begin{lemma}[Deterministic Reverse] \label{lemma-reverse}
  If transition system $A$ is deterministic and transition system $B$ is safe, then $B$ refines $A$ implies that $A$ refines $B$ and $A$ is safe.
\end{lemma}
\begin{proof}
  Appealing to Lemma \ref{lemma-refine-safety}, we will exhibit the refinement for $A$ of $B$ and show that it implies unstuckness. \\
  Let $R$ be the refinement for $B$ of $A$. Define $R'$ to be: \\
  $a\;R'\;b$ iff \\
  $b_0\step^*b\eand((\exists \olabel\;b'.\;b\step_\olabel b'\eand\exists n_1\;a_2\;a_3\;a_4.\;(a\step^*_\epsilon a_2\eor a_2\step^*_\epsilon a)\eand b\;R\;a_2\eand a_2\step^*_\olabel a_3\eand a\step^{n_1}_\olabel a_3\eand a_3\step^*_\epsilon a_4\eand b'\;R\;a_4)\eor(b\in F_B\eand \exists n_2\;a_2\;a_3.\;(a\step^*_\epsilon a_2\eor a_2\step^*_\epsilon a)\eand b\;R\;a_2\eand a_2\step^*_\epsilon a_3\eand a\step^{n_2}_\epsilon a_3\eand a_3\in F_A\eand b\;R\;a_3))$. \\
  Let's first prove the fact (Fact 1) that if $b_0\step^* b$ and $a\step^*_\epsilon a'$ and $b\;R\;a$, then $a\;R'\;b$. \\
  Because $B$ is safe, we know that either $b\step_\olabel b'$ or $b\in F_B$. \\
  In the first case, because $B$ refines $A$ by $R$, we have $a'\step^n_\olabel a''$ and $b\;R\;a''$. \\
  It's easy to show $a\;R'\;b$ by choosing the first disjunct and picking $\olabel,b',a_2, a_4$ to be $\olabel,b',a',a''$. $n_1$ and $a_3$ exist in this case. \\
  In the second case, because $B$ refines $A$ by $R$, we have $a'\Downarrow_\epsilon a''$ and $b\;R\;a''$. \\
  It's easy to show $a\;R'\;b$ by choosing the second disjunt and picking $a_2, a_3$ to be $a', a''$. $n_2$ obviously exists. \\
  \\
  Now we are to show $A$ refines $B$ by $R'$. Unfold Definition \ref{def-refine}. \\
  For Condition 1, define $|a|_b$ to be lexicographic order of two numbers. \\
  The first number is the minimal of the number $n_2$ in the definition of $R'$ if $b$ is a value, which uniquely exists; or 0 otherwise. \\
  The second number is the minimal of the number $n_1$ in the definition of $R'$ if $b$ can take a step, which uniquely exists; or 0 otherwise. \\
  For Condition 2, we are to show $a_0\;R'\;b_0$, which is true because of $b_0\;R\;a_0$ and Fact 1. \\
  For Condition 3(a), we have $a\;R'\;b$ and $a\step_\olabel a'$. \\
  We are to exhibit $b'$ and $n$ such that $b\step^n_\olabel b'$ and $a'\;R\;b'$ and that $n=0$ implies $|a|_b>|a'|_b$. \\
  Unfold $a\;R'\;b$, we have the two disjuncts. \\
  In case $\olabel=l$, only the first disjunt is possible, and we have $b\step_l b'$ and $a'\step^*_\epsilon a_4$ and $b'\;R\;a_4$. \\
  Pick $b',n$ to be $b',1$. From Fact 1, we know $a'\;R'\;b'$. \\
  In case $\olabel=\epsilon$, both disjuncts of $a\;R'\;b$ are possible. \\
  If $a\;R'\;b$ because of the first conjunct, we have $b\step_{\olabel'} b'\eand (a\step^*_\epsilon a_2\eor a_2\step^*_\epsilon a)\eand b\;R\;a_2\eand a_2\step^*_{\olabel'} a_3\eand a\step^{n_1}_{\olabel'} a_3\eand a_3\step^*_\epsilon a_4\eand b'\;R\;a_4$. \\
  We case-analyse on whether $\olabel'$ is $\epsilon$. \\
  \\
  If $\olabel'=l$, let the second component of $|a|_b$ (denoted by $|a|_b.2$) be $m$. We case-analyse on whether $m>1$. \\
  If $m>1$, pick $b',n$ to be $b, 0$ (i.e. do not move on the $B$ side). We need to show $a'\;R'\;b$ and $|a|_b>|a'|_b$. \\
  $a'\;R'\;b$ because according to $m>1$ and $a\step_\epsilon a'$, we know that $a'$ is still before the $l$-label step. \\
  $|a|_b>|a'|_b$ is true because according to $m>1$, it must be the case that $|a'|_b.2=|a|_b.2-1$; and as for $|a'|_b.1$, which represents the minimal number of steps to terminate (or 0 otherwise), taking one step will not increase it. \\
  If $m\leq 1$, we know that $a\step_l a''$ for some $a''$. But we also have $a\step_\epsilon a'$, so this case is impossible because of $A$'s determinism. \\
  \\
  If $\olabel'=\epsilon$, because $B$ is safe and $B$ refines $A$ by $R$, we can step on the $B$ side for finite steps to reach $b_2$ such that $b'\step^*_\epsilon b_2$ and $b_2\;R\;a$ and either $b_2\step_{\olabel''} b_3\eand a\step^+_{\olabel''}a''\eand b_3\;R\;a''$ or $b_2\in F_B\eand a\step^*_\epsilon a''\eand a''\in F_A\eand b_2\;R\;a''$. \\
  In the first case, because $A$ is deterministic, we have $a\step_\epsilon a'\step^*_{\olabel''}a''$. \\
  If $\olabel''=\epsilon$, pick $b'$ to be $b_3$. Because of $a'\step^*_\epsilon a''$ and $b_3\;R\;a''$ and Fact 1, we get $a'\;R'\;b_3$. \\
  If $\olabel''=l$, pick $b'$ to be $b_2$. Since $b\step_\epsilon b'\step^*_\epsilon b_2$, we just need to show that $a'\;R'\;b_2$, which is easy to show by choosing the first disjunct for $R'$ and picking $b',a_2,a_4$ to be $b_3, a, a''$. \\
  In the second case ($b_2\in F_B$), it must be case that $a\step_\epsilon a'\step^*_\epsilon a''\in F_A$. Pick $b'$ to be $b_2$, we need to show $a'\;R'\;b_2$, which is true because $a'\step^*_\epsilon a''$ and $a''\;R'\;b_2$. \\
  \\
  If $a\;R'\;b$ because of the second conjunct, we have $b\in F_B\eand (a\step^*_\epsilon a_2\eor a_2\step^*_\epsilon a)\eand b\;R\;a_2\eand a_2\step^*_\epsilon a_3\eand a\step^{n_2}_\epsilon a_3\eand a_3\in F_A\eand b\;R\;a_3$. \\
  Because $A$ is deterministic, it must be the case that $a\step_\epsilon a'\step^*_\epsilon a_3$. \\
  Pick $b',n$ to be $b,0$. We need to show $a'\;R'\;b$ and $|a|_b>|a'|_b$. $a'\;R'\;b$ because $a'\step^*_\epsilon a_3$ and $a_3\;R'\;b$. $|a|_b>|a'|_b$ because $|a|_b.1>|a'|_b.1$, which is true because $a'$ is one step closer to terminate. \\
  \\
  For Condition 3(b), we have $a\;R'\;b$ and $a\in F_A$. \\
  We are to exhibit $b'$ such that $b\Downarrow_\epsilon b'$ and $a\;R'\;b'$. \\
  If $a\;R'\;b$ because of the second disjunct, we have $b\in F_B\eand a\step^*_\epsilon a_3\eand b\;R\;a_3$. Because $a\in F_A$ and $A$ is deterministic, we know that $a_3=a$\\
  Pick $b'$ to be $b$. $a\;R'\;b$ is true because $b\;R\;a$ and Fact 1. \\
  If $a\;R'\;b$ because of the first disjunct, we have $b\step_\olabel b_2\eand a\step^*_\olabel a_4\eand b_2\;R\;a_4$. \\
  Because $a\in F_A$ and $A$ is deterministic, we know that $a_4=a$ and $\olabel=\epsilon$. \\
  If $b_2\in F_B$, pick $b'$ to be $b_2$. $a\;R'\;b_2$ is true with the same reasoning as before. \\
  Otherwise, because $B$ is safe and $B$ refines $A$ by $R$, we can step $b_2$ for finite steps (because $a$ cannot step and $|b|_a>|b_2|_a$) to have $b_2\step^*_\epsilon b_3\eand b_3\;R\;a$. \\
  Pick $b'$ to be $b_3$. $a\;R'\;b_3$ is true with the same reasoning as before. \\
  \\
  Now we prove that $a\;R'\;b$ implies $\unstuck(a)$ for any $a$ and $b$. \\
  Unfolding $a\;R'\;b$, in both disjuncts we have $a\step^na'$ and $b'\;R\;a'$ for some $b'$. \\
  If $n>0$, $\unstuck(a)$ is obviously true. \\
  If $n=0$, we have $b'\;R\;a$. Because $B$ is safe, we know that either $b'\step b''$ or $b'\in F_B$. \\
  In case $b'\in F_B$, we know $a\Downarrow_\epsilon a_2$. Because $A$ is deterministic, $\unstuck(a)$ is true. \\
  In case $b'\step b''$, because $B$ refines $A$ by $R$, we know $a\step^ka_2$ and $b''\;R\;a_2$ and that $k=0$ implies $|b'|_a>|b''|_a$. \\
  Thus $b'$ can step finite number of $k=0$ steps before hitting the $b'\in F_B$ case or the $k>0$ case, in both of which we have $\unstuck(a)$.
\end{proof}

\begin{corollary}[Deterministic Reverse] \label{coro-reverse}
  If transition system $A$ is deterministic and transition system $B$ is safe, then $B$ refines $A$ implies that $A$ bisimulates $B$ and $A$ is safe.
\end{corollary}

\begin{definition}[Relation $R$] \label{def-R}
  For any $p$ and $lp$, define relation $R_{p,lp}$ as: $(H, le)\;R_{p,lp}\;(S, V, ss)$ iff there exists a minimal $n$ such that $(H, le)\step_{lp}^n(H, le')$ and $(H, le')=\ctolowc{(S, V, ss)}$, where $\ctolowc{(S, V, ss)}\defeq(\mem(S), \unravel(S, V(\ctolow{(\normal{ss}{(p,V)})})))$ and $\mem(S)$ is all the memory parts of $S$ collected together (and requiring that there is no $\None$ in $S$'s memory parts). \\
  $\unravel(S, le)\defeq\foldl\;\unravelframe\;le\;S$ \\
  $\unravelframe((M, V, E), le)\defeq \\ \begin{cases}
    V(\fplug{(\ctolowE{E})}{le}) & \text{if }M=\None\\
    V(\fplug{(\ctolowE{E})}{\epop\;le}) & \text{if }M=\Some{\_}
  \end{cases}$.
\end{definition}

\begin{figure*}[!htbp]
\begin{small}
\begin{flushleft}
\fbox{$\normal{ss}{(p,V)}=ss$}
\end{flushleft}
\begin{mathpar}
\inferrule*
{
  \eval{e}{(p,V)}=v
}
{
  \normal{(t\;x=e;ss)}{(p,V)} = (t\;x=v;ss)
}

\quad
\quad

\inferrule*
{
  \eval{e}{(p,V)}=v
}
{
  \normal{(t\;x=f(e);ss)}{(p,V)} = (t\;x=f(v);ss)
}

\quad
\quad

\inferrule*
{
}
{
  \normal{(\earray{t}{x}{n};ss)}{(p,V)} = (\earray{t}{x}{n};ss)
}

\\

\inferrule*
{
  \eval{e}{(p,V)}=v
}
{
  \normal{(\ereturn{e};ss)}{(p,V)} = (\ereturn{v};ss)
}

\quad
\quad

\inferrule*
{
  \eval{e_i}{(p,V)}=v_i \;(i=1,2)
}
{
  \normal{(t\;x=\eread{e_1}{e_2};ss)}{(p,V)} = (t\;x=\eread{v_1}{v_2};ss)
}

\quad
\quad

\inferrule*
{
  \eval{e_i}{(p,V)}=v_i \;(i=1,2,3)
}
{
  \normal{(\ewrite{e_1+e_2}{e_3};ss)}{(p,V)} = (\ewrite{v_1+v_2}{v_3};ss)
}

\\

\inferrule*
{
  \eval{e}{(p,V)}=v
}
{
  \normal{(e;ss)}{(p,V)} = (v;ss)
}

\quad
\quad

\inferrule*
{
  \eval{e_i}{(p,V)}=v_i \;(i=1,2)
}
{
  \normal{(\memset{e_1}{n}{e_2};ss)}{(p,V)} = (\memset{v_1}{n}{v_2};ss)
}
\end{mathpar}
\end{small}
\caption{Normalize \cstar head expression}
\label{fig:normalize}
\end{figure*}

\clearpage

\begin{figure*}[!htbp]
\begin{scriptsize}
\begin{flushleft}
\fbox{$\ctolowe{e}=le$}\quad\text{and}\quad
\fbox{$\ctolow{ss}=le$}\quad\text{and}\quad
\fbox{$\ctolowd{d}=ld$}\quad\text{and}\quad
\fbox{$\ctolowE{E}=LE$}
\end{flushleft}
\begin{mathpar}
\inferrule*
{
  \;
}
{
  \ctolowe{n} = n
}

\quad
\quad

\inferrule*
{
  \;
}
{
  \ctolowe{(b,n,[])} = (b,n,[])
}

\quad
\quad

\inferrule*
{
  \;
}
{
  \ctolowe{\{\ls{fd=e}\}} = \{\ls{fd=\ctolowe{e}}\}
}

\quad
\quad

\inferrule*
{
}{
  \ctolowe{()} = ()
}

\\

\inferrule*
{
}{
  \ctolowe{x} = x
}

\quad
\quad

\inferrule*
{
  \ctolowe{le_i}=e_i \;(i=1,2)
}
{
  \ctolowe{(e_1+e_2)} = \esubbuf{le_1}{le_2}
}

\quad
\quad

\inferrule*
{
  \ctolowe{le_i}=e_i \;(i=1)
}
{
  \ctolowe{\eptrfd{e_1}{fd}} = \estructfield{le_1}{fd}
}

\\

\inferrule*
{
  \ctolowe{e}=le_1 \\
  \ctolow{ss} = le
}
{
  \ctolow{(t\;x=f(e);ss)} = (\elet{x:t}{f\;le_1}{le})
}

\quad
\quad

\inferrule*
{
  \ctolowe{e}=le_1 \\
  \ctolow{ss} = le
}
{
  \ctolow{(\earray{t}{x}{n};\memset{x}{n}{e};ss)} = (\elet{x}{\enewbuf{n}{(le_1:t)}}{le})
}

\\

\inferrule*
{
  \ctolowe{e}=le_1 \\
  \ctolow{ss} = le \\
  t ~ \text{is a struct type}
}
{
  \ctolow{(\earray{t}{x}{1};\memset{x}{1}{e};ss)} = (\elet{x}{\enewstruct{(le_1:t)}}{le})
}

\quad
\quad

\inferrule*
{
  \ctolowe{e}=le_1 \\
  \ctolow{ss} = le
}
{
  \ctolow{(t\;x=e;ss)} = (\elet{x:t}{le_1}{le})
}

\\

\inferrule*
{
  \ctolowe{e_i}=le_i \; (i=1,2) \\
  \ctolow{ss} = le
}
{
  \ctolow{(t\;x=\eread{e_1+e_2};ss)} = (\elet{\_}{\ereadbuf{le_1}{le_2}}{le})
}

\quad
\quad

\inferrule*
{
  \ctolowe{e_i}=le_i \; (i=1) \\
  \ctolow{ss} = le
}
{
  \ctolow{(t\;x=\eread{e_1};ss)} = (\elet{\_}{\ereadstruct{le_1}}{le})
}

\\

\inferrule*
{
  \ctolowe{e_i}=le_i \; (i=1,2,3) \\
  \ctolow{ss} = le
}
{
  \ctolow{(\ewrite{e_1+e_2}{e_3};ss)} = (\elet{\_}{\ewritebuf{le_1}{le_2}{le_3}}{le})
}

\quad
\quad

\inferrule*
{
  \ctolowe{e_i}=le_i \; (i=1,2) \\
  \ctolow{ss} = le
}
{
  \ctolow{(\ewrite{e_1}{e_2};ss)} = (\elet{\_}{\ewritestruct{le_1}{le_2}}{le})
}

\\

\inferrule*
{
  \ctolow{ss_1}=le_1 \\
  \ctolow{ss} = le
}
{
  \ctolow{(\{ss_1\};ss)} = (\elet{\_}{\withframe\;le_1}{le})
}

\quad
\quad

\inferrule*
{
  \ctolowe{e}=le_1 \\
  \ctolow{ss} = le
}
{
  \ctolow{(e;ss)} = (\elet{\_}{le_1}{le})
}

\\

\inferrule*
{
  \ctolowe{e}=le \\
  \ctolow{ss_i}=le_i \;(i=1,2,3)
}
{
  \ctolow{(\eif{e}{ss_1}{ss_2};ss_3)} = (\elet{\_}{\eif{le}{le_1}{le_2}}{le_3})
}

\quad
\quad

\inferrule*
{
  \ctolowe{e}=le
}
{
  \ctolow{[e]} = le
}

\quad
\quad

\inferrule*
{
  \;
}
{
  \ctolow{[]} = ()
}

\\

\inferrule*
{
  \;
}
{
  \ctolowd{(t\;x=v)} = (\etlet{x:t}{\ctolowe{v}})
}

\quad
\quad

\inferrule*
{
  \ctolow{(ss;e)}=le
}
{
  \ctolowd{(\ecfun{f}{x}{t_1}{t_2}{ss; \ereturn{e}})} = (\etlet{f}{\lambda x:t_1.\;\withframe\;le:t_2})
}

\\

\inferrule*
{
  \ctolow{ss} = le
}
{
  \ctolowE{(\symhole;ss)} = (\elet{\_}{\symhole}{le})
}

\quad
\quad

\inferrule*
{
  \ctolow{ss} = le
}
{
  \ctolowE{(t\;x=\symhole;ss)} = (\elet{x:t}{\symhole}{le})
}
\end{mathpar}
\end{scriptsize}
\caption{\cstar to \lamstar back-translation}
\label{fig:ctolow}
\end{figure*}

\clearpage

\begin{lemma}[\lamstar Refines \cstar] \label{lemma-back-refine}
  For all \lamstar program $lp$, closed expression $le$ and closing substitution $V$, if $\lowtocd{lp}=p$ and $\lowtoc{le}=ss$, then $\sys_{\lamstar}(lp,V(le))$ refines $\sys_{\cstar}(p, V, ss)$.
\end{lemma}
\begin{proof}
  We apply Lemma \ref{lemma-quasi} and \ref{lemma-lowstar-deter}, and prove that $\sys_{\lamstar}(lp,le)$ quasi-refines $\sys_{\cstar}(p, ss)$.
  We pick the relation $R_{p,lp}$ in Definition \ref{def-R} to be the simulation relation and prove $R_{p,lp}$ is a quasi-refinement for $\sys_{\lamstar}(lp,le)$ of $\sys_{\cstar}(p, ss)$. \\
  Unfold Definition \ref{def-quasi-refine}. \\
  For Condition 1, define the well-founded measure $|(H,le)|_{(S, V, ss)}$ (where $(H,le)\;R\;(S,V,ss)$) to be the minimal of the number $n$ in $R$'s definition. \\
  For condition 2, appeal to Lemma \ref{lemma-init}. \\
  Now prove Condition 3(a). Let $(H,le)$ be the \lamstar configuration and $C=(S,V,ss)$ be the \cstar configuration. \\
  We are to exhibit $(H'', le'')$ and $C'$ and $n$ such that $(H',le')\step^*(H'',le'')$ and $C\step^n C'$ and $(H'',le'')\;R\;C'$ and that $n=0$ implies $|(H,le)|_C > |(H',le')|_C$. \\
  For all the cases except \emph{Case Pop}, we pick $(H'', le'')$ to be $(H',le')$ (i.e. do not use the extra flexibility offered by Quasi-Refinement). \\
  Induction on $(H, le)\step(H', le')$. \\
  \\
  \emph{Case Let}: on case $(H, \fplug{LE}{\elet{x:t}{lv}{le}})\step(H, \fplug{LE}{\subst{x}{lv}{le}})$. \\
  We are to exhibit $C'$ such that $(S, V, ss)\step^+C'$ and \\
  $(H, \fplug{LE}{\subst{x}{lv}{le}})\;R\;C'$.\\
  Apply Lemma \ref{lemma-invert-let}. \\
  In the first case, we have $le=V(\ctolow{ss'})$ and $lv=\ctolowe{v}$ and $LE=\unravel(S, \symhole)$, where $v\defeq\eval{e}{(p,V)}$. \\
  The \cstar side runs with nonzero steps to $(S, V[x\mapsto v], ss')$. \\
  Pick $C'$ to be this configuration. \\
  It suffices to show that $(H, \fplug{LE}{\subst{x}{lv}{le}})\;R\;(S, V[x\mapsto v], ss')$. \\
  Appealing to Lemma \ref{lemma-eq-normal}, it suffices to show that $\fplug{LE}{\subst{x}{lv}{le}}=\unravel(S, V[x\mapsto v](\ctolow{ss'}))$, which is true. \\
  In the second case, the \cstar side runs with nonzero steps to $(S', V'[x\mapsto v], ss')$. The proof is the same as the first case. \\
  End of case. \\
  \\
  \emph{Case ALet}: on case $(H, \fplug{LE}{\elet{\_}{lv}{le}})\step(H, \fplug{LE}{le})$. \\
  Appealing to Lemma \ref{lemma-invert-alet}, the proof is similar to the previous case. \\
  End of case. \\
  \\
  \emph{Case App}: on case $(H, \fplug{LE}{\elet{x:t}{f\;lv}{le}})\step(H, \fplug{LE}{\elet{x:t}{\subst{y}{lv}{le_1}}}{le})$ and $lp(f)=\lambda y:t_1.\;le_1:t_2$. \\
  Appealing to Lemma \ref{lemma-invert-app}, we have $ss=(t\;x=f(v);ss')$ and $le=V(\ctolow{ss'})$ and $lv=\ctolowe{v}$ and $LE=\unravel(S, \symhole)$. \\
  Because $\lowtocd{lp}=p$, we know $le_1=\withframe\;le_2$ and $\lowtoc{le_2}=ss_2;e$ and $p(f)=\ecfuntwo{y}{t_1}{t_2}{ss_1}$ and $ss_1=(ss_2;\ereturn\;e)$. \\
  Appealing to Lemma \ref{lemma-lowtoc-ctolow}, we know $\ctolow{ss_1}=le_2$ hence $\ctolow{\{ss_1\}}=le_1$. \\
  Pick $C'$ to be $(S;(\None, V, t\;x=\symhole;ss'), \{y\mapsto v\}, \{ss_1\})$. \\
  It suffices to show that \\
  $(H, \fplug{LE}{\elet{x:t}{\subst{v}{lv}{le_1}}{le}})\;R\;(S;(\None, V, t\;x=\symhole;ss'), \{y\mapsto v\}, \{ss_1\})$, which is true. \\
  End of case. \\
  \\
  \emph{Case Withframe}: \\
  on case $(H, \fplug{LE}{\withframe\;le})\step(H;\{\}, \fplug{LE}{\epop\;le})$. \\
  Appealing to Lemma \ref{lemma-invert-withframe}, we have $ss=\{ss_1\};ss_2$ and $le=V(\ctolow{ss_1})$ and $LE=\unravel(S, V(\ctolowE{(\symhole;ss_2)}))$. \\
  Pack $C'$ to be $(S;(\{\}, V, \symhole;ss_2), V, ss_1)$. \\
  It suffices to show that \\
  $(H;\{\}, \fplug{LE}{\epop\;le})\;R\;(S;(\{\}, V, \symhole;ss_2), V, ss_1)$, which is true. \\
  End of case. \\
  \\
  \emph{Case Newbuf}: on case $(H;h, \fplug{LE}{\elet{x}{\enewbuf{n}{(lv:t)}}{le}})\step_{\symwrite\;(b,0,[]),\dots,\symwrite\;(b,n-1,[])}(H;h[b\mapsto lv^n], \fplug{LE}{\subst{x}{(b,0,[])}{le}})$ and $b\not\in H;h$. \\
  We have $(H;h, \fplug{LE}{\elet{x}{\enewbuf{n}{(lv:t)}}{le}})\;R\;(S, V, ss)$. \\
  We are to exhibit $C'$ so that \\
  $(H;h[b\mapsto lv^n], \subst{x}{(b,0,[])}{le})\;R\;C'$ and $(S, V, ss)\step^+C'$. \\
  Appealing to Lemma \ref{lemma-invert-newbuf}, we have $ss=(\earray{t}{x}{n};\memset{x}{n}{v};ss')$ and $le=V(\ctolow{ss'})$ and $lv=\ctolowe{v}$ and $LE=\unravel(S, \symhole)$ and $S=S';(M, V', E)$. \\
  Pick $C'$ to be $(S';(M[b\mapsto v^n], V', E), V[x\mapsto(b,0, [])], ss')$. \\
  It suffices to show that $(H;h[b\mapsto lv^n], \fplug{LE}{\subst{x}{(b,0,[])}{le}})\;R\;(S';(M[b\mapsto v^n], V', E), V[x\mapsto(b,0,[])], ss')$, which is true. \\
  End of case. \\
  \\
  \emph{Case Newstruct}: on case $(H;h,\fplug{LE}{\elet{x}{\enewstruct{(lv:t)}}{le}})\step_{\symwrite\;(b,0,[])}(H;h[b\mapsto lv], \subst{x}{(b,0,[])}{le})$ and $b\not\in H;h$. \\
  We have $(H;h, \elet{x}{\enewstruct{(lv:t)}}{le})\;R\;(S, V, ss)$. \\
  We are to exhibit $C'$ so that \\
  $(H;h[b\mapsto lv], \fplug{LE}{\subst{x}{(b,0,[])}{le}})\;R\;C'$ and $(S, V, ss)\step^+C'$. \\
  Appealing to Lemma \ref{lemma-invert-newstruct}, we have $ss=(\earray{t}{x}{1};\memset{x}{1}{v};ss')$ and $le=V(\ctolow{ss'})$ and $lv=\ctolowe{v}$ and $LE=\unravel(S, \symhole)$ and $S=S';(M, V', E)$. \\
  Pick $C'$ to be $(S';(M[b\mapsto v], V', E), V[x\mapsto(b,0, [])], ss')$. \\
  It suffices to show that $(H;h[b\mapsto lv], \fplug{LE}{\subst{x}{(b,0,[])}{le}})\;R\;(S';(M[b\mapsto v], V', E), V[x\mapsto(b,0,[])], ss')$, which is true. \\
  End of case. \\
  \\
  \emph{Case Readbuf}: on case \\
  $(H, \fplug{LE}{\elet{x:t}{\ereadbuf{(b,n,[])}{n'}}{le}})\step_{\symread\;(b,n+n',[])}(H, \fplug{LE}{\subst{x}{lv}{le}})$ and $H(b,n+n',[])=lv$. \\
  Appealing to Lemma \ref{lemma-invert-readbuf}, we have $ss=(t\;x=(b,n,[])[n'];ss')$ and $le=V(\ctolow{ss'})$ and $LE=\unravel(S, \symhole)$. \\
  Pick $C'$ to be $(S, V[x\mapsto v], ss')$ where $v=\lowtoce{lv}$. \\
  We know $C\step^+_{\symread\;(b,n+n',[])}C'.$ \\
  It suffices to show that $(H, \fplug{LE}{\subst{x}{lv}{le}})\;R\;(S, V[x\mapsto v], ss')$, which is true because $\subst{x}{lv}{le}=V[x\mapsto v](\ctolow{ss'})$. \\
  End of case. \\
  \\
  \emph{Case Readstruct}: on case \\
  $(H, \fplug{LE}{\elet{x:t}{\ereadstruct{(b,n,\ls{fd})}}{le}})\step_{\symread\;(b,n,\ls{fd})}(H, \fplug{LE}{\subst{x}{lv}{le}})$ and $H(b,n,\ls{fd})=lv$. \\
  Appealing to Lemma \ref{lemma-invert-readstruct}, we have $ss=(t\;x=\eread{(b,n,\ls{fd})};ss')$ and $le=V(\ctolow{ss'})$ and $LE=\unravel(S, \symhole)$. \\
  Pick $C'$ to be $(S, V[x\mapsto v], ss')$ where $v=\lowtoce{lv}$. \\
  We know $C\step^+_{\symread\;(b,n,\ls{fd})}C'.$ \\
  It suffices to show that $(H, \fplug{LE}{\subst{x}{lv}{le}})\;R\;(S, V[x\mapsto v], ss')$, which is true because $\subst{x}{lv}{le}=V[x\mapsto v](\ctolow{ss'})$. \\
  End of case. \\
  \\
  \emph{Case Writebuf}: on case \\
  $(H, \fplug{LE}{\elet{\_}{\ewritebuf{(b,n,[])}{n'}{lv}}{le}})\step_{\symwrite\;(b,n+n',[])}(H[(b,n+n',[])\mapsto lv], \fplug{LE}{le})$ and $(b,n+n',[])\in H$. \\
  Appealing to Lemma \ref{lemma-invert-writebuf}, we have $ss=((b,n,[])[n']=v;ss')$ and $le=V(\ctolow{ss'})$ and $lv=\ctolowe{v}$ and $LE=\unravel(S, \symhole)$. \\
  Pick $C'$ to be $(S', V, ss')$ where $\symset(S, (b,n,[]), v) = S'$. \\
  We know $C\step^+_{\symwrite\;(b,n+n',[])}C'.$ \\
  It suffices to show that $(H[(b,n+n',[])\mapsto lv], \fplug{LE}{le})\;R\;(S', V, ss')$, which is true. \\
  End of case. \\
  \\
  \emph{Case Writestruct}: on case \\
  $(H, \fplug{LE}{\elet{\_}{\ewritestruct{(b,n,\ls{fd})}{lv}}{le}})\step_{\symwrite\;(b,n,\ls{fd})}(H[(b,n,\ls{fd})\mapsto lv], \fplug{LE}{le})$ and $(b,n,\ls{fd})\in H$. \\
  Appealing to Lemma \ref{lemma-invert-writestruct}, we have $ss=(\ewrite{(b,n,\ls{fd})}{v};ss')$ and $le=V(\ctolow{ss'})$ and $lv=\ctolowe{v}$ and $LE=\unravel(S, \symhole)$. \\
  Pick $C'$ to be $(S', V, ss')$ where $\symset(S, (b,n,\ls{fd}), v) = S'$. \\
  We know $C\step^+_{\symwrite\;(b,n,\ls{fd})}C'.$ \\
  It suffices to show that $(H[(b,n,\ls{fd})\mapsto lv], \fplug{LE}{le})\;R\;(S', V, ss')$, which is true. \\
  End of case. \\
  \\
  \emph{Case Subbuf}: on case $(H, \fplug{LE}{\esubbuf{(b,n,[])}{n'}})\step(H, \fplug{LE}{(b,n+n',[])})$. \\
  Pick $C'$ to be $(S, V, ss)$. \\
  Because $(H, \fplug{LE}{\esubbuf{(b,n,[])}{n'}})\;R\;C'$ with some $m$, it must be the case that $m\geq 1$ and $(H, \fplug{LE}{(b,n+n',[])})\;R\;C'$ with $m-1$. \\
  End of case. \\
  \\
  \emph{Case Structfield}: on case $(H, \fplug{LE}{\estructfield{(b,n,\ls{fd})}{fd'}})\step(H, \fplug{LE}{(b,n,(\ls{fd}; fd'))})$. \\
  Pick $C'$ to be $(S, V, ss)$. \\
  Because $(H, \fplug{LE}{\estructfield{(b,n,\ls{fd})}{fd'}})\;R\;C'$ with some $m$, it must be the case that $m\geq 1$ and $(H, \fplug{LE}{(b,n,(\ls{fd};fd'))})\;R\;C'$ with $m-1$. \\
  End of case. \\
  \\
  \emph{Case IfTrue}: on case $(H, \fplug{LE}{\eif{n}{le_1}{le_2}})\step_\brt (H, \fplug{LE}{le_1})$ and $n\not=0$. \\
  Appealing to Lemma \ref{lemma-invert-if}, we have $ss=\eif{e}{ss_1}{ss_2};ss'$ and $\eval{e}{(p,V)}=n$ and $le_i=V(\ctolow{ss_i})$ ($i=1,2$) and $LE=\unravel(S, V(\ctolowE{(\symhole;ss')}))$. \\
  Pick $C'$ to be $(S, V, ss_1;ss')$. \\
  We know $C\step^+_\brt C'.$ \\
  It suffices to show that $(H, \fplug{LE}{le_1})\;R\;(S, V, ss_1;ss')$, which is true. \\
  End of case. \\
  \\
  \emph{Case IfFalse}: on case $(H, \fplug{LE}{\eif{n}{le_1}{le_2}})\step_\brf (H, \fplug{LE}{le_2})$ and $n=0$. \\
  Similar to previous case. \\
  End of case. \\
  \\
  \emph{Case Proj}: on case $(H, \fplug{LE}{\{\ls{fd=lv}\}.fd'})\step(H, \fplug{LE}{lv'})$ and $\{\ls{fd=lv}\}(fd')=lv'$. \\
  Pick $C'$ to be $(S, V, ss)$. \\
  Because $(H, \fplug{LE}{\{\ls{fd=lv}\}.fd'})\;R\;C'$ with some $m$, it must be the case that $m\geq 1$ and $(H, \fplug{LE}{lv'})\;R\;C'$ with $m-1$. \\
  End of case. \\
  \\
  \emph{Case Pop}: on case $(H;h, \fplug{LE}{\epop\;lv})\step(H, \fplug{LE}{lv})$. \\
  Apply Lemma \ref{lemma-invert-pop}. \\
  In the first case, from $LE=\unravel(S', V'(\ctolowE{\symhole;ss'}))$ we know $LE=(\elet{\_}{\symhole}{le})$ and $le=\unravel(S',V'(ss'))$. \\
  pick $C'$ to be $(S',V',ss')$ and $(H'', le'')$ to be $(H, le)$. \\
  Obviously $(H, \fplug{LE}{lv})\step^*(H, le)$. It suffices to show that $(H,le)\;R\;(S',V',ss')$, which is true. \\
  In the second case, pick $C'$ to be $(S',V',\fplug{E}{v})$ and $(H'', le'')$ to be $(H, \fplug{LE}{lv})$. \\
  To suffices to show $(H, \fplug{LE}{lv})\;R\;(S', V', \fplug{E}{v})$, which follows from $LE=\unravel(S',V'(\ctolowE{E}))$. \\
  \\
  For Condition 3(b), because \lamstar and \cstar's values are almost the same (except that \cstar locations have a field-path component), every \lamstar value has an obvious corresponding \cstar value, so condition 3(b) is trivially true.
\end{proof}

\begin{lemma}[\cstar Deterministic] \label{lemma-cstar-deter}
  For all $p$ and $ss$, transition system $\sys_{\cstar}(p, ss)$ is deterministic, modulo renaming of block identifiers.
\end{lemma}

\begin{lemma}[\lamstar Deterministic] \label{lemma-lowstar-deter}
  For all $lp$ and $le$, transition system $\sys_{\lamstar}(lp, le)$ is deterministic, modulo renaming of block identifiers.
\end{lemma}

\begin{lemma}[Init] \label{lemma-init}
  For all \lamstar program $lp$, closed expression $le$ and closing substitution $V$, if $\lowtocd{lp}=p$ and $\lowtoc{le}=ss$, then $([],V(le))\;R_{p,lp}\;([],V,ss)$.
\end{lemma}
\begin{proof}
  Unfold $R$'s definition, it suffices to show: \\
  $([], V(le))\step^*([], V(\ctolow{(\normal{\lowtoc{le}}{(p, \{\})})}))$.
\end{proof}

\begin{lemma}[Equal-Normalize] \label{lemma-eq-normal}
  If $H=\mem(S)$ and $le=\unravel(S, V(\ctolow{ss}))$ and $\normal{ss}{(p,V)}=ss'$, then $(H, le)\step^*(H, \unravel(S, V(\ctolow{ss'})))$.
\end{lemma}

\begin{lemma}[Invert Let] \label{lemma-invert-let}
  If $(H, \fplug{LE}{\elet{x:t}{lv}{le}})\;R\;(S, V, ss)$, then either $ss=(t\;x=e;ss')$ and $le=V(\ctolow{ss'})$ and $lv=\ctolowe{v}$ and $LE=\unravel(S, \symhole)$, where $v\defeq\eval{e}{(p,V)}$ or $S=S';(\None, V', t\;x=\symhole;ss')$ and $ss=\ereturn{v}$ and $le=V'(\ctolow{ss'})$ and $lv=\ctolowe{v}$ and $LE=\unravel(S', \symhole)$.
\end{lemma}

\begin{lemma}[Invert ALet] \label{lemma-invert-alet}
  If $(H, \fplug{LE}{\elet{\_}{lv}{le}})\;R\;(S, V, ss)$, then either $ss=(e;ss')$ and $le=V(\ctolow{ss'})$ and $lv=\ctolowe{v}$ and $LE=\unravel(S, \symhole)$, where $v\defeq\eval{e}{(p,V)}$ or $S=S';(\None, V', \symhole;ss')$ and $ss=\ereturn{v}$ and $le=V'(\ctolow{ss'})$ and $lv=\ctolowe{v}$ and $LE=\unravel(S', \symhole)$.
\end{lemma}

\begin{lemma}[Invert Newbuf] \label{lemma-invert-newbuf}
  If $(H;h, \fplug{LE}{\elet{x}{\enewbuf{n}{(lv:t)}}{le}})\;R\;(S, V, ss)$, then $ss=(\earray{t}{x}{n};\memset{x}{n}{v};ss')$ and $le=V(\ctolow{ss'})$ and $lv=\ctolowe{v}$ and $LE=\unravel(S, \symhole)$ and $S=S';(M, V', E)$.
\end{lemma}

\begin{lemma}[Invert Newstruct] \label{lemma-invert-newstruct}
  If $(H;h, \fplug{LE}{\elet{x}{\enewstruct{(lv:t)}}{le}})\;R\;(S, V, ss)$, then $ss=(\earray{t}{x}{1};\memset{x}{1}{v};ss')$ and $le=V(\ctolow{ss'})$ and $lv=\ctolowe{v}$ and $LE=\unravel(S, \symhole)$ and $S=S';(M, V', E)$.
\end{lemma}

\begin{lemma}[Invert Readbuf] \label{lemma-invert-readbuf}
  If \\
  $(H;h, \fplug{LE}{\elet{x:t}{\ereadbuf{(b,n,[])}{n'}}{le}})\;R\;(S, V, ss)$, then $ss=(t\;x=(b,n,[])[n'];ss')$ and $le=V(\ctolow{ss'})$ and $LE=\unravel(S, \symhole)$.
\end{lemma}

\begin{lemma}[Invert Readstruct] \label{lemma-invert-readstruct}
  If \\
  $(H;h, \fplug{LE}{\elet{x:t}{\ereadstruct{(b,n,\ls{fd})}}{le}})\;R\;(S, V, ss)$, then $ss=(t\;x=\eread{(b,n,\ls{fd})};ss')$ and $le=V(\ctolow{ss'})$ and $LE=\unravel(S, \symhole)$.
\end{lemma}

\begin{lemma}[Invert Writebuf] \label{lemma-invert-writebuf}
  If \\
  $(H;h, \fplug{LE}{\elet{\_}{\ewritebuf{(b,n,[])}{n'}{lv}}{le}})\;R\;(S, V, ss)$, then $ss=((b,n,[])[n']=v;ss')$ and $le=V(\ctolow{ss'})$ and $lv=\ctolowe{v}$ and $LE=\unravel(S, \symhole)$.
\end{lemma}

\begin{lemma}[Invert Writestruct] \label{lemma-invert-writestruct}
  If \\
  $(H;h, \fplug{LE}{\elet{\_}{\ewritestruct{(b,n,\ls{fd})}{lv}}{le}})\;R\;(S, V, ss)$, then $ss=(\ewrite{(b,n,\ls{fd})}{v};ss')$ and $le=V(\ctolow{ss'})$ and $lv=\ctolowe{v}$ and $LE=\unravel(S, \symhole)$.
\end{lemma}

\begin{lemma}[Invert App] \label{lemma-invert-app}
  If $(H, \fplug{LE}{\elet{x:t}{f\;lv}{le}})\;R\;(S, V, ss)$, then $ss=(t\;x=f(v);ss')$ and $le=V(\ctolow{ss'})$ and $lv=\ctolowe{v}$ and $LE=\unravel(S, \symhole)$.
\end{lemma}

\begin{lemma}[Invert Withframe] \label{lemma-invert-withframe}
  If $(H, \fplug{LE}{\withframe\;le})\;R\;(S, V, ss)$, then $ss=\{ss_1\};ss_2$ and $le=V(\ctolow{ss_1})$ and $LE=\unravel(S, V(\ctolowE{(\symhole;ss_2)}))$.
\end{lemma}

\begin{lemma}[Invert If] \label{lemma-invert-if}
  If $(H, \fplug{LE}{\eif{n}{le_1}{le_2}})\;R\;(S, V, ss)$, then $ss=\eif{e}{ss_1}{ss_2};ss'$ and $\eval{e}{(p,V)}=n$ and $le_i=V(\ctolow{ss_i})$ ($i=1,2$) and $LE=\unravel(S, V(\ctolowE{(\symhole;ss')}))$.
\end{lemma}

\begin{lemma}[Invert Pop] \label{lemma-invert-pop}
  If $(H;h, \fplug{LE}{\epop\;lv})\;R\;(S, V, ss)$, then either $ss=e$ and $\eval{e}{(p,V)}=v$ and $\ctolowe{v}=lv$ and $S=S';(M, V', \symhole;ss')$ and $LE=\unravel(S', V'(\ctolowE{\symhole;ss'}))$, or $ss=\ereturn{e}$ and $\eval{e}{(p,V)}=v$ and $\ctolowe{v}=lv$ and $S=S';(M,V',E)$ and $LE=\unravel(S',V'(\ctolowE{E}))$.
\end{lemma}

\begin{lemma}[Low2C-C2Low] \label{lemma-lowtoc-ctolow}
  If $\lowtoc{le}=ss;e$, then $\ctolow{(ss;\ereturn\;e)}=le$.
\end{lemma}

\section{From \cstar to CompCert C and beyond}

To further back our claim that KreMLin offers a practical yet
trustworthy way to preserve security properties of \fstar programs down to
the executable code, we have to demonstrate that security guarantees
can be propagated from \cstar down to assembly.

Our idea here is to use the CompCert verified C compiler
\cite{Leroy-Compcert-CACM,compcert-url}. CompCert formally proves the
preservation of functional correctness guarantees from C down to
assembly code (for x86, PowerPC and ARM platforms.)

\subsection{Reminder: CompCert Clight}

CompCert Clight \cite{Blazy-Leroy-Clight-09} is a subset of C with no
side effects in expressions, and actual byte-level representation of
values. Syntax in Figure~\ref{fig:clight-syntax}. Semantics
definitions in Figure~\ref{fig:clight-semantics-defs}. Evaluation of
expressions in Figure~\ref{fig:clight-expr}. Small-step semantics in
Figure~\ref{fig:clight-stmts-reduction}.

The semantics of a Clight program is given by the return value of its
$\kw{main}$ function called with no arguments.\footnote{CompCert does
  not support semantics preservation with system arguments.}  Thus,
given a Clight program $p$, the initial configuration of a CompCert
Clight transition from $p$ is $(\{ \}, [], [], [], \kw{int} ~ r =
\kw{main}())$, and a configuration is final with return value $i$ if,
and only if, it is of the form $(\{ \}, \_, \_V, \_, [])$ with $\_V(r)
= i$.

\newcommand{\claddr}[1]{\kw{\&}{#1}}
\newcommand{\clderef}[1]{\kw{*}{#1}}
\newcommand{\clread}[3]{{#1} =_{#2} \left\lbrack{#3} \right\rbrack}
\newcommand{\clwrite}[3]{{#2} =_{#1} {#3}}
\newcommand{\clvalloc}{\kw{valloc}}
\newcommand{\clalloc}{\kw{Alloc}}
\newcommand{\clunkn}{\kw{unkn}}
\newcommand{\clannot}{\kw{annot}}

\begin{figure}[h]
\begin{center}
  \begin{tabularx}{\columnwidth}{rlR}
    $p ::= $ & & program \\
      & $\ls d$                      & series of declarations \\[1.2mm]

    $d ::= $ & & declaration \\
    & $\ecfun fxtt{\ls{ad},ss}$                & top-level function \\
    & & with stack-allocated local variables $\ls{ad}$ \\
    & ${ad}$ & top-level value \\ [1.2mm]
    
    ${ad} ::= $ & & array declaration \\
    & $\earray{t}{x}{n}{}$               & uninitialized global variable \\
    [1.2mm]

    $ss ::= $ & & statement lists \\
    & $\ls{s}$                  &  \\
    [1.2mm]
    
    $s ::= $ & & statements \\
    & $\_x = e$ & assign rvalue to a non-stack-allocated local variable \\
    & $\_x = {\eapply fe}$                  & application \\
    & $\clread{\_x}{t}{e}$               & memory read from lvalue \\
    & $\clwrite tee$               & memory write rvalue to lvalue \\
    & $\clannot(\symread,t,e)$ & annotation to produce read event \\
    & $\clannot(\symwrite,t,e)$ & annotation to produce write event \\
    & $\eif{e}{ss}{ss}$ & conditional \\
    & $\{\stmts\}$ &  block \\
    & \ereturn e & return \\
    [1.2mm]
    
    $e ::= $ & & expressions \\
    & $n$    & integer constant (rvalue) \\
      & $x$                           & stack-allocated variable (lvalue) \\
      & $\_x$                           & non-stack-allocated variable (rvalue) \\ 
     & $e_1+_te_2$                        & pointer add (rvalue, $e_1$ is a rvalue pointer to a value of type $t$ and $e_2$ is a rvalue $\kw{int}$) \\
    & $e._tfd$ & struct field projection (lvalue, $e$ lvalue) \\
    & $\claddr{e}$ & address of a lvalue (rvalue, $e$ lvalue) \\
    & $\clderef{e}$ & pointer dereference (lvalue, $e$ rvalue)
  \end{tabularx}
\end{center}
\caption{Clight Syntax}
\label{fig:clight-syntax}
\end{figure}

Just like C, there are two ways to evaluate Clight expressions: in
lvalue position or in rvalue position. Roughly speaking, in an
expression assignment $\clwrite t{e_l}{e_r}$, expression $e_l$ is said to be at
lvalue position and thus must evaluate into a memory location, whereas
$e_r$ is said to be at rvalue position and evaluates into a value
(integer or pointer to memory location). The operation $\claddr{e}$
takes an lvalue $e$ and transforms it into a rvalue, namely the
pointer to the memory location $e$ designates as an
lvalue. Conversely, $\clderef{e}$ takes an rvalue $e$, which must
evaluate to a pointer, and turns it into the corresponding memory
location as an lvalue.

\paragraph{Memory accesses in the trace}
To account for memory accesses in the trace, we make each statement
perform at most one memory access in our generated Clight code. Then,
we prepend each such memory access statement with a \emph{built-in
  call}, a no-op annotation $\clannot(\mathit{ev},t,e)$ whose
semantics is merely to produce the corresponding memory access event
$\mathit{ev}(b, n)$ in the trace, where $e$ evaluates to the pointer
to offset $n$ within block $e$.

\begin{figure}[h]
\begin{small}
\begin{center}
  \begin{tabularx}{\columnwidth}{rlR}
    $v ::= $ & & values \\
    & $n$    & integer constant \\
    & $(b, n)$                        & memory location \\
    & $\clunkn{}$                        & defined but unknown value \\
    [1.2mm]

    ${vf} ::= $ & & value fragments \\
    & $n$    & byte constant \\
    & $((b,n), n')$ & $n'$-th byte of pointer value $(b,n)$ \\
    & $\clunkn{}$                        & defined but unknown value fragment \\
    [1.2mm]
    
    $V ::=$ & & stack-allocated variable assignments \\
    & $x\sympartial b$ & map from variable to memory block identifier \\
    [1.2mm]

    $\_V ::=$ & & non-stack-allocated variable assignments \\
    & $\_x\sympartial v$ & map from variable to value \\
    [1.2mm]

    $E ::=$ & & evaluation ctx (plug expr to get stmts) \\
    & $\symhole; ss$ & discard returned value \\
    & $\_x=\symhole; ss$ & receive returned value \\
    [1.2mm]

    $F ::=$ & & frames \\
    & $(V, \_V, E)$ & stack frame \\
    [1.2mm]
    
    $S ::=$ & & stack \\
    & $\ls{F}$ & list of frames \\
    [1.2mm]

    $M ::=$ & & memory \\
    & $(b, n) \sympartial vf$ & map from block id and offset to value fragment \\
    [1.2mm]

    $C ::=$ & & configuration \\
    & $(S, V, \_V, M, ss)$ &  \\
    [1.2mm]
  \end{tabularx}
\end{center}
\end{small}
\caption{Clight Semantics Definitions}
\label{fig:clight-semantics-defs}
\end{figure}

\newcommand{\cllv}[2]{\kw{lv}({#1}, {#2})}
\newcommand{\clrv}[2]{\kw{rv}({#1}, {#2})}

\begin{figure*}[h]
\begin{scriptsize}
\begin{flushleft}
  \fbox{$\cllv{e}{(p, V, \_V)}=(b,n)$} \quad \text{and} \quad
  \fbox{$\clrv{e}{(p, V, \_V)}=v$}
\end{flushleft}
\begin{mathpar}
\inferrule* [Right=Var]
{
  V(x) = b
}
{
  \cllv{x}{(p,V, \_V)}=(b, 0)
} 

\quad
\quad
\quad

\inferrule* [Right=GVar]
{
  x \not\in V \\
  p(x) = b
}
{
  \cllv{x}{(p,V, \_V)}=(b, 0)
} 

\\

\inferrule* [Right=PtrFd]
{
  \cllv{e}{(p,V, \_V)} = (b, n) \\
}
{
  \cllv{e._t{fd}}{(p,V, \_V)}=(b, n + \textsf{offsetof}(t, fd))
}

\quad
\quad
\quad
\quad
\quad

\inferrule* [Right=PtrDeref]
{
  \clrv{e}{(p,V,\_V)} = (b,n)
}
{
  \cllv{\clderef{e}}{(p,V,\_V)} = (b,n)
}

\\

\inferrule* [Right=RVar]
{
  \_V(\_x) = v
}
{
  \clrv{\_x}{(p,V,\_V)}=v
} 

\quad
\quad
\quad

\inferrule* [Right=PtrAdd]
{
  \clrv{e_1}{(p,V,\_V)} = (b, n) \\
  \clrv{e_2}{(p,V,\_V)} = n' \\
}
{
  \clrv{e_1+e_2}{(p,V)}=(b, n+n')
} 

\quad
\quad
\quad
\quad

\inferrule* [Right=AddrOf]
{
  \cllv{e}{(p,V,\_V)}= (b, n)
}
{
  \clrv{\claddr{e}}{(p,V,\_V)}= (b,n)
}
\end{mathpar}
\end{scriptsize}
\caption{\label{fig:clight-expr}Clight Expression Evaluation}
\end{figure*}

\begin{figure*}[h]
\begin{scriptsize}
\begin{flushleft}
  \fbox{$p \vdash C\step_\olabel C'$}
\end{flushleft}
\begin{mathpar}
\inferrule* [Right=Read]
{
  \cllv{e}{(p, V, \_V)} = (b, n) \\
  \symget(M, (b, n), \kw{sizeof}(t)) = v
}
{
  p \vdash (S, V, \_V, M, \clread{\_x}{t}{e}; ss) \step (S, V, \_V[\_x\mapsto v], M, ss)
} 

\quad
\quad
\quad
\quad

\\

\inferrule* [Right=Write]
{
  \cllv{e_1}{(p, V, \_V)} = (b, n) \\
  \clrv{e_2}{(p, V, \_V)} = v \\
  \symset(M, (b, n), \kw{sizeof}(t), v) = S'
}
{
  p \vdash (S, V, \_V, M, \clwrite{t}{e_1}{e_2}; ss) \step (S', V, \_V, M, ss)
} 

\\

\inferrule* [Right=Annot]
{
  \cllv{e}{(p, V, \_V)} = (b, n) \\
}
{
  p \vdash (S, V, \_V, M, \clannot{\mathit{ev}, e}; ss) \step_{\mathit{ev}\;(b,n)}(S, V, \_V, M, SS)
}

\\

\inferrule* [Right=Ret]
{
  \clrv{e}{(p,V,\_V)}=v
}
{
  p \vdash (S;(V', \_V',E), V, M, \ereturn\;e; ss) \step (S, V', \_V', M, \fplug{E}{v})
}

\\

\inferrule* [Right=Call]
{
  p(f)=\ecfuntwo{\_y}{t_1}{t_2}{ads; ss_1} \\
  \clrv{e}{(p,V,\_V)}=v \\
  \clvalloc{}(ads, \None, M) = (V', M')
}
{
  p \vdash (S, V, \_V, M, \_x=f\;e; ss) \step (S;(V, \_V, \_x=\symhole;ss), V', \{\}[\_y\mapsto v], M', ss_1)
} 

\\

\inferrule* [Right=AllocNil]
{
}
{
  \clvalloc{}([], V, M) = (V, M)
}

\quad
\quad
\quad
\quad
\quad

\inferrule* [Right=AllocCons]
{
  \clalloc(M,n \times \kw{sizeof}(t)) = (b,M_1) \\
  \clvalloc{}(ads, V[x \mapsto (b,0)], M_1) = (V', M')
}
{
  \clvalloc{}((\earray{t}{x}{n}{}; ads), V, M) = (V', M')
}
  
\\

\inferrule* [Right=Expr]
{
  \clrv{e}{(p,V,\_V)} = v
}
{
  p \vdash (S, V, \_V, M, e; ss) \step (S, V, \_V, M, ss)
} 

\quad
\quad
\quad

\inferrule* [Right=Empty]
{
  \;
}
{
  p \vdash (S, V, \_V, M, []) \step (S, V, \_V, M, \ereturn\;\clunkn{})
} 

\\

\inferrule* [Right=Block]
{
  \;
}
{
  p \vdash (S, V, \_V, M, \{ss_1\};ss_2) \step (S; V, \_V, M, ss_1; ss_2)
} 

\\

\inferrule* [Right=IfT]
{
  \clrv{e}{(p, V, \_V)} = v \\
  v \not= 0 \\
  v \not= \clunkn{}
}
{
  p \vdash (S, V, \_V, M, \eif{e}{ss_1}{ss_2};ss) \step_\brt (S, V, ss_1;ss)
} 

\quad
\quad
\quad

\inferrule* [Right=IfF]
{
  \clrv{e}{(p, V, \_V)} = 0
}
{
  p \vdash (S, V, \_V, M, \eif{e}{ss_1}{ss_2};ss) \step_\brf (S, V, \_V, M, ss_2;ss)
} 
\end{mathpar}
\end{scriptsize}
\caption{\label{fig:clight-stmts-reduction}Clight Configuration Reduction}
\end{figure*}

\paragraph{The CompCert memory model}
The semantics of CompCert Clight statements depends on the CompCert
memory model \cite{Leroy-Blazy-memory-model}. Here we need three
operations: $\symget$, $\symset$ and $\clalloc$, whose description
follows.\footnote{In the current version of CompCert and its memory
  model \cite{2012-Leroy-Appel-Blazy-Stewart}, each memory location is
  equipped with a \emph{permission} to more faithfully model the fact
  that the memory locations of a local variable cannot be read or
  coincidentally reused in a Clight program after exiting the scope of
  the variable. To this end, thanks to this permission model, the
  CompCert memory model also defines a $\kw{free}$ operation which
  invalidates memory accesses while preventing from reusing the memory
  block for further allocations; although we do not describe it here,
  CompCert Clight actually uses this operator to free all local
  variables upon function exit. Thus, it is also necessary to amend
  the semantics of \cstar in a similar way.}

$\symget(M, (b, n), n')$ reads $n'$ bytes from memory $M$ at offset
$n$ within block $b$, and decodes the obtained byte fragments into a
value. It fails if not all locations from $(b, n)$ to $(b, n'-1)$ are
defined. It returns $\clunkn$ if all locations are defined but the
decoding fails.

$\symset(M, (b, n), n', v)$ writes $n'$ bytes from memory $M$ at
offset $n$ within block $b$ corresponding to encoding value $v$ into
$n'$ value fragments. It fails if not all locations from $(b, n)$ to
$(b, n'-1)$ are defined.

$\clalloc{}(M, n)$ returns a pair $(b, M')$ where $b \not\in M$ is a
fresh block identifier and all locations from $(b,0)$ to $(b,n-1)$ in
$M'$ contain $\clunkn$.

\subsection{Issues}

\paragraph{Local structures}

The main difference between \cstar and Clight is that \cstar allows structures
as values. Although converting a \cstar value into a Clight value is no
problem in terms of memory representation (since the layout of Clight
structures\footnote{which can be chosen when configuring CompCert with
  a suitable platform} is already formalized in CompCert with basic
proofs such as the fact that two distinct fields of a structure
designate disjoint sets of memory locations), local structures cause
issues in terms of managing memory accesses (due to our desire for
noninterference in terms of memory accesses), as we describe in
Section~\ref{sec:local-struct}.

\paragraph{Stack-allocated local variables}

Another difference between \cstar and Clight is that, whereas \cstar allows
the user to stack-allocate local variables on the fly, Clight mandates
all local variables of a function to be hoisted to the beginning of
the function (in fact, the list of all stack-allocated variables of a
function is actually part of the function definition), and so they are
allocated all at once when entering the function.

Hoisting local variables is not supported in the verified part of
CompCert.

Consider the following \cstar example, for a given conditional expression
$e$:
  \[
  \eifthenelse
      {e}
      {\earray{\kint}{x}{1}{18}; }
      {\{ \earray{\kword}{x}{1}{42}; \ewrite{x+0}{1729} \} }
   \]
After hoisting, following the same strategy as the corresponding
unverified pass of CompCert, \cstar code will look like this:
    \[
  \earray{\kint}{x_1}{1}{};
  \earray{\kword}{x_2}{1}{};
  \]
  \[
  \eifthenelse
      {e}
      {\ewrite{x_1+0}{18}; }
      {\{ \ewrite{x_2+0}{42}; \ewrite{x_2+0}{1729} \} }
  \]
This example shows the following issue:
when producing the trace, the
  memory blocks corresponding to the accessed memory location will
  differ between the non-hoisted and the hoisted \cstar code. Indeed, in
  the non-hoisted \cstar code, only one variable $x$ is allocated in the
  stack, whereas in the hoisted \cstar code, two variables $x_1$ and $x_2$
  corresponding to both branches will be allocated anyway, regardless
  of the fact that only one branch is executed. Thus, in the \cstar code
  before hoisting, allocating $x$ will create, say, block 1, whereas
  after hoisting, two variables will be allocated, $x_1$ at block 1,
  and $x_2$ at block 2. Thus, the statement $\ewrite{x+0}{1729}$ in
  the \cstar code before hoisting will produce $\symread\;(1,0,[])$ on the
  trace, whereas its corresponding translation after hoisting,
  $\ewrite{x_2+0}{1729}$, will produce
  $\symread\;(2,0,[])$.

\ignore{
Anyway, a similar problem of discrepancy of traces of memory accesses
will also exist in CompCert passes after Clight, this time because the
structure of the stack will change across intermediate languages of
CompCert (in particular, at some point, all stack-allocated variables
of one function call will be merged into one single stack block for
the entire stack frame.) Thus, so far, CompCert will very well
preserve functional correctness down to the assembly, but will require
more work to \emph{preserve} security guarantees beyond Clight, unless
we instrument it as described before.
}

One quick solution to ensure that those event traces are exactly
preserved, is to replace the actual pointer $(b, n,
\mathit{fd})$ on $\symread$ and $\symwrite$ events with $(f, i, x, n,
\mathit{fd})$ where $f$ is the name of the function, $i$ is the
recursion depth of the function (which would be maintained as a global
variable, increased whenever entering the function, and decreased
whenever exiting) and $x$ is the local variable being accessed. (In
concurrent contexts, one could add a parameter $\theta$ recording the
identifier of the current thread within which $f$ is run, and so the
global variable maintaining recursion depth would become an array
indexed by thread identifiers.)

\ignore{
[Catalin: Isn't this much weakening the attacker model though?
          And how does this new model apply to machine code?]

[Tahina: This abstract trace model is only for the purpose of proofs
  of memory transformations, beginning with the \cstar to Clight
  proof. This model also applies to C$\sharp$minor (exact
  preservation, the only difference being the erasure of C-style type
  information), and it can also be applied to Cminor (where all local
  variables are merged into one single stack frame) when translating
  from C$\sharp$minor. Then, from Cminor on, it becomes useless. In
  fact, from Cminor on, the local variable region becomes contiguous
  for a given function call, and (with the notable exception of in the
  function inlining pass, which may very well merge several stack
  frames of inlined callees into one single one. Anyway, function
  inlining requires special treatment.) there are no pointer
  transformations until the time spilling locations are introduced in
  memory when generating assembly code.

  Once again, the ultimate trace model that we want for our final
  side-effect security correctness statement will definitely be the
  one with concrete pointer values, and abstract pointer values are
  only for the purpose of compiler verification.]
}

Finally, we adopted this solution as we describe in
Section~\ref{sec:norm-traces}, and we heavily use it to prove the
correctness of hoisting within \cstar in Section~\ref{sec:hoisting}.

\subsection{Summary: from \cstar to Clight}

Given a \cstar program $p$ and an entrypoint $ss$, we are going to
transform it into a CompCert Clight program in such a way that both
functional correctness and noninterference are preserved.

This will not necessarily mean that traces are exactly preserved
between \cstar and Clight, due to the memory representation discrepancy
described before. Instead, by functional correctness, we mean that a
safe \cstar program is turned into a safe Clight program, and for such
safe programs, termination, I/O events and return value are preserved;
and by noninterference, we mean that if two executions with different
secrets produce identical (whole) traces in \cstar, then they will also
produce identical traces in CompCert Clight (although the trace may
have changed between \cstar and Clight.)

\begin{enumerate}
\item In section \ref{sec:unambiguous-variables}, we transform a \cstar
  programs into a program with unambiguous variable names, and we take
  advantage of such a syntactic property by enriching the
  configuration of \cstar with more information regarding variable names,
  thus yielding the \cstar2 language.
\item In section \ref{sec:norm-traces-detail}, we execute the obtained \cstar2
  program with a different, more abstract, trace model, as proposed
  above. This is not a program transformation, but only a
  reinterpretation of the same \cstar program with a different operational
  semantics, which we call \cstar3.
\item In section \ref{sec:hoisting}, we transform the \cstar3 program into
  a \cstar3 program where all local arrays are hoisted from block-scope to
  function-scope. The abstract trace model critically helps in the
  success of this proof where memory state representations need to
  change.
\item In section \ref{sec:struct-return}, we transform the obtained
  \cstar3 program into a \cstar3 program where functions returning structures
  are replaced with functions taking a pointer to the return location
  as additional argument. Thus, we need to account for an additional
  memory access, which we do through the \cstar4 intermediate semantics,
  another reinterpretation of the source \cstar3 program producing new
  events at functions returning structures. Then, our reinterpreted
  \cstar4 program is translated back to \cstar3 with those additional memory
  accesses made explicit.
\item In section \ref{sec:struct-events}, we reinterpret our obtained
  \cstar3 program with a different event model where memory access events
  of structure type are replaced by the sequences of memory access
  events of all their non-structure fields. We call this new language
  \cstar5.
\item In section \ref{sec:local-struct-detail}, we transform our \cstar5
  program back into \cstar3 by erasing all local structures that are not
  local arrays, replacing them with their individual non-structure
  fields. Thus, the more ``elementary'' memory accesses introduced in
  the \cstar3 to \cstar5 reinterpretation are made concrete.
\item We then reinterpret the obtained \cstar3 program back into \cstar2 as
  described in \ref{sec:norm-traces-detail}, reverting to the traces
  with concrete memory locations at events, thus accounting for all
  memory accesses.
\item Finally, in section \ref{sec:clight-gen}, we compile the
  obtained \cstar2 program, now in the desired form, into CompCert Clight.
\end{enumerate}

\subsection{Normalized event traces in \cstar} \label{sec:norm-traces}

As described above, traces where memory locations explicitly appear
are notoriously hard to reason about in terms of semantics
preservation for verified compilation. Thus, it becomes desirable to
find a common representation of traces that can be preserved between
different memory layouts across different intermediate languages.

In particular here, we would like to replace concrete pointers into
abstract pointers representing the local variable being modified in a
given nested function call.

\subsubsection{Disambiguation of variable names} \label{sec:unambiguous-variables}

To this end, we first need to disambiguate the names of the local
variables of a \cstar function:

\begin{definition}[Unambiguous local variables]
  We say that a list of \cstar instructions has \emph{unambiguous local
    variables} if, and only if, it contains no two distinct array
  declarations with the same variable name, and does not contain both
  an array declaration and a non-array declaration with the same
  variable name.
  
  We say that a \cstar program has unambiguous local variables if, and
  only if, for each of its functions, its body has unambiguous local
  variables.

  We say that a \cstar transition system $\sys(p, V, ss)$ has unambiguous
  local variables if, and only if, $p$ has unambiguous local
  variables, $ss$ has unambiguous local variables, and $V$ does not
  define any variable with the same name as an array declared in $ss$.
\end{definition}

\begin{lemma}[\label{lemma-cstar-alpha} Disambiguation]
  There exists a transformation $T$ on lists of instructions (extended
  to programs by morphism) such that, for any \cstar program $p$, and for
  any list of \cstar instructions $ss$, %
  for any variable mapping $V'$ 
  such that $\sys(T(p),V',T(ss))$ has unambiguous local variables,
  there exists a variable mapping $V$ 
  such that $\sys(p,V,ss)$ and $\sys(T(p),V',T(ss))$ have the same
  execution traces.
\end{lemma}
\begin{proof} $\alpha$-renaming. \end{proof}

The noninterference property can be proven to be stable by such
$\alpha$-renaming:


\begin{lemma}
  Let $p$ be a \cstar program and $ss$ be a list of instructions. Assume
  that for any $V_1, V_2$ such that $\sys(p, V_1, ss)$ and $\sys(p,
  V_2, ss)$ are both safe, then they have the same execution
  traces.

  Then, for any $V_1', V_2'$ such that $\sys(T(p), V_1', T(ss))$ and
  $\sys(T(p), V_2', T(ss))$ are both safe, they have the same
  execution traces.
\end{lemma}
\begin{proof}
  By Lemma~\ref{lemma-cstar-alpha}, there is some $V_1$ such that
  $\sys(p, V_1, ss)$ and $\sys(T(p), V_1', T(ss))$ have the same
  execution traces, thus in particular, $\sys(p, V_1, ss)$ is
  safe. Same for some $V_2$. By hypothesis, $\sys(p, V_1, ss)$ and
  $\sys(p, V_2, ss)$ have the same execution traces, thus the result
  follows by transitivity of equality.
\end{proof}

Thus, we can now restrict our study to \cstar programs whose functions
have no two distinct array declarations with the same variable names.

Let us first enrich the configuration $(S, V, ss)$ of \cstar small-step
semantics with additional information recording the current function
$f$ being executed (or maybe $\None$) and the set $A$ of the variable
names of local arrays currently declared in the scope. Thus, a \cstar
stack frame $(\None, V, E)$ becomes $(\None, V, E, f, A)$ where $f$ is
the caller, a block frame $(M, V, E)$ becomes $(M, V, E, f, A)$ where
$f$ is the enclosing function of the block, and the configuration $(S,
V, ss)$ becomes $(S, V, ss, f, A)$ where $f$ is the current function
(with the frames of $S$ changed accordingly.) Let us then change some
rules accordingly as described in Fig.~\ref{fig:cstar-2}, leaving
other rules unchanged except with the corresponding $f$ and $A$
components preserved.

\begin{figure*}
\begin{small}
  \begin{mathpar}
\inferrule* [Right=$\text{ArrDecl}_2$]
{
  \eval{e}{(p,V)}=v \\
  S = S'; (M, V, E, f', A') \\
  b\not\in S
}
{
  p \vdash (S, V, t\;x[n]=e; ss, f, A) \step (S';(M[b\mapsto v^n], V, E, f', A'), V[x\mapsto (b, 0, [])], ss, f, A\cup\{x\})
}

\\

\inferrule* [Right=$\text{Ret}_2$]
{
  \eval{e}{(p,V)}=v
}
{
  p \vdash (S;(\None, V',E, f', A'), V, \ereturn\;e; ss, f, A) \step (S, V', \fplug{E}{v}, f', A')
}

\\
    
\inferrule* [Right=$\text{Call}_2$]
{
  p(f)=\ecfuntwo{y}{t_1}{t_2}{ss_1} \\
  \eval{e}{(p,V)}=v
}
{
  p \vdash (S, V, t\;x=f\;e; ss, f', A') \step (S;(\None, V, t\;x=\symhole;ss, f', A'), \{\}[y\mapsto v], ss_1, f, \{\})
} 

  \end{mathpar}
\end{small}
\caption{\cstar2 Amended Configuration Reduction} \label{fig:cstar-2}
\end{figure*}

Let us call \cstar2 the obtained language where the initial state of the
transition system $\sys(p, V, ss)$ shall now be $([], V, ss, \None,
[])$.

Then, it is easy to prove the following:
\begin{lemma}[\label{lemma-cstar-to-cstar-2]}\cstar to \cstar2]
    If $\sys(p, V, ss)$ has unambiguous local variables and is safe in
    \cstar, then it has the same execution traces in \cstar as in \cstar2 (and in
    particular, it is also safe in \cstar2.)

    Thus, both functional correctness and noninterference are
    preserved from \cstar to \cstar2.
\end{lemma}
\begin{proof}
  Lock-step bisimulation where the common parts of the configurations
  (besides the \cstar2-specific $f, A$ parts) are equal between \cstar and
  \cstar2.
\end{proof}

Then, we can prove an invariant over the small-step execution of a \cstar2
program:
\begin{lemma}[\label{lemma-cstar-2-invar}\cstar2 invariant]
Let $p$ be a \cstar2 program and $V$ a variable environment such that
$\sys(p,V,ss)$ has unambiguous local variables.

Let $n \in \mathbb N$. Then, for any \cstar2 configuration $(S, V', ss',
f, A)$ obtained after $n$ \cstar2 steps from $(\{\}, V, ss, \None, \{\})$,
the following invariants hold:
\begin{enumerate}
\item for any variable or array declaration $x$ in $ss'$, it does not
  appear in $A_1$
\item any variable name in $A$ or in an array declaration of $ss'$ is
  in an array declaration of $ss$ (if $f = \None$) or the body of $f$
  (otherwise.)
\item for each frame of $S$ of the form $(\_, \_, E, f'', A'')$, then
  any variable name or array declaration in $E$ does not appear in
  $A''$, and any variable name in $A''$ or in an array declaration of
  $E$ is in an array declaration in $ss$ (if $f'' = \None$) or the
  body of $f''$ (otherwise.)
\item \label{cstar-2-invar-memory-block-nil} if $S = S'; (M, \_, \_,
  f', A')$ with $M \not= \None$, then $f' = f$, $A' \subseteq A$ and
  for all block identifiers $b$ defined in $M$, there exists a unique
  variable $x \in A$ such that $V'(x) = b$
\item \label{cstar-2-invar-memory-block-cons} for any two consecutive
  frames $(M, \_, \_, f_1, A_1)$ just below $(\_, V_2, \_, f_2, A_2)$
  with $M \not= \None$, then $f_1 = f_2$ and $A_1 \subseteq A_2$ and
  $V_1(x) = V_2(x)$ for all variables $x \in A_1$, and for all block
  identifiers $b$ defined in $M$, there exists a unique variable $x
  \in A_2$ such that $V_2(x) = b$
\item \label{cstar-2-invar-memory-disjoint} for any two different
  frames of $S$ of the form $(M_1, \_, \_, \_, \_)$ and $(M_2, \_, \_,
  \_, \_)$ with $M_1 \not= \None$ and $M_2 \not= \None$, the sets of
  block identifiers of $M_1$ and $M_2$ are disjoint
\end{enumerate}
\end{lemma}
\begin{proof}
  By induction on $n$ and case analysis on $\step$.
\end{proof}

Then, we can prove a strong invariant between two executions of
the same \cstar2 program with different secrets. This strong invariant
will serve as a preparation towards changing the event traces of \cstar2.

\begin{lemma}[\label{lemma-cstar-2-noninterference-invar} \cstar2 noninterference invariant]
Let $p$ be a \cstar2 program, and $V_1, V_2$ be two variable environments
such that $\sys(p, V_1, ss)$ and $\sys(p, V_2, ss)$ have unambiguous
local variables, are both safe and produce the same traces in \cstar2.

Let $n \in \mathbb N$. Then, for any two \cstar2 configurations $(S_1,
V_1', ss_1, f_1, A_1)$ and $(S_2, V_2', ss_2, f_2, A_2)$ obtained
after $n$ \cstar2 execution steps, the following invariants hold:
\begin{itemize}
\item $ss_1 = ss_2$
\item $S_1, S_2$ have the same length
\item $f_1 = f_2$
\item $A_1 = A_2$
\item $V_1(x) = V_2(x)$ for each $x \in A_1$
\item for each $i$, if the $i$-th frames of $S_1, S_2$ are $(M_1,
  V_1'', E_1, f_1'', A_1'')$ and $(M_2, V_2'', E_2, f_2'', A_2'')$,
  then $E_1 = E_2$, $f_1'' = f_2''$ and $A_1'' = A_2''$ and $V_1''(x)
  = V_2''(x)$ for any $x \in A_1''$. Moreover, $M_1 = \None$ if and
  only if $M_2 = \None$, and if $M_1 \not= \None$, then $M_1$ and
  $M_2$ have the same block domain.
\end{itemize}
Thus, the $n+1$-th step in both executions applies the same \cstar2 rule.
\end{lemma}
\begin{proof}
  By induction on $n$ and case analysis on $\step$, also using the
  invariant of Lemma~\ref{lemma-cstar-2-invar}. In particular the
  equality of codes is a consequence of the fact that there are no
  function pointers in \cstar.\footnote{If we were to allow function
    pointers in \cstar, then we would have to add function call/return
    events into the \cstar trace beforehand, and assume that traces with
    those events are equal before renaming = prove that they are equal
    on the \lowstar program as well. This might have consequences in the
    proof of function inlining in the \fstar-to-\cstar translation.} Then,
  both executions apply the same \cstar2 rules since \cstar2 small-step rules
  are actually syntax-directed.
\end{proof}

\subsubsection{Normalized traces} \label{sec:norm-traces-detail}

Now, consider an execution of \cstar2 from some initial state. In fact,
for any block identifier $b$ defined in $S$, it is easy to prove that
it actually corresponds to some variable defined in the scope. The
corresponding \textsc{VarOfBlock} algorithm is shown in
Figure~\ref{fig:cstar-2-block-to-variable}. \ignore{ \textbf{FIXME: global
  variables}\footnote{Global variables (top-level values) are not
  allocated anywhere in the memory of a \cstar program, so \cstar cannot
  access them. \cstar can only access the local variables of the
  entrypoint list of statements, which is not part of the program.}}

\begin{figure}
  \textbf{Algorithm:} \textsc{VarOfBlock}
  
  \textbf{Inputs:}
  \begin{itemize}
    \item \cstar2 configuration $(S, V, \_, \_, A)$ such that the invariants
      of Lemma~\ref{lemma-cstar-2-invar} hold
    \item Memory block $b$ defined in $S$
  \end{itemize}
  \textbf{Output:} function, recursion depth and local variable
  corresponding to the memory block

  Let $S = S_1 ; (M, \_, \_, f, \_) ; S_2$ such that $b$ defined in
  $M$. (Such a decomposition exists and is unique because of
  Invariant~\ref{cstar-2-invar-memory-disjoint}. $f$ may be $\None$.)

  Let $n$ be the number of frames in $S_1$ of the form $(\None, \_,
  \_, f', \_)$ with $f' = f$.

  Let $V'$ and $A'$ such that $S_2 = (\_, V', \_, \_, A) ; \_$, or $V'
  = V$ and $A' = A$ if $S_2 = \{\}$.

  Let $x$ such that $V'(x) = (b, 0)$ (exists and is unique because of
  Invariants~\ref{cstar-2-invar-memory-block-nil}
  and~\ref{cstar-2-invar-memory-block-cons}.)

  \textbf{Result:} $(f, n, x)$
  
  \caption{\cstar2: retrieving the local variable corresponding to a
    memory block}
  \label{fig:cstar-2-block-to-variable}
\end{figure}

Then, let \cstar3 be the \cstar2 language where the \textsc{Read} and
\textsc{Write} rules are changed according to
Figure~\ref{fig:cstar-3}, with event traces where the actual pointer
is replaced into an abstract pointer obtained using the
\textsc{VarOfBlock} algorithm above.

\begin{figure*}
\begin{small}
  \begin{mathpar}
\inferrule* [Right=$\text{Read}_3$]
{
  C = (S, V, t\;x=*[e]; ss, f, A) \\
  \eval{e}{(p, V)} = (b, n, \ls{fd}) \\
  \symget(S, (b, n, \ls{fd})) = v \\
  \ell = \textsc{VarOfBlock}(C, b)
}
{
  p \vdash C \step_{\symread\;(\ell,n,\ls{fd})} (S, V[x\mapsto v], ss, f, A)
} 

\\

\inferrule* [Right=$\text{Write}_3$]
{
  C = (S, V, *e_1=e_2; ss, f, A) \\
  \eval{e_1}{(p, V)} = (b, n, \ls{fd}) \\
  \eval{e_2}{(p, V)} = v \\\\
  \symset(S, (b, n, \ls{fd}), v) = S' \\
  \ell = \textsc{VarOfBlock}(C, b)
}
{
  p \vdash C \step_{\symwrite\;(\ell,n,\ls{fd})} (S', V, ss, f, A)
} 
  \end{mathpar}
\end{small}
\caption{\cstar3 Amended Configuration Reduction} \label{fig:cstar-3}
\end{figure*}

\begin{lemma}[\label{lemma-cstar-2-to-cstar-3-correct}\cstar2 to \cstar3 functional correctness]
  If $\sys(p, V, ss)$ has no unambiguous variables, then $\sys(p, V,
  ss)$, has the same behaviors in \cstar2 with event traces with
  $\symread, \symwrite$ removed, as in \cstar3 with event traces with
  $\symread, \symwrite$ removed.
\end{lemma}
\begin{proof}
  Lock-step bisimulation with equal configurations. Steps
  \textsc{Read} and \textsc{Write} need the invariant of
  Lemma~\ref{lemma-cstar-2-invar} on \cstar2 to prove that \cstar3 does not
  get stuck (ability to apply \textsc{VarOfBlock}.)
\end{proof}

\begin{lemma}[\label{lemma-cstar-2-varofblock-invert}\textsc{VarOfBlock} inversion]
  Let $C_1, C_2$ two \cstar2 configurations such that invariants of
  Lemma~\ref{lemma-cstar-2-invar} and
  Lemma~\ref{lemma-cstar-2-noninterference-invar} hold. Then, for any
  block identifiers $b_1, b_2$ such that $\textsc{VarOfBlock}(C_1,
  b_1)$ and $\textsc{VarOfBlock}(C_2, b_2)$ are both defined and
  equal, then $b_1 = b_2$.
\end{lemma}
\begin{proof}
   Assume $\textsc{VarOfBlock}(C_1, b_1) = \textsc{VarOfBlock}(C_2,
   b_2) = (f, n, x)$. When applying $\textsc{VarOfBlock}$, consider
   the frames $F_1, F_2$ holding the memory states defining $b_1,
   b_2$. Consider the variable mapping $V'_2$ in the frame just above
   $F_2$ (or in $C_2$ if such frame is missing.) Then, it is such that
   $V'_2(x) = b_2$.
  \begin{itemize}
  \item If $F_1$ and $F_2$ are at the same level in their respective
    stacks, then the variable mapping $V'_1$ in the frame directly
    above $F_1$ (or in $C_1$ if such frame is missing) is such that
    $V'_1(x) = b_1$, and also $V'_1(x) = V'_2(x)$ by the invariant, so
    $b_1 = b_2$.
  \item Otherwise, without loss of generality (by symmetry), assume
    that $F_2$ is strictly above $F_1$ (i.e. $F_2$ is strictly closer
    to the top of its own stack than $F_1$ is in its own stack.) Thus,
    in the stack of $C_1$, all frames in between $F_1$ and the frame
    $F_1''$ corresponding to $F_2$ are of the form $(M', V', \_, f',
    \_)$ with $f' = f$ and $M' \not= \None$ (otherwise the functions
    and/or recursion depths would be different.) By
    invariant~\ref{cstar-2-invar-memory-block-cons} of
    Lemma~\ref{lemma-cstar-2-invar}, it is easy to prove that $V'(x) =
    b_1$, and thus also for the variable mapping $V'_1$ in the frame
    just above $F_1''$ (or in $C_1$ directly if there is no such
    frame.) By invariants of
    Lemma~\ref{lemma-cstar-2-noninterference-invar}, we have $V'_1(x)
    = V'_2(x)$, thus $b_1 = b_2$.
  \end{itemize}
\end{proof}
  
\begin{lemma}[\label{lemma-cstar-2-to-cstar-3-noninterference}\cstar2 to \cstar3 noninterference]
  Assume that $\sys(p, V_1, ss)$ and $\sys(p, V_2, ss)$ have no
  unambiguous variables. Then, they are both safe in \cstar2 and produce
  the same traces in \cstar2, if and only if they are both safe in \cstar3 and
  produce the same traces in \cstar3.
\end{lemma}
\begin{proof}
  Use the invariants of Lemma~\ref{lemma-cstar-2-invar} and
  Lemma~\ref{lemma-cstar-2-noninterference-invar}. In fact, the
  configurations and steps are the same in \cstar2 as in \cstar3, only the
  traces differ between \cstar2 and \cstar3. \textsc{Read} and \textsc{Write}
  steps match between \cstar2 and \cstar3 thanks to
  Lemma~\ref{lemma-cstar-2-varofblock-invert}.
\end{proof}

\begin{lemma}[\cstar3 invariants]
  The invariants of Lemma~\ref{lemma-cstar-2-invar} and
  Lemma~\ref{lemma-cstar-2-noninterference-invar} also hold in \cstar3.
\end{lemma}

\subsection{Local variable hoisting} \label{sec:hoisting}

On \cstar3, hoisting can be performed, which will modify the structure of
the memory (namely the number of memory blocks allocated), which is
fine thanks to the fact that event traces carry abstract pointer
representations instead of concrete pointer values.

\paragraph{Memory allocator and dangling pointers}
However, we have to cope with dangling pointers whose address should
not be reused. Consider the following \cstar code:
\[
\begin{array}{l}
  \earray{\kint}{x}{1}{18}; \\
  \earray{\kint *}{p}{1}{x}; \\
  \{ \\
  ~ \earray{\kint}{y}{1}{42}; \\
  ~ \ewrite{p}{y}; \\
  \} \\
  \{ \\
  ~ \earray{\kint}{z}{1}{1729}; \\
  ~ \evardecl{\kint *}{q}{\eread{p}}; \\
  ~ f(q) \\
  \} \\
\end{array}
\]
With a careless memory allocator which would reuse the space of $y$
for $z$, the above program would call $f$ not with a dangling pointer
to $y$, but instead with a valid pointer to $z$, which might not be
expected by the programmer. Then, if $f$ uses its argument to access
memory, what should the \textsc{VarOfBlock} algorithm compute? I claim
that such a \cstar3 program generated from a safe \fstar program should never
try to access memory through dangling pointers.

As far as I understood, a \lowstar program obtained from a well-typed \fstar
program should be safe \emph{with any memory allocator}, including
with a memory allocator which never reuses previously allocated block
identifiers, as in CompCert.\footnote{Formally, a \lowstar (or \cstar)
  configuration should be augmented with a state $\Sigma$ so that the
  \lowstar \textsc{NewBuf} rule (or the \cstar \textsc{ArrDecl} rule), instead
  of picking a block identifier $b$ not in the domain of the memory,
  call an allocator $\kw{alloc}$ with two parameters, the domain $D$
  of the memory and the state $\Sigma$, and returning the fresh block
  $b \not\in D$ and a new state $\Sigma'$ for future
  allocations. Then, a CompCert-style allocator would, for instance,
  use $\mathbb N$ as the type of block identifiers, as well as for the
  type of $\Sigma$, so that if $\kw{alloc}(D, \Sigma) = (b, \Sigma')$,
  then it is ensured that $b \not\in D$, $\Sigma \leq b$ and $b <
  \Sigma'$. In that case, the domain of the memory being always within
  $\Sigma$, could then be easily proven as an invariant of \lowstar (or
  \cstar).} In particular, a \lowstar program safe with such a CompCert-style
allocator will actually never try to access memory through a dangling
pointer to a local variable no longer in scope.

Then, traces with concrete pointer values are preserved from \lowstar to
\cstar2 \emph{with the allocator fixed} in advance in all of \lowstar, \cstar and
\cstar2; and functional correctness and noninterference are also
propagated down to \cstar3 using the same memory allocator.

There should be a way to prove the following:
\begin{lemma}
  If a \cstar3 program is safe with a CompCert-style memory allocator,
  then it is safe with any memory allocator and the traces (with
  abstract pointer representations) are preserved by change of memory
  allocator.
\end{lemma}
\begin{proof}
  Lock-step simulation where the configurations have the same
  structure but a (functional but not necessarily injective) renaming
  of block identifiers from a CompCert-style allocator to any
  allocator is maintained and augmented throughout the execution. In
  particular, we have to prove that \textsc{VarOfBlock} is stable
  under such renaming.
\end{proof}

If so, then for the remainder of this paper, we can consider a
CompCert-style allocator.

\paragraph{Hoisting}

\begin{definition}[Hoisting]
  For any list of statements $\mathit{ss}$ with unambiguous local
  variables, the \emph{hoisting} operation $\kw{hoist}(\mathit{ss}) =
  (\mathit{ads}, \mathit{ss}')$ is so that $\mathit{ads}$ is the list
  of all array declarations in $\mathit{ss}$ (regardless of their
  enclosing code blocks) and $\mathit{ss}'$ is the list of statements
  $\mathit{ss}$ with all array declarations replaced with $()$.

  Then, hoisting the local variables in the body $ss$ of a function is
  defined as replacing $ss$ with the code block $\{ \mathit{ads};
  \mathit{ss}' \} $ where $\kw{hoist}(\mathit{ss}) = (\mathit{ads},
  \mathit{ss}')$; and then, hoisting the local variables in a program
  $p$, $\kw{hoist}(p)$, is defined as hoisting the local variables in
  each of its functions.
\end{definition}

\begin{definition}[Renaming of block identifiers] \label{def:cstar-3-rename-blocks}
  Let $C_1, C_2$ be two \cstar3 configurations. Block identifier $b_1$ is
  said to \emph{correspond} to block identifier $b_2$ from $C_1$ to
  $C_2$ if, and only if, either $\textsc{VarOfBlock}(C_1, b_1)$ is
  undefined, or $\textsc{VarOfBlock}(C_1, b_1)$ is defined and equal
  to $\textsc{VarOfBlock}(C_2, b_2)$.

  Then, value $v_1$ corresponds to $v_2$ from $C_1$ to $C_2$ if, and
  only if, either they are equal integers, or they are pointers $(b_1,
  n_1, fd_1)$, $(b_2, n_2, fd_2)$ such that $fd_1 = fd_2$, $n_1 = n_2$
  and $b_1$ corresponds to $b_2$ from $C_1$ to $C_2$, or they are
  structures with the same field identifiers and, for each field $f$,
  the value of the field $f$ in $v_1$ corresponds to the value of the
  field $f$ in $v_2$ from $C$ to $C'$.
\end{definition}

\begin{theorem}[Correctness of hoisting]
  If $\sys(p, V, ss)$ is safe in \cstar3 with a CompCert-style allocator,
  then $\sys(p, V, ss)$ and $\sys(\kw{hoist}(p), V, \kw{hoist}(ss))$
  have the same execution traces (and in particular, the latter is
  also safe) in \cstar3 using the same CompCert-style allocator.
\end{theorem}
\begin{proof}
  Forward downward simulation from \cstar3 before to \cstar3 after hoisting,
  where one step before corresponds to one step after, except at
  function entry where at least two steps are required in the compiled
  program (function entry, followed by entering the enclosing block
  that was added at function translation, then allocating all local
  variables if any), and at function exit, where two steps are
  required in the compiled program (exiting the added block before
  exiting the function.)

  Then, since \cstar3 is deterministic, the forward downward simulation is
  flipped into an upward simulation in the flavor of CompCert; thus
  preservation of traces.

  For the simulation diagram, we combine the invariants of
  Lemma~\ref{lemma-cstar-2-invar} with the following invariant between
  configurations $C = (S, V, ss, f, A)$ before hoisting and $C' = (S',
  V', ss', f', A')$ after hoisting:
  \begin{itemize}
  \item for all variables $x$ defined in $V$, $V(x)$, if defined,
    corresponds to $V'(x)$ from $C$ to $C'$
  \item $ss'$ is obtained from $ss$ by replacing all array
    declarations with $()$
  \item $f' = f$
  \item $A \subseteq A'$
  \item the set of variables declared in $ss$ is included in $A'$
  \item if a block identifier $b$ corresponds to $b'$ from $C$ to
    $C'$, then the value $\symget(S, b, n, fds)$, if defined, corresponds
    to $\symget(S', b', n, fd)$ from $C$ to $C'$
  \end{itemize}
  Each frame of the form $(\None, V_1, E, f_1, A_1)$ in $S$ is
  replaced with two frames in $S'$, namely $(\None, V_1', E, f_1,
  A_1') ; (M', V_2', \symhole, f_2, A_2')$ where:
  \begin{itemize}
  \item for all variables $x$ defined in $V_1$, $V_1(x)$ corresponds to
    $V'_1(x)$ from $C$ to $C'$
  \item all array declarations of $E$ are present in $A_1'$
  \item $E'$ is obtained from $E$ by replacing all array
    declarations with $()$
  \item $f_1' = f_1$
  \item $A_1 \subseteq A_1'$
  \item $A_2'$ contains all variable names of arrays declared in
    $f_2$, and is the block domain of $M'$
  \item $V_2'$ is defined for all variable names in $A_2'$ as a block
    identifier valid in $M'$
  \item all memory locations of arrays declared in $f_2$ are valid in
    $M'$
  \end{itemize}
  Each frame of the form $(M, V_1, E, f_1, A_1)$ in $S$ with $M \not=
  \None$ is replaced with one frame in $S'$, namely $(\{ \}, V_1', E',
  f_1', A_1')$ where:
  \begin{itemize}
  \item all blocks of $M$ are defined in $A_1'$
  \item for all variables $x$ defined in $V_1$, $V_1(x)$, if defined,
    corresponds to $V_1'(x)$ from $C$ to $C'$
  \item all array declarations of $E$ are present in $A_1'$
  \item $E'$ is obtained from $E$ by replacing all array declarations
    with $()$
  \item $f_1' = f_1$
  \item $A_1 \subseteq A_1'$
  \end{itemize}

  The fact that we are using a CompCert-style memory allocator is
  crucial here to ensure that, once a source block identifier $b$
  starts corresponding to a target one, it remains so forever, in
  particular after its block has been freed (i.e. after its
  corresponding variable has fallen out of scope), since in the latter
  case, it corresponds to any block identifier and nothing has to be
  proven then (since accessing memory through it will fail in the
  source, per the fact that the CompCert-style memory allocator will
  never reuse $b$.)
\end{proof}

\subsection{Local structures} \label{sec:local-struct}

\cstar has structures as values, unlike CompCert C and Clight, which both
need all structures to be allocated in memory. With a naive \cstar-to-C
compilation phase, where \cstar structures are compiled as C structures
and passed by value to functions, we experienced more than 60\%
slowdown with CompCert compared to GCC -O1, using the \lamstar
benchmark in Figure~\ref{fig:struct-erase-benchmark}, extracted to C as Figure~\ref{fig:struct-erase-before}. This
is because, unlike GCC, CompCert cannot detect that a structure is
never taken its address, which is mostly the case for local structures
in code generated from \cstar. This is due to the fact that, even at the
level of the semantics of C structures in CompCert, a field access is
tantamount to reading in memory through a constant offset. In other
words, CompCert has no view of C structures other than as memory
regions. To solve this issue, we replace local structures with their
individual non-compound fields, dubbed as \emph{structure erasure}. Our
benchmark after structure erasure is shown in Figure~\ref{fig:struct-erase-after}.

\begin{figure}
\begin{lstlisting}[language=fstar]
module StructErase
open FStar.Int32
open FStar.ST

type u = { left: Int32.t; right: Int32.t }

let rec f (r: u) (n: Int32.t): Stack unit (fun _ -> true) (fun _ _ _ -> true)  =
 push_frame();
 (
  if lt n 1l
  then ()
  else
   let r' : u = { left = sub r.right 1l ; right = add r.left 1l } in
   f r' (sub n 1l)
 );
 pop_frame()

let test () = 
 let r : u = { left = 18l ; right = 42l } in
 let z2 = mul 2l 2l in
 let z4 = mul z2 z2 in
 let z8 = mul z4 z4 in
 let z16 = mul z8 z8 in
 let z24 = mul z8 z16 in
 let z = mul z24 2l in
 f r z  (* without structure erasure, CompCert segfaults
           if replaced with 2*z *)
\end{lstlisting}
\caption{\lowstar benchmarking for structure erasure}
\label{fig:struct-erase-benchmark}
\end{figure}

\begin{figure}
\end{figure}

\begin{figure}
\begin{lstlisting}
typedef struct {
  int32_t left;
  int32_t right;
} StructErase_u;

void StructErase_f(StructErase_u r, int32_t n) {
  if (n < (int32_t )1) { } else {
    StructErase_u r_ = {
      .left = r.right - (int32_t )1,
      .right = r.left + (int32_t )1
    };
    StructErase_f(r_, n - (int32_t )1);
  }
}

void StructErase_test() {
  StructErase_u r = {
    .left = (int32_t )18,
    .right = (int32_t )42
  };
  int32_t z2 = (int32_t )4;
  int32_t z4 = z2 * z2;
  int32_t z8 = z4 * z4;
  int32_t z16 = z8 * z8;
  int32_t z24 = z8 * z16;
  int32_t z = z24 * z2;
  StructErase_f(r, z);
  return;
}
\end{lstlisting}
\caption{Extracted C code, before structure erasure}
\label{fig:struct-erase-before}
\end{figure}

\begin{figure}
\begin{lstlisting}
void StructErase_f(int32_t r_left, int32_t r_right, int32_t n) {
  if (n < (int32_t )1) { } else {
    int32_t r__left = r_right - (int32_t )1
    int32_t r__right = r_left + (int32_t )1;
    StructErase_f(r__left, r__right, n - (int32_t )1);
  }
}

void StructErase_test() {
  int32_t r_left = (int32_t )18;
  int32_t r_right = (int32_t )42;
  int32_t z2 = (int32_t )4;
  int32_t z4 = z2 * z2;
  int32_t z8 = z4 * z4;
  int32_t z16 = z8 * z8;
  int32_t z24 = z8 * z16;
  int32_t z = z24 * z2;
  StructErase_f(r_left, r_right, z);
  return;
}
\end{lstlisting}
\caption{Extracted C code, after structure erasure}
\label{fig:struct-erase-after}
\end{figure}

In our noninterference proofs where we prove that memory accesses are
the same between two runs with different secrets, treating all local
structures as memory accesses would become a problem, especially
whenever a field of a local structure is read as an expression (in
addition to the performance decrease using CompCert.)  This is another
reason why, in this paper (although a departure from our current
KreMLin implementation), we propose an easier proof based on the fact
that \cstar local structures should not be considered as memory regions in
the generated C code.

In addition to buffers (stack-allocated arrays), \cstar uses local
structures in three ways: as local expressions, passed as an argument
to a function by value, and returned by a function. Here we claim that
it is always possible to not take them as memory accesses, except for
structures returned by a function: in the latter case, it is necessary
for the caller to allocate some space on its own stack and pass a
pointer to it to the callee, which will use this pointer to store its
result; then, the caller will read the result back from this memory
area. Thus, we claim that, at the level of CompCert Clight, the only
additional memory accesses due to local structures are structures
returned by value.

So we extend \cstar3 with the ability for functions to have several
arguments, all of which shall be passed at each call site (there shall
be no partial applications.)

\subsubsection{Structure return} \label{sec:struct-return}
To handle structure return, we also have to account for their memory
accesses by adding corresponding events in the trace. Instead of
directly adding the memory accesses and trying to prove both program
transformation and trace transformation at the same time, we will
first add new $\symread$ and $\symwrite$ events at function return,
without those events corresponding to actual memory accesses yet;
then, in a second pass, we will actually introduce the corresponding
new stack-allocated variables.

We assume given a function $\kw{FunResVar}$ such that for any list of
statements $ss$ and any variable $x$, $\kw{FunResVar}(ss, x)$ is a
local variable that does not appear in $ss$ and is distinct from
$\kw{FunResVar}(ss, x')$ for any $x' \not= x$.

Let $p$ be a program $p$ and $ss$ be an entrypoint list of statements,
so we define $\kw{FunResVar}(f, x) = \kw{FunResVar(ss', x)}$ if $f
(\_) \{ ss' \}$ is a function defined in $p$, and
$\kw{FunResVar}(\None, x) = \kw{FunResVar}(ss,x)$.

Then, we define \cstar4 as the language \cstar3 where the $\textsf{Ret}_2$
function return rule is replaced with two rules following
Figure~\ref{fig:cstar-4}, adding the fake $\symread$ and
$\symwrite$. We do not produce any such memory access event if the
result is discarded by the caller; thus, we also need to check in the
callee whether the caller actually needs the result. To prepare for
the second pass where this check will be done by testing whether the
return value pointer argument is null, we need to account for this
test in the event trace in \cstar4 as well.

\begin{figure*}
\begin{small}
  \begin{mathpar}
\inferrule* [Right=$\text{Ret}_4\text{Some}$]
{
  \eval{e}{(p,V)}=v \\
  \kw{FunResVar}(f', x) = x' \\
  \theta = \brt ; \symwrite\;(x', 0, []);\symread\;(x',0,[])
}
{
  p \vdash (S;(\None, V', t ~ x = \symhole; ss', f', A'), V, \ereturn\;e; ss, f, A) \step_{\theta}(S, V', t ~ x = v ; ss', f', A')
}

\\

\inferrule* [Right=$\text{Ret}_4\text{None}$]
{
  \eval{e}{(p,V)}=v
}
{
  p \vdash (S;(\None, V', \symhole; ss', f', A'), V, \ereturn\;e; ss, f, A) \step_{\brf}(S, V', ss', f', A')
}
  \end{mathpar}
\end{small}
\caption{\cstar4 Amended Configuration Reduction} \label{fig:cstar-4}
\end{figure*}

\begin{theorem}[\cstar3 to \cstar4 functional correctness]
  If $\sys(p, V, ss)$ is safe in \cstar3 and has unambiguous local
  variables, then it has the same behavior and trace as in \cstar4 with
  $\brt$, $\brf$, $\symread$ and $\symwrite$ events removed.
\end{theorem}
\begin{proof}
  With all such events removed, \cstar3 and \cstar4 are actually the same
  language.
\end{proof}

\begin{lemma}[\cstar4 invariants]
  The \cstar3 invariants of Lemma~\ref{lemma-cstar-2-invar}
  and~\ref{lemma-cstar-2-noninterference-invar} also hold on \cstar4.
\end{lemma}
\begin{proof}
  This is true because the invariants of
  Lemma~\ref{lemma-cstar-2-invar} actually do not depend on the traces
  produced; and it is obvious to prove that, if two executions have
  the same traces in \cstar4, then they have the same traces in \cstar3
  (because in \cstar3, some events are just removed.)
\end{proof}

\begin{theorem}[\cstar3 to \cstar4 noninterference]
  If $\sys(p, V_1, ss)$ and $\sys(p, V_2, ss)$ are safe in \cstar3, have
  unambiguous local variables, and produce the same traces in \cstar3, then
  they also produce the same traces in \cstar4.
\end{theorem}
\begin{proof}
  Two such executions actually make the same \cstar4 steps.
\end{proof}

Then, we define the $\kw{StructRet}$ structure return transformation
from \cstar4 to \cstar3 in Figure~\ref{fig:cstar-3-struct-ret}, thus removing
all structure returns from \cstar3 programs.

\begin{figure*}
\begin{small}
\[
\kw{StructRet}(\_, \kw{return}\; e, x) = \left\{
\begin{array}{ll}
  \eif{x}{ *[x] = e}{()} ; \kw{return}\;() & \text{if} ~ x \not= \None \\
  \kw{return}\; e & \text{otherwise}
\end{array}
\right.
\]

\[
\kw{StructRet}(f', t ~ x = f(e), \_) = \left\{
\begin{array}{ll}
  t ~ x'[1]; f(x', e); t ~ x = *[x'] & \text{if} ~ t ~ \text{is a} ~ \kw{struct} \\
  & \text{and} ~ x' = \kw{FunRetVar}(f', x) \\
  t ~ x = f(e) & \text{otherwise}
\end{array}
\right.
\]

\[
\kw{StructRet}(\_, f(e), \_) = \left\{
\begin{array}{ll}
  f(0, e) & \text{if the return type of } ~ f ~ \text{is a} ~ \kw{struct} \\
  f(e) & \text{otherwise}
\end{array}
\right.
\]

\newcommand{\ecfuntwoargs}[7]                {%
  \ensuremath{\kw{fun}\;#1\,(#2:#3,#4:#5):#6\,\{\;#7\}}%
}

\[
\kw{StructRet}(\ecfun{f}{x}{t}{t'}{ss}) = \left\{
\begin{array}{ll}
  \ecfuntwoargs{f}{r}{t'*}{x}{t}{\kw{unit}}{\{ \kw{StructRet}(f, ss, r) \} } &
  \text{if} ~ t' ~ \text{is a} ~ \kw{struct} \\
  & (r ~ \text{fresh}) \\
  \ecfun{f}{x}{t}{t'}{\{ \kw{StructRet}(f, ss, \None) \} } & \text{otherwise}
\end{array}
\right.
\]
\end{small}
\caption{\cstar4 to \cstar3 structure return
  transformation} \label{fig:cstar-3-struct-ret}
\end{figure*}

Then, the transformation back to \cstar3 exactly preserves the traces of
\cstar4 programs, so that we obtain both functional correctness and
noninterference at once:

\begin{theorem}[$\kw{StructRet}$ correctness]
  If $\sys(p, V, ss)$ is safe in \cstar4 and has unambiguous local
  variables, then it has the same behavior and trace as
  $\sys(\kw{StructRet}(p), V, \kw{StructRet}(ss, \None))$.
\end{theorem}
\begin{proof}
  Forward downward simulation, where the compilation invariant also
  involves block identifier renaming from
  Definition~\ref{def:cstar-3-rename-blocks} due to the new local
  arrays introduced by the transformed program.

  Each \cstar4 step is actually matched by the same \cstar3 step, except for
  function return and return from block: for the $\textsc{RetBlock}$
  rule, the simulation diagram has to stutter as many times as the
  level of block nesting in the source program before the actual
  application of a $\textsc{Ret}_4$ rule. Then, when the
  $\textsc{RetSome}_4$ rule applies, the trace events are produced by
  the transformed program, the $\brt$ and the $\symwrite$ events from
  within the callee, then the callee blocks are exited, and finally
  the $\symread$ event is produced from within the caller.

  Then, the diagram is turned into bisimulation since \cstar3 is
  deterministic.
\end{proof}

After a further hoisting pass, we can now restrict our study to those
\cstar3 programs with unambiguous local variables, functions with multiple
arguments, function-scoped local arrays, and no functions returning
structures.

\subsubsection{Events for accessing structure buffers} \label{sec:struct-events}

Now, we transform an access to one structure into the
sequence of accesses to all of its individual atomic (non-structure)
fields.

\ignore{
\textbf{NOTE:} we need to have some typing information added to the
$\symread$, $\symwrite$ events of \cstar and each of \cstar$n$: whenever we
read or write from a buffer, we need to mark the type of the data
read/written. This makes sense in the context of observing memory
accesses, since without this data, two programs accessing the same
memory location but reading data of different types, thus of
potentially different lengths, would still be considered as having the
same trace, which we would like to rule out.  So, in the remaining
parts of this document, we assume that $\symread$ and $\symwrite$
events carry the type of the data read or written.
}

Consider the following transformation for \cstar3 $\symread$ (and
similarly for $\symwrite$) events:

\[
\begin{array}{ll}
  \multicolumn{2}{l}{\llbracket \symread\;(f, i, x, j, \ls{fd}, t) \rrbracket} \\
  = & \llbracket \symread\; (f, i, x, j, \ls{fd};fd_1, t_1) \rrbracket \\
  ; & \dots \\
  ; & \llbracket \symread\; (f, i, x, j, \ls{fd};fd_n, t_n) \rrbracket \\
  \text{if} & t = \kw{struct} \{ fd_1: t_1, \dots, fd_n: t_n \} \\
\\
  \multicolumn{2}{l}{\llbracket \symread\;(f, i, x, j, \ls{fd}, t) \rrbracket} \\
  = & \symread\; (f, i, x, j, \ls{fd}, t) \\
  \multicolumn{2}{l}{\text{otherwise}}
\end{array}
\]

Then, let \cstar5 be the \cstar3 language obtained by replacing each
$\symread$, $\symwrite$ event with its translation. Then, it is easy
to show the following:
\begin{lemma}[\cstar3 to \cstar5 correctness]
  Let $p$ be a \cstar3 program. Then, $p$ has a trace $t$ in \cstar5 if, and
  only if, there exists a trace $t'$ such that $p$ has trace $t'$ in
  \cstar3 and $\llbracket t' \rrbracket = t$.
\end{lemma}

Thus, this trace transformation preserves functional correctness; and,
although this trace transformation is not necessarily injective (since
it is not possible to disambiguate between an access to a 1-field
structure and an access to its unique field), noninterference is also
preserved.

After such transformation, all $\symread$ and $\symwrite$ events now
are restricted to atomic (non-structure) types.

\subsubsection{Local structures} \label{sec:local-struct-detail}

Now, we are removing all local structures, in such a way that the only
remaining structures are those of local arrays, and all structures are
accessed only through their atomic fields. In particular, we are
replacing every local (non-array) variable $x$ of type struct with the
sequence of variable names $x\_\ls{fds}$ for all field name sequences
$\ls{fds}$ valid from $x$ such that $x.\ls{fds}$ is of
non-$\kw{struct}$ type. (We omit the details as to how to construct
names of the form $x\_\ls{fds}$ so that they do not clash with other
variables; at worst, we could also rename other variables to avoid
clashes as needed.)

Using our benchmark in Figure~\ref{fig:struct-erase-benchmark}, with C
code after structure erasure in Figure~\ref{fig:struct-erase-after},
on a 4-core Intel Core i7 1.7 GHz laptop with 8 Gb RAM, structure
erasure saves 20\% time with CompCert 2.7.

If we assume that we know about the type of a \cstar expression, then it
can be first statically reduced to a normal form as in
Figure~\ref{fig:cstar-expr-struct-erase}.

\begin{figure*}
\begin{scriptsize}
  \begin{mathpar}
\inferrule*
    [Right=Int]
    {~}
    {
      \Gamma \vdash n \downarrow^{\kw{Int}} n
    }

\inferrule*
    [Right=Var]
    {
      (x : t) \in \Gamma
    }
    {
      \Gamma \vdash x \downarrow^t x
    }

\inferrule*
    [Right=PtrAdd]
    {
      \Gamma \vdash e_1 \downarrow^{t*} e_1' \\
      \Gamma \vdash e_2 \downarrow^{\kw{int}} e_2'
    }
    {
      \Gamma \vdash e_1 + e_2 \downarrow^{t*} e_1' + e_2'
    }

\\

\inferrule*
    [Right=PtrFd]
    {
      \Gamma \vdash e \downarrow^{\kw{struct} \{ fd : t; \dots \} *} e' \\
    }
    {
      \Gamma \vdash \eptrfd{e}{fd} \downarrow^{t*} \eptrfd{e'}{fd}
    }
    
\inferrule*
    [Right=StructFieldName]
    {
      \Gamma \vdash e \downarrow^{\kw{struct} \{ fd : t ; \dots \}} x.\ls{fds} \\
      t ~ \text{is a} ~ \kw{struct}
    }
    {
      \Gamma \vdash e.fd \downarrow^t x.\ls{fds}.fd
    }

\\

\inferrule*
    [Right=ScalarFieldName]
    {
      \Gamma \vdash e \downarrow^{\kw{struct} \{ fd : t ; \dots \}} x.\ls{fds} \\
      t ~ \text{is not a} ~ \kw{struct}
    }
    {
      \Gamma \vdash e.fd \downarrow^t x\_\ls{fds}\_fd
    }

    \\

\inferrule*
    [Right=FieldProj]
    {
      \Gamma \vdash e \downarrow^{\kw{struct} \{ fd : t ; \dots \}} \{ f = e' ; \dots \}
    }
    {
      \Gamma \vdash e.fd \downarrow^t e'
    }

\inferrule*
    [Right=Struct]
    {
      \ls{\Gamma \vdash e_i \downarrow^{t_i} e'_i}
    }
    {
      \Gamma \vdash \{ \ls{fd_i = e_i} \} \downarrow^{\kw{struct} \{ \ls{fd_i : t_i} \} } \{ \ls{fd_i = e'_i} \}
    }
\end{mathpar}
\end{scriptsize}
  \caption{\cstar Structure Erasure: Expressions}
  \label{fig:cstar-expr-struct-erase}
\end{figure*}

\begin{lemma}[\label{lemma-cstar-expr-struct-erase-correct} \cstar structure erasure in expressions: correctness]
  For any value $v$ of type $t$, if $\llbracket e \rrbracket_{(p,V)} =
  v$ and $\Gamma \vdash e \downarrow^t e'$ and $V'$ is such that:
  \begin{itemize}
  \item for any $(x': t') \in \Gamma$,  $V(x')$ exists, is of type $t'$ and is equal to $V(x)$
  \item for any $(x': t') \in \Gamma$ that is a struct and for any $\ls{fds}$ such that $x.\ls{fds}$ is not a struct, then $V'(x\_\ls{fds}) = V(x)(\ls{fds})$
  \end{itemize}
  Then, $\llbracket e' \rrbracket_{(p,V')} = v$
\end{lemma}
\begin{proof}
  By structural induction on $\downarrow$.
\end{proof}

\begin{definition}[\cstar expression without structures]
  A \cstar expression $e$ is said to be of type $t$ without structures if, and only
  if, one of the following is true:
  \begin{itemize}
  \item $e$ contains neither a structure field projection nor a structure expression
  \item $e$ is of the form $x.\ls{fds}$ where $x$ is a variable such
    that $x.\ls{fds}$ is of $\kw{struct}$ type
  \item t is of the form $\{ \ls{fd_i : t_i} \}$ and $e$ is of the form $\{ \ls{fd_i = e_i} \}$ where for each $i$, $e_i$ is of type $t_i$ without structures
  \end{itemize}
\end{definition}

\begin{lemma}[\label{lemma-cstar-expr-struct-erase-shape}\cstar expression reduction: shape]
  If $\Gamma \vdash e \downarrow^t e'$, then $e'$ is of type $t$
  without structures.
\end{lemma}
\begin{proof}
  By structural induction on $\downarrow$.
\end{proof}

Then, once structure expressions are reduced within an expression
computing a non-structure value, we can show that evaluating such a
reduced expression no longer depends on any local structures:

\begin{lemma} \label{lemma-cstar-expr-struct-erase-non-struct}
  If $\llbracket e \rrbracket_{(p,V)} = v$ for some value $v$, and $e$
  is of type $t$ without structures, and $t$ is not a $\kw{struct}$
  type, then, for any variable mapping $V'$ such that $V'(x) = V(x)$
  for all variables $x$ of non-$\kw{struct}$ types, $\llbracket e
  \rrbracket_{(p,V')} = v$.
\end{lemma}
\begin{proof}
  By structural induction on $\llbracket e \rrbracket_{(p,V)}$.
\end{proof}

Now, we take advantage of this transformation to transform \cstar5
statements into \cstar3 statements without structure assignments.  This
$\downarrow$ translation is detailed in
Figure~\ref{fig:cstar-5-to-cstar-3}

In particular, each function parameter of structure type passed by
value is replaced with its recursive list of all non-structure
fields.\footnote{Our solution, although semantics-preserving as we
  show further down, yet causes ABI compliance issues. Indeed, in the
  System V x86 ABI, structures passed by value must be replaced not
  with their fields, but with their sequence of bytes, some of which
  may correspond to padding related to no field of the original
  structure. CompCert does support this feature but as an
  \textbf{unverified} elaboration pass over source C code. So, we
  should investigate whether we really need to expose functions taking
  structures passed by value at the interface level.}

\begin{figure*}
\begin{footnotesize}
  \begin{mathpar}
    \inferrule*
    [Right=ReadScalar]
    {
      t ~ \text{not a} ~ \kw{struct} \\
      \Gamma \vdash e \downarrow^{t*} e'
    }
    {
      p,\Gamma \vdash t ~ x = \eread{e} \downarrow t ~ x = \eread{e'} \\\\
      \Gamma[x] \leftarrow t
    }

    \quad
    \quad
    
    \inferrule*
    [Right=WriteScalar]
    {
      t ~ \text{not a} ~ \kw{struct} \\\\
      \Gamma \vdash e_1 \downarrow^{t*} e'_1 \\
      \Gamma \vdash e_2 \downarrow^{t} e'_2 \\
    }
    {
      p,\Gamma \vdash \ewrite{e_1}{e_2} \downarrow \ewrite{e'_1}{e'_2}
    }

    \\    
    
\inferrule*
    [Right=ReadStruct]
    {
      t = \kw{struct} \{ \ls{\mathit{fd_i} : t_i} \} \\
      \Gamma \vdash e \downarrow^{t*} e' \\
      (x, \_) \not\in \Gamma \\
      (x', \_) \not\in \Gamma \\\\
      \ls{p,\Gamma \cup (x' : t*) \cup (x\_{\mathit{fd}_i} : t_i) \vdash x\_{\mathit{fd}_i} = \eread{\eptrfd{x'}{\mathit{fd}_i}} \downarrow \mathit{ss}_i}
    }
    {
      p,\Gamma \vdash t ~ x = \eread{e} \downarrow t* ~ x' = e'; \ls{\mathit{ss}_i} \\
      \Gamma[x] \leftarrow t
    }

    \\    
    
\inferrule*
    [Right=WriteStruct]
    {
      t = \kw{struct} \{ \ls{\mathit{fd_i} : t_i} \} \\
      \Gamma \vdash e_1 \downarrow^{t*} e'_1 \\
      \Gamma \vdash e_2 \downarrow^{t} e'_2 \\
      (x', \_) \not\in \Gamma \\\\
      \ls{p,\Gamma \cup (x' : t*) \vdash \ewrite{\eptrfd{x'}\mathit{fd}_i}{e'_2.\mathit{fd}_i} \downarrow \mathit{ss}_i}
    }
    {
      p,\Gamma \vdash \ewrite{e_1}{e_2} \downarrow t* ~ x' = e'_1 ; \ls{\mathit{ss}_i}
    }

    \\

    \inferrule*
    [Right=Ret]
    {
      t ~ \text{not a} ~ \kw{struct} \\
      \Gamma \vdash e \downarrow^{t} e'
    }
    {
      p,\Gamma \vdash \kw{return} ~ e \downarrow \kw{return} ~ e'
    }

\inferrule* [Right=Block]
{
  p,\Gamma \vdash ss \downarrow ss'
}
{
  p,\Gamma \vdash \{ ss \} \downarrow \{ ss' \}
}

    \inferrule* [Right=Call]
    {
      p(f)= \kw{fun} (\ls{\_ : t_i}) : t \{ \_ \} \\
      t ~ \text{not a} ~ \kw{struct} \\
      p,\Gamma \vdash (\ls{t_i}, \mathit{el}) \downarrow \mathit{el}'
    }
    {
      p,\Gamma \vdash t\;x=f\;(\mathit{el}) \downarrow t\;x = f(\mathit{el}') \\
      \Gamma[x] \leftarrow t
    } 

    \\
    
    \inferrule*
    [Right=ArgNil]
    {~}
    {
      p,\Gamma \vdash ([],[]) \downarrow []
    }

    \inferrule*
        [Right=ArgCons]
        {
          p,\Gamma \vdash (t, e) \downarrow \mathit{el}_1 \\
          p,\Gamma \vdash (\mathit{tl}, \mathit{el}) \downarrow \mathit{el}_2
        }
        {
          p,\Gamma \vdash (t; \mathit{tl}, e; \mathit{el}) \downarrow \mathit{el}_1; \mathit{el}_2
        }

        \\
        
    \inferrule*
        [Right=ScalarArg]
        {
          t ~ \text{not a} ~ \kw{struct} \\
          \Gamma \vdash e \downarrow^t e'
        }
        {
          p, \Gamma \vdash (t, e) \downarrow [e']
        }

    \inferrule*
        [Right=StructArg]
        {
          t = \kw{struct} \{ \ls{\mathit{fd}_i : t_i} \} \\
          \Gamma \vdash e \downarrow^t e' \\\\
          \ls{p,\Gamma \vdash (t_i, e'.\mathit{fd}_i) \downarrow \mathit{el}_i}
        }
        {
          p, \Gamma \vdash (t, e) \downarrow \ls{\mathit{el}_i}
        }

    \\

\inferrule* [Right=If]
{
  t ~ \text{not a} ~ \kw{struct} \\
  \Gamma \vdash e \downarrow^t e' \\
  \ls{p,\Gamma \vdash ss_i \downarrow ss'_i}
}
{
  p,\Gamma \vdash \eif{e}{ss_1}{ss_2} \downarrow \eif{e'}{ss_1'}{ss_2'}
} 

\\    
        
    \inferrule*
    [Right=ScalarParam]
    {
      t ~ \text{not a} ~ \kw{struct} \\
    }
    {
      (x : t) \downarrow (x : t)
    }

    \inferrule*
    [Right=StructParam]
    {
      t = \kw{struct} \{ \ls{\mathit{fd}_i : t_i} \} \\\\
      \ls{ (x\_\mathit{fd}_i : t_i) \downarrow \mathit{vt}_i }
    }
    {
      (x : t) \downarrow \ls{\mathit{vt}_i}
    }
    \\

    \inferrule* [Right=Fun]
    {
      \mathit{vt} \downarrow \mathit{vt}' \\
      p, (\mathit{vt} \cup \ls{x_i : t_i*}) \vdash \mathit{ss} \downarrow \mathit{ss}'
    }
    {
      \kw{fun} (\mathit{vt}) : t \{ \ls{t_i ~ x_i[\_] = \{ \} } ; \mathit{ss} \}
      \downarrow
      \kw{fun} (\mathit{vt}') : t \{ \ls{t_i ~ x_i[\_] = \{ \} } ; \mathit{ss}' \}
    }
\end{mathpar}
\end{footnotesize}
  \caption{\cstar5 to \cstar3 Structure Erasure: Statements}
  \label{fig:cstar-5-to-cstar-3}
\end{figure*}

\begin{theorem}[\cstar5 to \cstar3 structure erasure: shape]
  If $p \downarrow p'$ following Figure~\ref{fig:cstar-5-to-cstar-3},
  then $p'$ no longer has any variables of local structure type, and no
  longer has any structure or field projection expressions.
\end{theorem}
\begin{proof}
  By structural induction over $\downarrow$, also using
  Lemma~\ref{lemma-cstar-expr-struct-erase-shape}.
\end{proof}

\begin{theorem}[\cstar5 to \cstar3 structure erasure: correctness]
  If $p$ is a \cstar5 program (that is, syntactically, a \cstar program with
  unambiguous local variables, no block-scoped local arrays other than
  function-scoped, and no functions returning structures) such that
  $p$ is safe in \cstar5 and $p \downarrow p'$ following
  Figure~\ref{fig:cstar-5-to-cstar-3}, then $p$ and $p'$ have the same
  execution traces.
\end{theorem}
\begin{proof}
  Forward downward simulation where one \cstar5 step triggers one or
  several \cstar3 steps. Then, determinism of \cstar3 turns this forward
  downward simulation into bisimulation.

  The compilation invariant is as follows: the code fragments are
  translated using $\downarrow$, and variable maps $V$ in source \cstar5
  vs. their compiled \cstar3 counterparts $V'$ follow the conditions of
  Lemma~\ref{lemma-cstar-expr-struct-erase-correct}, also using
  Lemma~\ref{lemma-cstar-expr-struct-erase-non-struct}. Memory states
  $M$ are exactly preserved, as well as the structure of the stack.
\end{proof}

Then, after a further $\alpha$-renaming pass, we obtain
a \cstar3 program that no longer has any local (non-stack-allocated)
structures at all, and where all memory accesses are of non-structure
type. The shape of this program is now suitable for translation to a
CompCert Clight program in a straightforward way, which we describe
in the next subsection.

\subsection{Generation of CompCert Clight code} \label{sec:clight-gen}

Recall that going from \cstar3 (with abstract pointer events) back to \cstar2
(with concrete pointer events) is possible thanks to the fact that
Lemma~\ref{lemma-cstar-2-to-cstar-3-correct} and
Lemma~\ref{lemma-cstar-2-to-cstar-3-noninterference} are actually
equivalences.

Recall that a \cstar$n$ transition system is of the form $\sys(p, V,
\mathit{ss})$ where $p$ is a list of functions\footnote{and global
  variables, although the semantics of \cstar says nothing about how to
  actually initially allocate them in memory}, $\mathit{ss}$ is a list
of \cstar statements with undeclared local variables, the values of which
shall be taken from the map $V$. $\mathit{ss}$ is actually taken as
the entrypoint of the program, and $V$ is deemed to store the initial
values of secrets, ensuring that $p$ and $\mathit{ss}$ are
syntactically secret-independent.

In CompCert Clight, it is not nominally possible to start with a set
of undeclared variables and a map to define them. So, when translating
the \cstar entrypoint into Clight, we have to introduce a
\emph{secret-independent} way of representing $V$ and how they are
read in the entrypoint. Fortunately, CompCert introduces the notion of
\emph{built-in functions}, which are special constructs whose
semantics can be customized and that are guaranteed to be preserved by
compilation down to the assembly.

Thus, we can populate the values of local non-stack-allocated
variables of a \cstar entrypoint by uniformly calling builtins in Clight,
and only the semantics of those builtins will depend on secrets, so
that the actually generated Clight code will be syntactically
secret-independent.

Translating \cstar2 expressions with no structures or structure field
projections into CompCert Clight is straightforward, as shown in
Figure~\ref{fig:cstar-2-to-clight-expr}. For any \cstar2 expression $e$ of
type $t$, assuming that $A$ is a set of local variables to be
considered as local arrays, we define $\mathbb C^t_A(e)$ to be the
compiled Clight expression corresponding to $e$.

\begin{figure}
  \[
  \begin{array}{lcl}
    \mathbb C^{\kw{int}}_A (n) & = & n \\
    \mathbb C^{\kw{unit}}_A (()) & = & 0 \\
    \mathbb C^t_A (x) & = & \&x \\
     & \text{if} & x \in A \\
    \mathbb C^t_A (x) & = & \_x \\
     & \text{if} & x \not\in A \\
    \mathbb C^{t*}_A (e_1 + e_2) & = & \mathbb C^{t*}_A (e_1) +_t \mathbb C^{\kw{int}}_A (e_2) \\
    \mathbb C^{t*}_A (\eptrfd{e}{fd}) & = & \&(*\mathbb C^{t'*}_A (e) ._{t'}\mathit{fd}) \\
    & \text{if} & t' = \kw{struct}\{ \mathit{fd} : t, \dots \}
  \end{array}
  \]
  \caption{\cstar2 to Clight: Expressions}
  \label{fig:cstar-2-to-clight-expr}
\end{figure}

\begin{lemma}
  Let $V$ be a \cstar2 local variable map, and $A$ be a set of local
  variables to be considered as local arrays. Assume that, for all $x
  \in A$, there exists a block identifier $b$ such that $V(x) = (b,
  0)$, and define $V'(x)$ be such block identifier $b$. Then, define
  $\_V'(\_x) = V(x)$ for all $x \not\in A$.

  Then, for any expression $e$ with no structures or structure
  projections, $\kw{rv}(\mathbb C^t_A (e), (p, V', \_V')) = \llbracket
  e \rrbracket_{(p, V)}$.
\end{lemma}
\begin{proof}
  By structural induction on $e$. \ignore{\textbf{TODO:} global variables.}
\end{proof}

Translating \cstar2 statements with no local array declarations, no read
or write of structure type and no functions returning structures into
Clight is straightforward as well, as shown in
Figure~\ref{fig:cstar-2-to-clight-stmt}.

\begin{figure}
  \[
  \arraycolsep=1pt
  \begin{array}{lcl}
    \mathbb C_A (t ~ x = e) & = & t ~ x = C^t_A(e) \\
    \mathbb C_A (t ~ x = f(\ls{e_i})) & = & t ~ x = f(\ls{C^{t_i}_A(e_i)}) \\
    & & \text{if} ~ f ~ \text{is} ~ \kw{fun}(\ls{(\_ : t_i)}) : \_ \{ \_ \} \\
    \mathbb C_A (t ~ x = \eread{e}) & = & \clannot(\symread,t,e) ; t ~ x = [*\mathbb C^{t*}_A (e) ] \\
    \mathbb C_A (\ewrite{e_1}{e_2}) & = & \clannot(\symwrite,t,e) ; *\mathbb C^{t*}_A(e_1) = \mathbb C^{t}_A(e_2) \\
    \mathbb C_A (\eifthenelse{e}{ss_1}{ss_2}) & = & \eifthenelse{\mathbb C^t_A (e) }{\mathbb C_A(ss_1)}{\mathbb C_A(ss_2)} \\
    & & \text{for some} ~ t ~ \text{not} ~ \kw{struct} \\
    \mathbb C_A ( \{ ss \} ) & = & \{ \mathbb C_A (ss) \} \\
    \mathbb C_A ( \kw{return} ~ e ) & = & \kw{return} ~ \mathbb C^t_A(e) \\
  \end{array}
  \]
  \caption{\cstar2 to Clight: Statements}
  \label{fig:cstar-2-to-clight-stmt}
\end{figure}

Let $p$ be a \cstar2 program, and $ss$ be a \cstar2 entrypoint sequence of
statements. Assume that $p$ and $ss$ have unambiguous local variables,
no functions returning structures, no local arrays other than
function-scoped arrays, and no local structures other than local
arrays. Further assume that $p$ has no function called $\kw{main}$,
and no function with the same name as a built-in function. Then, we
can define the compiled CompCert Clight program $\mathbb C(p, ss)$ as
in Figure~\ref{fig:cstar-2-to-clight-prog}.
\begin{figure}
\[
\begin{array}{lcl}
  \mathbb C(p, ss)(f) & = & \kw{fun} ~ (\ls{x : t}) : t' \{ \ls{t_i ~ x_i[n_i]} ; \mathbb C_{\ls{x_i}}(ss') \} \\
  & \text{if} & p(f) = \kw{fun} ~ (\ls{x : t}) : t' \{ \ls{t_i ~ x_i[n_i]} ; ss' \} \\
  \mathbb C(p, ss)(\kw{main}) & = & \kw{fun} () : \kw{int} \{ \\
  & & ~ \ls{t_i ~ x_i[n_i]} ; \\
  & & ~ \ls{\_y_i = \kw{get}\_y_i()} ; \\
  & & ~ \mathbb C_{\ls{x_i}}(ss') \\
  & & \} \\
  & \text{if} & \mathit{ss} = \{ \ls{t_i ~ x_i[n_i]} ; ss' \} \\
  & & \text{with free variables} ~ \ls{y_i}
\end{array}
\]
\caption{\cstar2 to Clight: Program and Entrypoint}
\label{fig:cstar-2-to-clight-prog}
\end{figure}

\begin{theorem}[\cstar2 to Clight: correctness]
  If $\sys(p, V, ss)$ is safe in \cstar2, then it has the same execution
  trace as $\mathbb C(p, ss)$ in Clight, when the semantics of the
  built-in functions $\kw{get}\_x$ are given by $V$.
\end{theorem}
\begin{proof}
  Forward downward simulation where one \cstar2 step corresponds to one or
  several Clight steps. Then, since Clight is deterministic, the
  forward downward simulation diagram is turned into a bisimulation.

  The structure of the Clight stack is the same as in the \cstar2 stack,
  and the values of variables are the same, as well as the memory
  block identifiers. The only change is in the Clight representation
  of \cstar2 structure values . A \cstar2 memory state $M$ is said to
  correspond to a Clight memory state $M'$ if, for any block
  identifier $b$, for any array index $i$, and for any field sequence
  $\mathit{fds}$ leading to a non-structure value, $M'(b, n +
  \kw{offsetof}(\mathit{fds})) = M(b, n)(\mathit{fds})$.
\end{proof}

Thus, since both \cstar2 and Clight records all memory accesses in their
traces, this theorem entails both functional correctness and
preservation of noninterference.

\newpage
\section{Proof of the secret independence theorem}

\begin{small}
\begin{verbatim}
---------
Notations
---------

t ::= int | unit | { f_i : t_i } | buf t | a

v ::= x | n | () | { f_i = v_i } | (b, n) | stuck

P ::= Top level functions    (same as lp)

G ::= x_i:t_i    (type enviroment)

S ::= (b_i, n_i):t_i    (store typing)

G_s = a_i, f_j:t_j -> t'_j    (type variables and signature of functions that
                               manipulate secrets)

G_f = x_i:a_i    (free variables of secret types)

q ::= a_i |-> t_i    (instantiations of type variables)

p ::= x_i |-> v_i    (substitutions of secret variables)


-----------
Judgments:
-----------

Remark: In the following judgments, we use the program P itself in the typing context,
rather than just the signatures, as we do in the main paper (G_p). This is equivalent
since the rule [T-App] below, which looks up the function signature in P, does not
rely on the function body.

P; S; G_s, G |- e : t    (expression typing)

P; S; G_s |- P'    (program typing)

G_s |- t    (type well-formedness)

G_s |- q    (well-formedness of type variables instantiations)

S; G_s, G |-_q p    (well-typed substitutions)

P; S; G_s, G |- H    (heap typing)



** Expression typing    P; S; G_s, G |- e : t


     G(x) = t    G_s |- t
    --------------------------- [T-Var]
     P; S; G_s, G |- x : t


    -------------------------- [T-Int]
     P; S; G_s, G |- n : int


    ---------------------------- [T-Unit]
     P; S; G_s, G |- () : unit


    -------------------------------- [T-Buf]
     P; S; G_s, G |- (b, n) : S(b)


     P; S; G_s, G       |- e1 : buf t1
     P; S; G_s, G       |- e2 : int
     P; S; G_s, G, x:t1 |- e : t
    ---------------------------------------------------- [T-Readbuf]
     P; S; G_s, G |- let x:t1 = readbuf e1 e2 in e : t


     P; S; G_s, G |- e1 : buf t1
     P; S; G_s, G |- e2 : int
     P; S; G_s, G |- e3 : t1
     P; S; G_s, G |- e  : t
    ----------------------------------------------------- [T-Writebuf]
     P; S; G_s, G |- let _ = writebuf e1 e2 e3 in e : t

     P; S; G_s, G          |- e1 : t
     P; S; G_s, G, x:buf t |- e2 : t2
    ----------------------------------------------------- [T-Newbuf]
     P; S; G_s, G |- let x = newbuf n (e1:t) in e2 : t2


     P; S; G_s, G |- e1 : buf t
     P; S; G_s, G |- e2 : int
    --------------------------------------- [T-Subbuf]
     P; S; G_s, G |- subbuf e1 e2 : buf t


     P; S; G_s, G |- e : t
    ---------------------------------- [T-Withframe]
     P; S; G_s, G |- withframe e : t


     P; S; G_s, G |- e : t
    ---------------------------- [T-Pop]
     P; S; G_s, G |- pop e : t


     P; S; G_s, G |- e   : int
     P; S; G_s, G |- e_i : t
    ------------------------------------------- [T-If]
     P; S; G_s, G |- if e then e1 else e2 : t


     P(d) = fun (x:t1) -> _:t \/ G_s(d) = t1 -> t
     P; S; G_s, G      |- e1 : t1
     P; S; G_s, G, x:t |- e2 : t2
    --------------------------------------------- [T-App]
     P; S; G_s, G |- let x:t = d e1 in e2 : t2


     P; S; G_s, G      |- e1 : t
     P; S; G_s, G, x:t |- e2 : t2
    ------------------------------------------ [T-Let]
     P; S; G_s, G |- let x:t = e1 in e2 : t2


     P; S; G_s, G |- e_i : t_i
    ------------------------------------------------ [T-Rec]
     P; S; G_s, G |- { f_i : e_i } : { f_i : t_i }


     P; S: G_s, G|- e : { f_i : t_i }
    ----------------------------------- [T-Proj]
     P; S; G_s, G |- e.f_i : t_i


** Program typing   P; S; G_s |- P'


    ----------------------- [P-Emp]
     P; S; G_s |- .


     G_s             |- t_i
     P; S; G_s, y:t1 |- e : t2
     P; S; G_s       |- P'
    ---------------------------------------------- [P-Fun]
     P; S; G_s |- let d = fun(y:t1) -> e:t2, P'


-------------------------------------
Notion of equivalence-modulo-secrets
-------------------------------------

Equivalence of values:
----------------------

A lax relation typing-wise.

V_int           = { (n, n) }
V_unit          = { ( (), () ) }
V_{ f_i : t_i } = { ({f_i = v_i}, {f_i = v_i'}) | forall i. (v_i, v_i') \in V_(t_i) }
V_(buf t)       = { ( (b, n), (b, n) ) }
V_a             = { (v1, v2) }


Equivalence of heaps:
---------------------

Again, a lax relation typing-wise. Further, S can contain type variables.

S |- H1 =eq= H2 iff

1. dom H1 = dom H2

2. forall (b, n) \in dom H1.
     let buf t = S (b, n)
     forall 0 <= i < n. (H1[(b, n)][i], H2[(b, n)][i]) \in V_t


Gen function for values:
-------------------------

Let v1 is closed, v2 is closed, and (v1, v2) \in V_t (t can have type variables)

Gen v1 v2 t = v, p1, p2, G


Gen n n int    = (n, ., ., .)

Gen () () unit = ((), ., ., .)

Gen { f_i = v_i } { f_i = v'_i } { f_i : t_i } = ({f_i = v''_i}, p1, p2, G)
  where forall i. Gen v_i v'_i t_i = (v''_i, p1_i, p2_i, G_i)
  and   p1 = compose_for_all_i p1_i
  and   p2 = compose_for_all_i p2_i
  and   G  = compose_for_all_i G_i

  (compose is a simple append in each case)

Gen (b, n) (b, n) (buf t) = ((b, n), ., ., .)

Gen v1 v2 a = (x, x |-> v1, x |-> v2, x:a)


Gen function for heaps:
------------------------

Let H1 is closed, H2 is closed, and S |- H1 =eq= H2

Gen H1 H2 S = (H, p1, p2, G)

s.t. dom H = dom H1 (= dom H2)
     H[(b, n)][i] = v,
     where Gen H1[(b, n)][i] H2[(b, n)][i] t = v, p1_b_n_i, p2_b_n_i, G_b_n_i
     and p1 = compose of all p1_b_n_i and p2 = compose of all p2_b_n_i
     and G = compose of all G_b_n_i

--------------------------------------------------------------------------------------
Remark: Gen does not preserve sharing, but we recover sharing after the substitition.
Further, substitution is just a proof device.
--------------------------------------------------------------------------------------


Assumption on secrets manipulating functions:
----------------------------------------------

If

(1) G_s (f) = t1 -> t2, where G_s |- t1 and G_s |- t2  (f has type t1 -> t2
                                                        in the signature,
                                                        t1 and t2 can have
                                                        type variables)

(2) S |- H1 =eq= H2    (H1 and H2 are closed heaps (see 6, 7 below)
                        and related under S)

(3) P; S[q]; G_s[q] |- v1 : t1[q]    (v1 is a closed value of type t1[q])

(4) P; S[q]; G_s[q] |- v2 : t1[q]    (so is v2)

(5) (v1, v2) \in V_t1    (v1 and v2 are related)

(6) P; S[q]; G_s[q] |- H1    (H1 is well-typed heap, and closed)

(7) P; S[q]; G_s[q] |- H2    (so is H2)

Then

(a) P, P_s |- (H1, (P_s (f)) v1) -m->tr (H1', v1')    (f v1 terminates in m steps)

(b) P, P_s |- (H2, (P_s (f)) v2) -n->tr (H2', v2')    (f v2 terminates in n steps
                                                       with same trace as f v1)

(c) S' |- H1' =eq= H2'    (output heaps are equivalent modulo secrets)

(d) P; S'[q]; G_s[q] |- v1' : t2[q]

(e) P; S'[q]; G_s[q] |- v2' : t2[q]

(f) (v1', v2') \in V_t2    (v1' and v2' are equivalent modulo secrets)

(g) P; S'[q]; G_s[q] |- H1'

(h) P; S'[q]; G_s[q] |- H2'


----------------------------------------
Remark: a separate lemma can show m = n
----------------------------------------


-------------------------------------
Lemma: Gen-function-lemma-for-values
-------------------------------------

If

(1) G_s |- t

(2) G_s |- q

(3) P; S[q]; G_s |- v1 : t[q]

(4) P; S[q]; G_s |- v2 : t[q]

(5) (v1, v2) \in V_t

Then

(a) Gen v1 v2 t = v, p1, p2, G

(b) P; S; G_s, G |- v : t

(c) S; G_s, G |-_q p1

(d) S; G_s, G |-_q p2


------------------------------------
Lemma: Gen-function-lemma-for-heaps
------------------------------------

If

(1) G_s |- q

(2) P; S[q]; G_s |- H1

(3) P; S[q]; G_s |- H2

(4) S |- H1 =eq= H2

Then

(a) Gen H1 H2 S = (H, p1, p2, G)

(b) P; S; G_s, G |- H

(c) S; G_s, G |-_q p1

(d) S; G_s, G |-_q p2


----------------------------
Substitution-lemma-for-Low*
----------------------------

If

(1) P; S; G_s, G, x:t |- e : t'

(2) P; S; G_s, G |- v : t

Then

(a) P; S; G_s, G |- e[v/x] : t'


-----------------------
Lemma: Value-inversion
-----------------------

(a) If P; S; G_s |- v : int,   then v = n.

(b) If P; S; G_s |- v : buf t, then v = (b, n)


--------------------------
Subject reduction of Low*
--------------------------

If

(1) P; S; G_s |- P    (program is well typed)

(2) P; S; G_s, G |- e : t    (e has type t)

(3) .; S[q]; G_s[q] |-  P_s    (P_s is a well typed implementation
                                of the interface G_s)

(4) P, P_s |- H, e -->tr H', e'

(5) P; S; G_s |- H

Then, exists S' \superset S s.t.

(a) P; S'; G_s, G |- e': t

(b) P; S'; G_s |- H'


Proof: Low* type system is a simply-typed type system, so the proofs go through
in an unsurprising manner.

Also, the type system ensures subject reduction, but does not ensure progress
since it does not track the buffer bounds.


-----------------------------------------
Lemma: Equivalence-modulo-secrets-values
-----------------------------------------

If

(1) P; S; G_s, G_f |- v : t    (v is well-typed against the secrets interface)

(2) S; G_s, G_f |-_q p    (p is a well-typed substitution of secret variables under
instantiations of type variables q)

(3) S; G_s, G_f |-_q p1    (p1 is also a well-typed substitution)

Then

(a) (v[p], v[p1]) \in V_t    (the substituted values are related)


----------------------------------------
Lemma: Equivalence-modulo-secrets-heaps
----------------------------------------

(Explanation same as the above lemma)

If

(1) P; S; G_s, G_f |- H

(2) S; G_s, G_f |-_q p    (p is a well-typed substitution of secret variables under
instantiations of type variables q)

(3) S; G_s, G_f |-_q p1    (p1 is also a well-typed substitution)

Then

(a) S |- H[p] =eq= H[p1]


-------------------------------
Lemma: Substituted-value-typing
-------------------------------

If

(1) P; S; G_s, G_f |- v : t    (v is well-typed)

(2) S; G_s, G_f |-_q p    (substitution p is well-typed
                           with type variables instantiations q)

Then

(a) P; S[q]; G_s[q] |- v[p] : t[q]    (substituted value is well-typed and closed)


-------------------------------
Lemma: Substituted-heap-typing
-------------------------------

If

(1) P; S; G_s, G_f |- H    (H is well-typed)

(2) S; G_s, G_f |-_q p    (substitution p is well-typed
                           with type variables instantiations q)

Then

(a) P; S[q]; G_s[q] |- H[p]    (substituted heap is well-typed and closed)


-------------------------------------
Lemma: Non-secret-value-substitution
-------------------------------------

If

(1) S; G_s, G_f |-_q p

(2) P; S; G_s, G_f |- v[p] : t

Then

(a) If t = int, then v = n.

(b) If t = buf t, then v = (b, n)


-----------------------------------
Theorem: Secret-independent-traces
-----------------------------------

(Also assume G_s |- q, G_s and q do not vary with program execution)

If

(1) P; S; G_s, G_f |- (H, e) : t    (only secret type variables are free)

(2) S; G_s, G_f |-_q p    (substitution of secret typed variables is well-typed)

(3) S; G_s, G_f |-_q p1    (p1 is some other substitution that is also well-typed)

(4) .; S[q]; G_s[q] |-  P_s    (P_s is a well typed implementation
                                of the interface G_s)

Then either e is a value, or
exists m > 0, n > 0,
exists S' \superseteq S, G_f' \superseteq G_f, H', e', p' \superseteq p,
       p1' \superseteq p1 s.t.

(a) P, P_s |- (H, e)[p] -m->tr (H1, e1)    (it takes m steps, emitting the trace tr)

(b) (H', e')[p'] = (H1, e1)    (the resulting H1 and e1 are some template
                                with some substitution)

(c) P; S'; G_s, G_f' |- (H', e') : t    (H' and e' are well-typed)

(d) S'; G_s, G_f' |-_q p'     (the resulting substitution is well-typed)

(e) P, P_s |- (H, e)[p1] -n->tr (H', e')[p1']    (running on p1 is basically
                                                  same trace and same template
                                                  with different substitution
                                                  for secrets)

(f) S'; G_s, G_f' |-_q p1'     (the second resulting substitution
                                is also well-typed)

-------------------------
Proof: by induction on e
-------------------------

Case e = let x:t1 = e1 in e2
------------------------------

Subcase 1: e1 is not a value:
------------------------------

Inverting [E-Let] with (1), we get:

(5) P; S; G_s, G_f |- e1 : t1

Applying I.H. on (5), (2), and (3):

(6) P, P_s |- (H, e1)[p] -m->tr (H1, e11)

(7) (H', e1')[p'] = (H1, e11)

(8) P; S'; G_s, G_f' |- (H', e1') : t1

(9) S; G_s, G_f' |-_q p'

(10) P, P_s |- (H, e1)[p1] -n->tr (H', e1')[p1']

(11) S; G_s, G_f' |-_q p1'

Choose m = m, n = n.
Choose S' = S', G_f' = G_f', H' = H', e' = let x = e1' in e2, p' = p', p1' = p1'

And, choose the evaluation context LE = let x = <> in e2

Using [Step], the evaluation context stepping rule, we get (a):

(12) P, P_s |- (H, let x = e1 in e2)[p] -m->tr (H1, let x = e11 in e2[p])

Consider (H', let x = e1' in e2)[p'], since p' \supserseteq p, e2[p'] = e2[p],
hence we get (b).

(c) follows from applying [T-Let] with (8), and weakening lemmas.

(d) follows from (9).

(e) follows similar to deriving (a) above.

(f) follow from (11).


Subcase 2: e1 is a value
-------------------------

(5) e[p] = let x = e1[p] in e2[p]

Using [E-Let]:

(6) P, P_s |- (H[p], let x = e1[p] in e2[p]) --> (H[p], e2[p][e1[p]/x])

Choosing m = 1, H1 = H[p], e1 = e2[p][e1[p]/x], tr = ., we get (a).

Choosing H' = H, e' = e2[e1/x], p' = p, we get (b).

Inverting [T-Let] with (1), we get:

(7) P; S; G_s, G_f |- e1 : t'

(6) P; S; G_s, G_f, x:t' |- e2 : t

Since e1 is a value, Using Lemma [Substitution-lemma-for-Low*], we get:

(8) P; S; G_s, G_f |- e2[e1/x] : t

Since H' = H, we get (c).

Since p' = p, we get (d) from (2).

(e) follows similar to above evaluation for p with n = 1
(p really did not play any role in the evaluation above).

(f) follows from (3), since p1' = p1.


Case e = let x = f e1 in e2
-----------------------------

Subcase 1: e1 is not a value:
-----------------------------

Inverting [T-App] with (1):

(5) P; S; G_s, G_f |- e1 : t1

Applying I.H. on (5), (2), and (3):

(6) P, P_s |- (H, e1)[p] -m->tr (H1, e11)

(7) (H', e1')[p'] = (H1, e11)

(8) P; S'; G_s, G_f' |- (H', e1') : t1

(9) S; G_s, G_f' |-_q p'

(10) P, P_s |- (H, e1)[p1] -n->tr (H', e1')[p1']

(11) S; G_s, G_f' |-_q p1'

Choose m = m, n = n.
Choose S' = S', G_f' = G_f', H' = H', e' = let x = f e1' in e2, p' = p', p1' = p1'

And, choose the evaluation context LE = let x = f <> in e2

Using [Step], the evaluation context stepping rule, we get (a):

(12) P, P_s |- (H, let x = f e1 in e2)[p] -m->tr (H1, let x = f e11 in e2[p])

Consider (H', let x = f e1' in e2)[p'], since p' \supserseteq p, e2[p'] = e2[p],
hence we get (b).

(c) follows from applying [T-App] with (8), and weakening lemmas.

(d) follows from (9).

(e) follows similar to deriving (a) above.

(f) follow from (11).


Subcase 2: e1 is a value
-------------------------

Subsubcase 1: f \in P:
----------------------

Inverting [T-App] with (1):

(5) P(f) = fun (x:t1) -> e:t

(6) P; S; G_s, G_f |- e1 : t1

(7) P; S; G_s, G_f, x:t |- e2 : t2

Using [App], we get:

(8) P, P_s |- (H, let x = f e1 in e2)[p] -0->. (H[p], let x = e[e1[p]/x] in e2[p])

which gives us (a).

Since e1 has no free secret variables,

(9) let x = e[e1[p]/x] in e2[p] = (let x = e[e1/x] in e2)[p]

Choose H' = H, e' = let x = e[e1/x] in e2, and p' = p to get (b).

Choose S' = S, G_f' = G_f

Using substitution lemma with weakening:

(10) P; S; G_s, G_f |- e[e1/x] : t (recall e1 is a value)

Applying [E-Let] with (10) and (7):

(11) P; S; G_s, G_f |- let x = e[e1/x] in e2 : t2

And since H' = H, we get (c).

(d) follows since G_f' = G_f and p' = p.

(e) follows since m = n = 0, and tr = . .

(f) follows since p1' = p1.


Subsubcase 2: f:t1 -> t2 \in G_f
----------------------------------

Inverting (1):

(5) P; S; G_s, G_f |- H

(6) P; S; G_s, G_f |- e1 : t1

(7) P; S; G_s, G_f, x:t2 |- e2 : t

Using Lemma [Substituted-value-typing] with (2) and (6):

(8) P; S[q]; G_s[q] |- e1[p] : t1[q]

Using Lemma [Substituted-value-typing] with (3) and (6):

(9) P; S[q]; G_s[q] |- e1[p1] : t1[q]

Using Lemma [Substituted-heap-typing] with (2) and (5):

(10) P; S[q]; G_s[q] |- H[p]

Using Lemma [Substituted-heap-typing] with (3) and (5):

(11) P; S[q]; G_s[q] |- H[p1]

Using Lemma [Equivalence-modulo-secrets-heaps] with (5), (2), and (3):

(12) S |- H[p] =eq= H[p1]

Using Lemma [Equivalence-modulo-secrets-values] with (6), (2), and (3):

(13) (e1[p], e2[p]) \in V_t1

Using assumption about secret manipuating functions with
H1 = H[p], H2 = H[p1], v1 = e1[p], v2 = e1[p1]
and (12), (8), (9), (13), (10), and (11):

(14) P, P_s |- (H[p], f e1[p]) -m->tr (H1', v1')

(15) P, P_s |- (H[p1], f e1[p1]) -n->tr (H2', v2')

(16) S' |- H1' =eq= H2'

(17) P; S'[q]; G_s |- v1' : t2[q]

(18) P; S'[q]; G_s |- v2' : t2[q]

(19) (v1', v2') \in V_t2

(20) P; S'[q]; G_s |- H1'

(21) P; S'[q]; G_s |- H2'

Using Lemma [Gen-function-lemma-for-values] with (assume), (17), (18), (19):

(22) Gen v1' v2' t2 = v, p', p1', G

(22a) P; S; G_s, G |- v : t2

(23) S; G_s, G |-_q p'

(24) S; G_s, G |-_q p1'

Using Lemma [Gen-function-lemma-for-heaps] with (assume), (20), (21), (16):

(25) Gen H1 H2 S' = H', p_h, p_h1, G_h

(25a) P; S; G_s, G_h |- H'

(26) S; G_s, G_h |-_q p_h

(27) S; G_s, G_h |-_q p_h1


For theorem consequences, choose:

m = m, n = n, S' = S', G_f' = G_f \union G \union G_h, H' = H',
e' = let x = v in e2, p' = p \union p' \union p_h, p1' = p1 \union p1' \union p_h1

(a), (b), and (e) follow from [Context] rule of semantics and (22) and (25) above.

(d) and (f) follow from (23), (24), (26), and (27).

(c) follows from (22a) and (7) with [T-Let] and heap typing from (25a).


Case e = let _ = writebuf e1 e2 e3 in e
----------------------------------------

Subcase 1: One of e1, e2, e3 is not a value: These are context cases proven similar
to the context cases of let and application above.

Subcase 2: e1, e2, e3 are values
---------------------------------

Inverting (1):

(4) P; S; G_s, G_f |- H

(5) P; S; G_s, G_f |- e1 : buf t1

(6) P; S; G_s, G_f |- e2 : int

(7) P; S; G_s, G_f |- e3 : t1

Using Lemma [Substituted-value-typing] with (5) and (2):

(8) P; S[q]; G_s[q] |- e1[p] : buf t1[q]

Using Lemma [Substituted-value-typing] with (6) and (2):

(9) P; S[q]; G_s[q] |- e2[p] : int

Using Lemma [Substituted-value-typing] with (7) and (2):

(10) P; S[q]; G_s[q] |- e3[p] : t1[q]

Using Lemma [Value-inversion], we get:

(11) e1[p] = (b, n), actually e1 = (b, n)

(12) e2[p] = n', actually e2 = n'    (Lemma [Non-secret-value-substitution])

Using [E-Write]:

(13) P, P_s |- (H, let _ = write e1 e2 e3 in e)[p] ->write(b, n + n')
               (H[p][(b, n + n') |-> e3[p]], e[p])

(the step can get stuck if the index is out of bound or the frame has been popped,
in which case we just step to stuck value)

For (a) and (b) choose m = 1, tr = write(b, n + n'), H' = H[(b, n + n') |-> e3],
e' = e, p' = p, S' = S, G_f' = G_f

(c) simply follows from the inversion of typing judgment (1),
and new heap is also well-typed.

(d) follows from (2).

We can do similar analysis with p1 as above and get that:

(13) P, P_s |- (H, let _ = write e1 e2 e3 in e)[p1] ->write(b, n + n')
               (H[p1][(b, n + n') |-> e3[p1]], e[p1])

Note that e1 and e2 are independent of p or p1, as noted in (11) and (12) above.

So, choose n = 1 after which (e) follows. (f) follows from (3).


Case e = let x = newbuf n e1 in e2
-----------------------------------

Subcase 1: Either e1 or e2 is not a value: These are context cases proven
similar to the let and application cases above.

Subcase 2: Both e1 and e2 are values
-------------------------------------

Inverting (1):

(4) P; S; G_s, G_f |- H

(5) P; S; G_s, G_f |- e1 : t1

(6) P; S; G_s, G_f, x:buf t |- e2 : t

Using [E-NewBuf]:

(7) P, P_s |- (H[p], let x = n e1[p] in e2[p]) ->write (b, 0) ... (b, n - 1)
              (H[p][b |-> e1[p]^n], e2[p][(b, 0)/x])

Choose m = 1, n = 1, S' = S[b |-> buf t1], G_f' = G_f, H' = H[b |-> e1^n],
e' = e2[(b, 0)/x], p' = p, p1' = p1, tr = write (b, 0) ... (b, n - 1)

(a) and (b) follow from (7).

(c) follows since H' is well-typed in S', and e' is well-typed after applying
Lemma [Substitution-lemma-for-Low*].

(d) follows from (2).

(e) follows with p1' = p1, and the trace is independent of the substitution p and p'

(f) follows from (3).


Case e = let x = readbuf e1 e2 in e
------------------------------------

Subcase 1: One of e1 or e2 is not a value: These are context cases proven
similar to the let and application cases above.

Subcase 2: Both e1 and e2 are values
-------------------------------------

Inverting (1):

(4) P; S; G_s, G_f |- e1 : buf t1

(5) P; S; G_s, G_f |- e2 : int

(6) P; S; G_s, G_f, x:t1 |- e : t

Using Lemma [Substituted-value-typing] with (4) and (2), (5) and (2), and
Lemma [Substituted-heap-typing] with (6) and (2):

(7) P; S[q]; G_s[q] |- e1[p] : buf t1[q]

(8) P; S[q]; G_s[q] |- e2[p] : int

(9) P; S[q]; G_s[q] |- H[p]

Using Lemma [Non-secret-value-substitution] with (7) and (8):

(10) e1 = (b, n)

(11) e2 = n'

Using [E-Read]:

(12) P, P_s |- (H, let x = readbuf e1 e2 in e)[p] ->read (b, n + n')
               (H[p], e[p][H[p](b, n + n')/x])

(the step can get stuck if the index is out of bound or the frame has been popped,
in which case we just step to stuck value)

Choose m = 1, n = 1.
Choose S' = S, G_f' = G_f, H' = H, e' = e[H(b, n + n')/x], p' = p, p1' = p1

(a) and (b) follow from (12).

(c) follows since H' = H, and e' is well-typed
    after applying Lemma [Substitution-lemma-for-Low*].

(d) follows from (2).

(e) follows similar to (12).

(f) follows from (3).
\end{verbatim}
\end{small}









\newpage
\fi


\bibliography{biblio,fstar,tls}

\iflong
\else
\balance
\fi

\end{document}